\pdfoutput=1
\RequirePackage{ifpdf}
\ifpdf % We~are running pdfTeX in pdf mode
\documentclass[pdftex]{sigma}
\else
\documentclass{sigma}
\fi

\numberwithin{equation}{section}

\newtheorem{Theorem}{Theorem}[section]
\newtheorem*{Theorem*}{Theorem}
\newtheorem{Corollary}[Theorem]{Corollary}
\newtheorem{Lemma}[Theorem]{Lemma}
\newtheorem{Proposition}[Theorem]{Proposition}

\theoremstyle{definition}
\newtheorem{Definition}[Theorem]{Definition}

\newtheorem{Remark}[Theorem]{Remark}

\usepackage{tikz}
\usetikzlibrary{decorations.text}
\usetikzlibrary{decorations.pathreplacing}
\usetikzlibrary{fit}
\usetikzlibrary{arrows.meta}
\usetikzlibrary{calc}
\usepackage{mathtools}

\usepackage{dsfont}

\newcommand{\flac}[2]{#1}

\definecolor{BW--Black}{RGB}{0,0,0}
\definecolor{BW--LightOrange}{RGB}{230,159,0}
\definecolor{BW--Cyan}{RGB}{86,180,233}
\definecolor{BW--Green}{RGB}{0,158,115}
\definecolor{BW--Yellow}{RGB}{240,228,66}
\definecolor{BW--Blue}{RGB}{0,114,178}
\definecolor{BW--DarkOrange}{RGB}{213,94,0}
\definecolor{BW--Pink}{RGB}{204,121,167}

\begin{document}

\allowdisplaybreaks

\newcommand{\arXivNumber}{2407.21015}

\renewcommand{\PaperNumber}{095}

\FirstPageHeading

\ShortArticleName{Comparative Analysis of the Type D ASEP}

\ArticleName{Comparative Analyses of the Type D ASEP:\\ Stochastic Fusion and Crystal Bases}

\Author{Erik BRODSKY~$^{\rm a}$, Eva ENGEL~$^{\rm b}$, Connor PANISH~$^{\rm c}$ and Lillian STOLBERG~$^{\rm d}$}

\AuthorNameForHeading{E.~Brodsky, E.~Engel, C.~Panish and L.~Stolberg}

\Address{$^{\rm a)}$~Department of Mathematics, Cornell University, Ithaca, NY, USA}
\EmailD{\mail{eab384@cornell.edu}}

\Address{$^{\rm b)}$~Department of Mathematics, Princeton University, Princeton, NJ, USA}
\EmailD{\mail{eva.engel@princeton.edu}}

\Address{$^{\rm c)}$~Department of Statistics, Cornell University, Ithaca, NY, USA}
\EmailD{\mail{cp732@cornell.edu}}

\Address{$^{\rm d)}$~Department of Mathematics, University of Michigan, Ann Arbor, MI, USA}
\EmailD{\mail{lstolber@umich.edu}}

\ArticleDates{Received February 11, 2025, in final form September 02, 2026; Published online October 02, 2026}

\Abstract{The Type~D asymmetric simple exclusion process (ASEP), a particle system involving two classes of particles, was studied from both probabilistic and algebraic perspectives by Kuan et al.~(2022). From a~probabilistic perspective, we perform stochastic fusion on the Type~D ASEP introduced by Kuan (2019) and analyze the outcome on the generator matrix, limits of drift speed, stationary distributions, and Markov self-duality. From an algebraic perspective, we construct a~fused Type D ASEP system from a~Casimir element of~$\mathcal{U}_q(\mathfrak{so}_6)$, using crystal bases to analyze and manipulate various representations of~$\mathcal{U}_q(\mathfrak{so}_6)$. We know that the same generators for the normal ASEP are produced by stochastic fusion and by a~suitable ground state transformation on a~central element of~$\mathcal{U}_q(\mathfrak{sl}_2)$ in a~symmetric tensor product representation. However, in the case of the Type~D ASEP, we find that the probabilistic and algebraic generator matrices are not the same and thus represent different processes. We conclude that the relationship between stochastic fusion and a~ground state transformation (specifically, what we term the Type~A ground state transformation) established by Kuan (2019) does not generalize to all finite-dimensional simple Lie algebras.}


\Classification{82C22; 60J27; 17B37; 70G65}

\section{Introduction} \label{section--Intro}

The interacting particle system called the Type D ASEP~\cite{kuan2020interactingparticlesystemstype} is a~generalization of Spitzer's ASEP~\cite{SPITZER1970246} that involves two classes of particles jumping on a~lattice. Two particles of different classes can exist at a~site but two particles of the same class cannot, and the jump rate of an individual particle depends on its class and configuration in relation to other particles. We offer examples of the process in Figure~\ref{figure--fig1}.

\begin{figure}[!ht]%\label{figure--fig1}
\centering
\begin{tikzpicture}
 % Define the coordinates
 \coordinate (A) at (-3.5,0);
 \coordinate (B) at (-0.5,0);
 \coordinate (C) at (0.5,0);
 \coordinate (D) at (3.5,0);

 % Draw the solid lines
 \draw (A) -- (B);
 \draw (C) -- (D);

 % Draw 5 dashed lines for each line
 \foreach \i in {-3.5, -2.5, -1.5, -0.5, 0.5, 1.5, 2.5, 3.5} {
 \draw (\i, -0.25) -- (\i, 0.25);
 }

 % Draw arrow between the two lattices
 \draw[->, thick] (-0.25,0) -- (0.25,0);

 % Draw an empty circle at x = -4.5
 \draw (-1.5, 0.7) circle (0.3);
 \node at (-1.5, 0.7) {1};

 % Draw an empty circle at x = -4.5
 \draw (-3.5, 0.7) circle (0.3);
 \node at (-3.5, 0.7) {2};

 % Draw an empty circle at x = -2.5
 \draw (-2.5, 0.7) circle (0.3);
 \node at (-2.5, 0.7) {1};

 % Draw an empty circle at x = 3.5
 \draw (3.5, 0.7) circle (0.3);
 \node at (3.5, 0.7) {1};

 % Draw an empty circle at x = 3.5
 \draw (0.5, 0.7) circle (0.3);
 \node at (0.5, 0.7) {2};

 % Draw an empty circle at x = 3.5
 \draw (1.5, 0.7) circle (0.3);
 \node at (1.5, 0.7) {1};

\end{tikzpicture}
\\
\begin{tikzpicture}
 % Define the coordinates
 \coordinate (A) at (-3.5,0);
 \coordinate (B) at (-0.5,0);
 \coordinate (C) at (0.5,0);
 \coordinate (D) at (3.5,0);

 % Draw the solid lines
 \draw (A) -- (B);
 \draw (C) -- (D);

 % Draw 5 dashed lines for each line
 \foreach \i in {-3.5, -2.5, -1.5, -0.5, 0.5, 1.5, 2.5, 3.5} {
 \draw (\i, -0.25) -- (\i, 0.25);
 }

 % Draw arrow between the two lattices
 \draw[->, thick] (-0.25,0) -- (0.25,0);

 % Draw an empty circle at x = -4.5
 \draw (-1.5, 0.7) circle (0.3);
 \node at (-1.5, 0.7) {1};

 % Draw an empty circle at x = -4.5
 \draw (-3.5, 0.7) circle (0.3);
 \node at (-3.5, 0.7) {2};

 % Draw an empty circle at x = -2.5
 \draw (-2.5, 0.7) circle (0.3);
 \node at (-2.5, 0.7) {1};

 % Draw an empty circle at x = 3.5
 \draw (2.5, 0.7) circle (0.3);
 \node at (2.5, 0.7) {1};

 % Draw an empty circle at x = 3.5
 \draw (1.5, 1.5) circle (0.3);
 \node at (1.5, 1.5) {2};

 % Draw an empty circle at x = 3.5
 \draw (1.5, 0.7) circle (0.3);
 \node at (1.5, 0.7) {1};

\end{tikzpicture}
\\
\begin{tikzpicture}
 % Define the coordinates
 \coordinate (A) at (-3.5,0);
 \coordinate (B) at (-0.5,0);
 \coordinate (C) at (0.5,0);
 \coordinate (D) at (3.5,0);

 % Draw the solid lines
 \draw (A) -- (B);
 \draw (C) -- (D);

 % Draw 5 dashed lines for each line
 \foreach \i in {-3.5, -2.5, -1.5, -0.5, 0.5, 1.5, 2.5, 3.5} {
 \draw (\i, -0.25) -- (\i, 0.25);
 }

 % Draw arrow between the two lattices
 \draw[->, thick] (-0.25,0) -- (0.25,0);

 % Draw an empty circle at x = -4.5
 \draw (-1.5, 0.7) circle (0.3);
 \node at (-1.5, 0.7) {1};

 % Draw an empty circle at x = -4.5
 \draw (-1.5, 1.5) circle (0.3);
 \node at (-1.5, 1.5) {2};

 % Draw an empty circle at x = -2.5
 \draw (-2.5, 0.7) circle (0.3);
 \node at (-2.5, 0.7) {1};

 % Draw an empty circle at x = -4.5
 \draw (1.5, 0.7) circle (0.3);
 \node at (1.5, 0.7) {1};

 % Draw an empty circle at x = -4.5
 \draw (1.5, 1.5) circle (0.3);
 \node at (1.5, 1.5) {2};

 % Draw an empty circle at x = -2.5
 \draw (2.5, 0.7) circle (0.3);
 \node at (2.5, 0.7) {1};

\end{tikzpicture}
\caption{Above are examples of the Type D ASEP process on four lattice sites. In the top figure, the rightmost particle of class 1 has jumped one site to the right. In the middle figure, the particle of class~2 has jumped one site to the right, stacking with a~particle of class~1. In the bottom figure, the particle of class~2 has jumped one site to the left, also stacking with a~particle of class~1.} \label{figure--fig1}
\end{figure}

We question whether constructing a~generator according to the probabilistic method of stochastic fusion, as defined in~\cite{kuan2019stochastic}, on the Type D ASEP will result in the same generator as the algebraic method, as studied in~\cite{Carinci2016,Carinci2014,CrampeRagoucyVanicatReview2014,Kuan_2016,Kuan_2018,Kuan_2017,KuanTwoDualities2021}, of applying a~ground state transformation on a~Casimir element in the second tensor power of an irreducible representation of~$\mathcal{U}_q(\mathfrak{so}_6)$, the case of Type~D Lie algebras. Ultimately, we determine that the generators found probabilistically and algebraically are not the same. In fact, the probabilistic generator allows for nine states at each lattice site while the algebraic generator allows for fourteen. We therefore explore various characteristics of the different processes generated using these separate approaches. The probabilistic characteristics include limits of component generator matrices, stationary distributions, spectral gaps, and Markov self-duality. In terms of algebraic characteristics, we study the algebraic structure of various representations needed to construct the generator, creating crystal bases and Young tableaux to decompose $\mathcal{U}_q(\mathfrak{so}_6)$-modules into irreducible representations as well as weight spaces. We then reflect on the impacts that these decompositions have on the eigenvalues and overall structure of the constructed Markov generator.

For background, \cite{10.1214/aop/1176996084} is the first paper considering two-class ASEP. Later research by~\cite{Belitsky_2015,belitskyschutz15,belitsky2016selfdualityshockdynamicsncomponent,Carinci2016,Carinci2014,Kuan_2016, Kuan_2018, Kuan_2017, kuan2020interactingparticlesystemstype, rohr2023type} shows that generators of multi-class interacting particle systems can be constructed using a~central element from various Drinfeld--Jimbo quantum subgroups of~$\mathcal{U}_q(\mathfrak{gl}_n)$ \cite{Drinfeld85, Jimbo85}. These results inform our algebraic approach to constructing the generator matrix for Type~D Lie algebras. In terms of probability background, \cite{kuan2019stochastic} develops the probabilistic procedure of stochastic fusion, which is shown to produce the same process as central elements on normal ASEP. This result suggests that the generators are equivalent for the algebraic and probabilistic approaches, a~finding extended by~\cite{rohr2023type}, whose methods inform our derivation of the generator for stochastic fusion on the Type D ASEP. In contrast, we find that the probabilistic and algebraic generators for the fused Type~D ASEP are different, as we state in the theorem below.\looseness=-1

\begin{Theorem} \label{thm-bigalgebra}
 Regardless of which ground state transformation is applied to the $\mathrm{Sym}^2_q(\mathfrak{so}_6)$ representation of the Casimir element of~{\rm\cite{kuan2020interactingparticlesystemstype}}, the resulting Markov generator does not match the probabilistically generated matrix of the stochastically fused Type D ASEP on four lattice sites.\looseness=-1
\end{Theorem}

To approach this problem, we apply observations about crystal bases to probability. The Type A~and Type D crystal graphs differ in that the Type A~crystal graph is simply connected, while the Type D crystal graph is not. This suggests that the algebraic constructions of normal ASEP and Type D ASEP will differ, which leads us to believe that the equivalence between the algebraic and probabilistic constructions of normal ASEP will not extend to the Type D ASEP. Of course, proving this fact rigorously takes more work, and is the content of this paper. To our knowledge, this is the first use of crystal bases in a~stochastic process.

The paper proceeds by the following outline: notation is introduced in the remainder of Section~\ref{section--Intro}, Theorem~\ref{thm-bigalgebra} is proved in Section~\ref{section--ProbResults}, and additional probabilistic results for the stochastically fused particle system are given in Section~\ref{section--ProbProofs}. We give explicit constructions of the two particle systems, and methods of computing them, in Appendix~\ref{section--Appendix}. In particular, the explicit algebraic generator can be found in Appendix~\ref{subsection-MarkovGen}.

\subsection{Probability notation} \label{subsection--ProbNotation}
\subsubsection{Type D ASEP}
\label{TypeDASEPIntro}
The Type D ASEP is a~Markov process in which there are two classes of particles jumping on a~lattice. Two particles of the same class cannot share a~lattice site, but two particles of different classes can. The Type~D ASEP is determined by parameters $(q,n,\delta = 0)$. The variable~$n$ measures drift speed, with larger values of $n$ indicating more rapid drifts in the Type~D ASEP. The variable $q$ represents drift direction, with rightward drift for $0<q<1$ and leftward drift for $q > 1$. Note that $q$ cannot equal $1$. The variable $\delta$ represents the interaction between the two particle classes. The jump rate of a~particle in the Type~D ASEP is determined by the particle's class and configuration on the lattice.

We use $\Gamma$ to denote lattice sites of the Type D ASEP. For each $\Gamma$-site, we use $0$ to denote an empty site (contains no particles), $1$ to denote a~particle of class 1 at a~site, $2$ to denote a~particle of class 2, and $3$ to denote a~particle of class 1 and a~particle of class 2. Furthermore, we use parentheses to denote a~$\Gamma$-state in the Type D ASEP. For example, we let $(3,1,0,2)$ denote the four $\Gamma$-state particle configuration in the top line of Figure~\ref{tikz--asepDiagram}.

We now consider the Type D ASEP over two $\Gamma$-lattice sites with $\delta = 0$. See~\cite{rohr2023type} for a~rigorous definition and~\cite{kuan2020interactingparticlesystemstype} for a~picture of the generator of Type D ASEP on two $\Gamma$-lattice sites. Note that there are sixteen possible $(a,b)$ states since $0 \leq a$, $b \leq 3$. We use \smash{$L_p^{(2)}$} to denote the generator for Type D ASEP on two $\Gamma$-sites, which is a~direct sum of
\[
L_1 =
\left[\begin{smallmatrix}
* & \frac{(- q^{1 - n} + q^{n - 1})^{2}}{q^{2}} & - q^{2 n - 4} + q^{2 n - 2} + \frac{2}{q^{2}} & - q^{2 n - 4} + q^{2 n - 2} + \frac{2}{q^{2}}
\\q^{2 n} - q^{2 n - 2} + 2 & * & - q^{2 - 2 n} + 2 + q^{- 2 n} & (- q^{1 - n} + q^{n - 1})^{2}
\\
q^{2} (- q^{1 - n} + q^{n - 1})^{2} & 2 q^{2} + q^{2 - 2 n} - q^{4 - 2 n} & * & 2 q^{2} + q^{2 - 2 n} - q^{4 - 2 n}
\\q^{2 n} - q^{2 n - 2} + 2 & (- q^{1 - n} + q^{n - 1})^{2} & - q^{2 - 2 n} + 2 + q^{- 2 n} & *
\end{smallmatrix}\right]
\]
corresponding to the communicating classes $\{(3,0),(2,1),(0,3),(1,2)\}$, four 2 by 2 blocks
\[L_2 = \begin{bmatrix} * & \frac{q^{1 - 2 n} + q^{2 n - 1}}{q}\\q \left(q^{1 - 2 n} + q^{2 n - 1}\right) & * \end{bmatrix}\]
corresponding to the communicating classes \[\{(1,0),(0,1)\},\{(2,0),(0,2)\},\{(3,1),(1,3)\},\{(3,2),(2,3)\},\] and four $1$ by $1$ blocks of zeros corresponding to the communicating classes \[\{(0,0)\},\{(1,1)\},\{(2,2)\},\{(3,3)\}.\] We use $*$ in row $i$ to denote the scalar such that row $i$ sums to 0.

Finally we express
\[L_p^{(2)} = L_1 \oplus \bigoplus_{i=1}^4 L_2 \oplus \bigoplus_{i=1}^4 [0]\]
with respect to the ordered basis
\begin{align*}
 \Omega ={}& ((3,0),(2,1),(0,3),(1,2),(1,0),(0,1),(2,0),(0,2),(3,1),(1,3),(3,2),(2,3),(0,0),\\ &(1,1),(2,2),(3,3) ).
\end{align*}
Having provided an example over two $\Gamma$-lattice sites, we expand the definition of Type D ASEP to $K$ $\Gamma$-lattice sites where the generator matrix is given by \smash{$L = L^{1,2}+L^{2,3}+\dots+L^{K-1,K}$}.
We let $L^{x,x+1}$ denote the matrix $L_p^{(2)}$ acting on lattice sites $x$ and $x+1$. For this paper, we usually restrict to the case of Type D ASEP on four lattice sites for ease of computation. Therefore, we define \smash{$L_p \coloneqq L^{1,2}+L^{2,3}+L^{3,4}$}
to represent the generator of Type D ASEP on four $\Gamma$-sites. The generator $L_p$ is a~256 by 256 matrix, since $256 = 4^4$. Furthermore, we define $L_m$ to be the generator where only the two middle particles interact, i.e., $L_m \coloneqq L^{2,3}$.

\subsubsection{Stochastic fusion} \label{stoch-fusion}
Taking a~lattice with an even number of sites, we define stochastic fusion as the merging of each set of two consecutive lattice sites such that the particles in each original site of the set are stacked at the new site. Note that two particles of the same class can exist at one new site, and that the total number of sites involved in the process is decreased by one half through stochastic fusion. See~\cite{kuan2019stochastic} for a~thorough treatment of stochastic fusion.
We proceed to define notation for the fused Type~D ASEP, using $\gamma$ to denote lattice sites and states that have undergone the process of fusion, and $\Gamma$ to denote lattice sites and states that have not. For $\gamma$-lattice sites, we use $11$ to denote two particles of class 1 at a~site, $22$ to denote two particles of class 2, $3$ to denote a~particle of class 1 and a~particle of class 2, $33$ to denote two particles of class 1 and two particles of class 2, $31$ to denote two particles of class 1 and one particle of class 2, and~$32$ to denote one particle of class 1 and two particles of class 2. We will use angle brackets to denote $\gamma$-states. See Figure~\ref{fig:visualLabel} for a~visualization of this notation. As another example of $\gamma$-state notation, note that $\langle 31,2 \rangle$ is the $\gamma$-state of $(3,1,0,2)$ and is depicted in the bottom line of Figure~\ref{tikz--asepDiagram}.

\begin{figure}[!ht]
 \centering
 \includegraphics[scale=0.35]{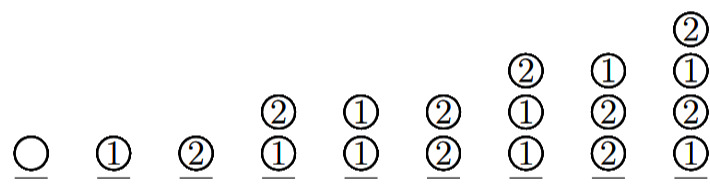}
 \caption{The fusion states for one $\gamma$-site from left to right are $\langle 0 \rangle, \langle 1 \rangle, \langle 2 \rangle,\langle 3 \rangle, \langle 11 \rangle,\langle 22 \rangle,\langle 31 \rangle,\langle 32 \rangle,\langle 33 \rangle$.}
 \label{fig:visualLabel}
\end{figure}

Let $S$ be the state space for the four-site Type D ASEP, in which there are $256 = 4^4$ possible states. Formally,
\begin{gather*}
 S = \{(a, b, c, d) \mid a, b, c, d \in \{0,1,2,3\} \}.
\end{gather*}
Fusing our lattice, we reach the state space $\hat{S}$ of two $\gamma$-lattice sites, in which there are $81 = 9^2$ possible states. More precisely,
\begin{gather*}
 \hat{S} = \{\langle x, y \rangle \mid x, y \in \{0,1,2,3,11,22,31,32,33\} \}.
\end{gather*}

We define fission as the process of splitting one $\gamma$-lattice site into two $\Gamma$-lattice sites (or $n$ $\gamma$-lattice sites into $n$ pairs of $\Gamma$-lattice sites, each pair computed independently). In words, the fission map is a~random map with a~probability measure over potential splits. Its formula is given in Definition~\ref{LambdaTwo}.

Note that fission is not a~deterministic process. For instance, fission can take $\langle 31,2 \rangle \to (3,1,0,2)$ and $\langle 31,2 \rangle \to (1,3,0,2)$. However, $\langle 31,2 \rangle \to (11,2,0,2)$ is not a~valid fission operation because there are two particles of the same class at the first output $\Gamma$-lattice site. Fission outputs only valid configurations on $\Gamma$-lattice sites, so two particles of the same class cannot share an output $\Gamma$-site. Additionally, $\langle 31,2 \rangle \to (1,2,1,2)$ is not a~valid fission operation because fission is applied to each $\gamma$-lattice site independently, so particles from different $\gamma$-lattice sites cannot mix. The third $\Gamma$-lattice site in $(1,2,1,2)$ violates this rule, since $1$ comes from the first $\gamma$-lattice site.

\begin{figure}[!ht]
 \centering
 \scalebox{1}{
 \includegraphics[scale=0.4]{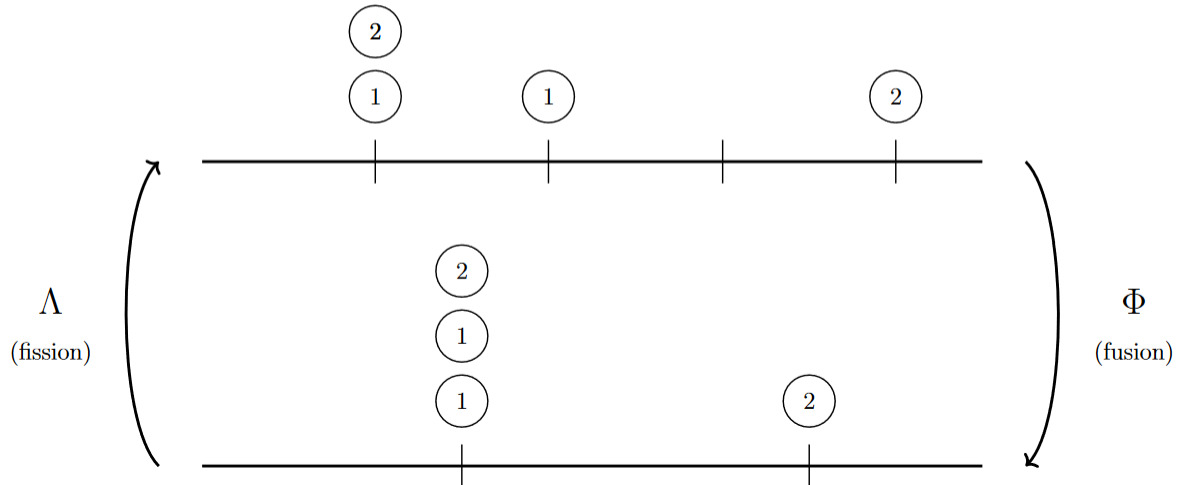}
 }
 \caption{Depiction of fusion mapping $(3,1,0,2) \rightarrow \langle 31,2 \rangle$ and fission mapping $\langle 31,2 \rangle \rightarrow (3,1,0,2)$.}
 \label{tikz--asepDiagram}
\end{figure}

Our generator, fusion, and fission matrices depend on the ordering of states, so we order states in the Type D ASEP over two $\Gamma$-lattice sites as $\Omega$ and we order states in the Type D ASEP over one $\gamma$-lattice site as
\[\chi = ( \langle 0 \rangle,\langle 1 \rangle, \langle 2 \rangle, \langle 11 \rangle, \langle 12 \rangle, \langle 22 \rangle,\langle 31 \rangle,\langle 32 \rangle, \langle 33 \rangle ). \]
 A~visualization of this ordering and the possible states over one $\gamma$-lattice site can be seen in Figure~\ref{fig:visualLabel}. See Appendix~\ref{PermutationsAppendix} for the ordering of $\Gamma-$ and $\gamma$-states over four $\Gamma$-lattice sites and two $\gamma$-lattice sites, respectively.

We now formally define the stochastic fusion map.

\begin{Definition}
\label{StochasticDef}
We define $\phi\colon\Omega \rightarrow \chi$ as the stochastic fusion map for the Type D ASEP on two $\Gamma$-lattice sites, where $x_i \in \{0,1,2,3\}$ describes the particles at $\Gamma$-lattice site $i$
\[
\phi(x_1,x_2)= \begin{cases}
 \langle 0 \rangle, & x_1,x_2 = 0 ,\\
 \langle x_1 \rangle, & x_1 \neq 0, \ x_2 = 0 ,\\
 \langle x_2 \rangle ,& x_1 = 0, \ x_2 \neq 0 , \\
 \langle x_1x_2 \rangle, & x_1,x_2 \neq 0, \ x_1 \geq x_2, \\
 \langle x_2x_1 \rangle ,& x_1,x_2 \neq 0, \ x_1 < x_2.
 \end{cases}
\]
\end{Definition}

Notice that this map corresponds to stacking particles where order does not matter. To perform stochastic fusion on four $\Gamma$-lattice sites of the Type D ASEP, the stochastic fusion map is $\phi \otimes \phi$, and so on for additional sites of stochastic fusion.

\subsubsection{Transition and generator matrices} \label{transition-gen}

We now turn to the transition matrices of the Type D ASEP. Let \smash{$P^{(K)}_t$} be the transition matrix and \smash{$L^{(K)}_p$} be the generator matrix for the Type D ASEP on $K$ $\Gamma$-lattice sites, where $K$ is an even integer. Let
\[
L^{(K)}_m = \sum_{i=1}^{\frac{K}{2}-1} L^{2i,2i+1},
\]
 a~restriction on which of the $K$ $\Gamma$-lattice sites can interact. Fusing down to $K/2$ $\gamma$-lattice sites, \smash{$Q^{(K)}_t$} is the fused transition matrix and \smash{$L^{(K)}_Q$} is the generator matrix for the Type D ASEP on $K/2$ $\gamma$-lattice sites. Summarizing this process, we let \smash{$\Lambda^{(K)}\colon \hat{S} \rightarrow S$} be the matrix of the fission map from $K/2$ $\gamma$-lattice sites to $K$ $\Gamma$-sites. \smash{$\Lambda^{(K)}$} is a~Markov kernel since fission is not a~deterministic process. When there is more than one possible output state, the result of fission is determined by the weights $G^2$ described in Section~\ref{fusion-gen}. Moreover, we define \smash{$\Phi^{(K)}\colon S \rightarrow \hat{S}$} to be the matrix of the stochastic fusion map from $K$ $\Gamma$-lattice sites to $K/2$ $\gamma$-lattice sites. We will define \smash{$\Lambda^{(K)}$} and~\smash{$\Phi^{(K)}$} more explicitly when constructing the generator matrix for stochastically fused Type~D ASEP in Section~\ref{fusion-gen}. In the future, we will also drop the superscript $(K)$ when~${K = 4}$, which is the first nontrivial example of stochastic fusion. For example, \smash{$L_p := L_p^{(4)}$}, \smash{$\Lambda :=\Lambda^{(4)}$}, and~\smash{$\Phi :=\Phi^{(4)}$}.

\subsection{Algebraic notation} \label{algebra_defns}
We must first define $\mathcal{U}_q(\mathfrak{so}_6)$ as well as its fundamental representation in order to construct more complex $\mathcal{U}_q(\mathfrak{so}_6)$-modules. We include the definitions building up to the $q$-deformed quantum group $\mathcal{U}_q(\mathfrak{so}_6)$ and its fundamental representation, but these last two definitions are the only ones mainly referenced throughout the paper. An introductory survey of quantum groups as it relates to probability is available at~\cite{KuanQuantum}. We additionally introduce crystal base theory as needed in Section~\ref{section--AlgebraGenerator}.
\begin{Definition}
 The special orthogonal Lie group ${\mathrm{SO}}_{6}$ is the multiplicative group with elements of the form ${\mathrm{SO}}_{6} = \bigl\{ X \in M_{6 \times 6} (\mathbb{C})\mid XX^{\mathsf{T}} = I,\, \det X = 1\bigr\}$.
\end{Definition}

\begin{Definition}
 The special orthogonal Lie algebra $\mathfrak{so}_{6}$ is the Lie algebra with elements of the form
 \begin{align*}
 \mathfrak{so}_{6} = \left\{\begin{bmatrix}
 A & C \\
 -C^{\mathsf{T}} & B \\
 \end{bmatrix}\mid A,B,C \in M_{3 \times 3} (\mathbb{C}),\, A = - A^{\mathsf{T}},\, B=-B^{\mathsf{T}}\right\}.
 \end{align*}
\end{Definition}

\begin{Definition}
 The universal enveloping algebra of $\mathfrak{so}_6$, denoted $\mathcal{U}(\mathfrak{so}_6)$, is generated by
 $\{E_1, E_2, E_3, F_1, F_2, F_3, H_1, H_2, H_3\}$ and the relations $[E_i, F_i] = H_i$, $1 \leq i \leq 3$, and $ E_l^2E_j + E_jE_l^2 = 2E_lE_jE_l$; $ F_l^2F_j + F_jF_l^2 = 2F_lF_jF_l$
 for $(l,j) \in \{ (1,2), (2,1), (1,3), (3,1) \}$.
\end{Definition}

\begin{Definition}
 The $q$-deformed quantum group $\mathcal{U}_q(\mathfrak{so}_6)$, as described by~\cite{Drinfeld85,Jimbo85}, is generated by $\bigl\{E_1, E_2, E_3, F_1, F_2, F_3, q^{H_1}, q^{H_2}, q^{H_3}\bigr\}$ and the relations
 \begin{align*}
 [E_i, F_i] &= \frac{q^{H_i} - q^{-H_i}}{q-q^{-1}}, \qquad
 q^{H_i}E_j = q^{\alpha_i \cdot \alpha_j}E_j q^{H_i}, \qquad
 q^{H_i}F_j = q^{-\alpha_i \cdot \alpha_j}F_j q^{H_i},
 \end{align*}
 for $ 1 \leq i$, $j \leq 3$ and
  \[
  E_l^2E_k + E_kE_l^2 = \bigl(q+q^{-1}\bigr)E_lE_kE_l; \qquad  F_l^2F_k + F_kF_l^2 = \bigl(q+q^{-1}\bigr)F_lF_kF_l
   \]
   for $(l,k) \in \{ (1,2), (2,1), (1,3), (3,1) \}$.

 For a~matrix $H$ in the Cartan subalgebra $\mathfrak{h}$ of $\mathfrak{so}_6$, let $L_i \in \mathfrak{h}^*$ map $H$ to $H_{i,i}$. We define the simple roots of $\mathcal{U}_q(\mathfrak{so}_6)$ to be $\Pi = \{\alpha_1 = L_1 - L_2, \alpha_2 = L_2 - L_2, \alpha_3 = L_2 + L_3\} $
 and the fundamental weights to be \smash{$\big\{\omega_1 = L_1, \omega_2 = \frac{1}{2}(L_1+L_2-L_3), \omega_3 = \frac{1}{2}(L_1+L_2+L_3) \big\}$}.
\end{Definition}

\begin{Definition} \label{definition--coproduct}
 We define the coproduct $\Delta$, counit $\epsilon$, and antipode $S$ of $\mathcal{U}_q(\mathfrak{so}_6)$ to be
 \begin{alignat*}{4}
 &\Delta (E_i) = E_i \otimes 1 + q^{H_i} \otimes E_i, \qquad &&\epsilon(E_i) = 0, \qquad&& S(E_i) = -E_iq^{H_i},& \\
 &\Delta(F_i) = 1 \otimes F_i + F_i \otimes q^{-H_i}, \qquad&& \epsilon(F_i) = 0,\qquad&& S(F_i) = -q^{-H_i}F_i,& \\
 &\Delta\bigl(q^{H_i}\bigr) = q^{H_i} \otimes q^{H_i},\qquad&& \epsilon\bigl(q^{H_i}\bigr) = 1,\qquad&&
 S\bigl(q^{H_i}\bigr) = q^{-H_i}.&
 \end{alignat*}
 Equipping $\mathcal{U}_q(\mathfrak{so}_6)$ with $\Delta$, $\epsilon$, and $S$ makes the quantum group a~Hopf algebra.
\end{Definition}

\begin{Definition} \label{definition-fundrep}
 A~fundamental representation of $\mathcal{U}_q(\mathfrak{so}_6)$ is defined to be the subset of
 \[M_{6 \times 6}\bigl(\mathbb{R}\bigl[q, q^{-1}\bigr]\bigr)\]
 given by the following:\vspace{-2mm}

\begin{center} \renewcommand{\arraystretch}{1.2}
 \resizebox{\columnwidth}{!}{
 \begin{tabular}{|c||c|c|c|}
 \hline
 & $E_i$ & $F_i$ & $q^{H_i}$ \\
 \hline \hline
 $i=1$ & $E_{1,2}-E_{5,4}$ & $E_{2,1}-E_{4,5}$ & $qE_{1,1} +q^{-1}E_{2,2} +E_{3,3}+q^{-1}E_{4,4} +qE_{5,5} + E_{6,6}$\\
 \hline
 $i=2$ & $E_{2,3}-E_{6,5}$ & $E_{3,2}-E_{5,6}$ & $E_{1,1} +qE_{2,2} +q^{-1}E_{3,3}+E_{4,4} +q^{-1}E_{5,5} + qE_{6,6}$\\
 \hline
 $i=3$ & $E_{2,6}-E_{3,5}$ & $E_{6,2}-E_{5,3}$ & $E_{1,1} +qE_{2,2} +qE_{3,3}+E_{4,4} +q^{-1}E_{5,5} + q^{-1}E_{6,6}$\\
 \hline
 \end{tabular}
 }
 \end{center}

\noindent
where $E_{i,j}$ indicates the $6\times 6$ zero matrix with a~$1$ in the $(i,j)$-th entry, and \smash{$q^{-H_i}$} is defined as the multiplicative inverse of \smash{$q^{H_i}$}. We denote $\mathbb{R}^6 \otimes \mathbb{R}^6$ by $V$.
\end{Definition}

\section{Computing and comparing systems}
\label{section--ProbResults}

In this section, we outline the construction of the probabilistic and algebraic particle systems, leading us to conclude that the two systems are different. This proves Theorem~\ref{thm-bigalgebra}.

\subsection{Producing the stochastic fusion generator}\label{fusion-gen}

We begin by defining the generator for stochastic fusion on $K$ $\Gamma$-lattice sites for the Type D ASEP, focusing on a~general simplification in the construction process. We will assume that $K$ is an even positive integer.

In the spirit of~\cite[Section 3]{kuan2019stochastic}, we define the transition matrix of the stochastically fused Type~D ASEP as
\smash{$Q^{(K)}_t \coloneqq \Lambda^{(K)} P^{(K)}_t \Phi^{(K)}$}.
Taking the derivative of both sides at time $t=0$, we express the same equation in terms of generators
\smash{$L^{(K)}_Q = \Lambda^{(K)} L^{(K)}_p \Phi^{(K)}$}.

Now, we will construct \smash{$\Lambda^{(2)}$}, the matrix representing fission on one $\gamma$-lattice site, and \smash{$\Phi^{(2)}$}, the matrix representing fusion on two $\Gamma$-lattice sites. We will use these two matrices to construct general \smash{$\Lambda^{(K)}$} and \smash{$\Phi^{(K)}$}, which appear in the formula for \smash{$L_Q^{(K)}$}.

To proceed, we must introduce some new notation. Any particle configuration on two lattice sites can be represented as a~function $\eta\colon\{1,2\} \rightarrow \{0,1,2,3\}$, which maps the index of a~lattice site to the particle configuration on that lattice site. Furthermore, we can let $A_1(\eta) \subset \{1,2\}$ denote the set of lattice sites that class 1 particles occupy and define $A_2(\eta)$ similarly. Now we can state the function from~\cite[Proposition 1.3]{kuan2020interactingparticlesystemstype}
\[G^2(\eta) := \prod_{x \in A_1(\eta)} q^{-2x}\prod_{x \in A_2(\eta)} q^{-2x}.\]
For simplicity, we will represent the function $\eta$ as an ordered pair $(\eta_1,\eta_2)$, where $\eta_1 := \eta(1)$ and~${\eta_2:= \eta(2)}$. For example, one can verify that
\[G^2((2,1)) = \prod_{x \in \{2\}}q^{-2x}\prod_{x \in \{1\}}q^{-2x} = q^{-6}.\]
In fact, $G^2$ is a~reversible measure on particle configurations on pairs of $\Gamma$-lattice sites, which implies that the Type D ASEP does not prefer class 1 or class 2 particles~\cite{kuan2020interactingparticlesystemstype}.

In Section~\ref{stoch-fusion}, fission was defined as taking one $\gamma$-site to two $\Gamma$-sites in a~random manner. We will now clarify how outcome states for fission are randomly selected. In particular, we will use the reversible measures defined above to construct \smash{$\Lambda^{(2)}$}, a~Markov kernel representing the fission map on one $\gamma$-site.

\begin{Definition}
 \label{LambdaTwo}
 The matrix \smash{$\Lambda^{(2)}$}, defined as follows, represents the fission map on one $\gamma$-site.

 Let the rows of \smash{$\Lambda^{(2)}$} be indexed by $x \in \chi$ and the columns be indexed by $\omega \in \Omega$. We define
\[
 \bigl(\Lambda^{(2)}\bigr)_{x, \omega} :=
 \begin{cases}
 \frac{G^2(\omega)}{\sum\limits_{m \in \phi^{-1}(x)} G^2(m)}, & \phi(\omega) = x ,\\
 0, & \phi(\omega) \neq x.
 \end{cases}
\]
\end{Definition}
By~\cite[Proposition 1.3]{kuan2020interactingparticlesystemstype} and~\cite[Section 3.1]{Blyschak_2023}, $G^2$ is a~reversible measure for the Type D ASEP where $n \in \mathbb{N}$, $q \in [0,1) \cup(1,\infty)$, and $\delta = 0$. In~\cite{kuan2019stochastic}, the fission map is defined as above using reversible measures. It is reasonable to generalize the probabilistic construction of the fission map from the particle systems considered in~\cite{kuan2019stochastic} to the Type D ASEP. Moreover, it is clear that~\smash{$\Lambda^{(2)}$} represents a~fission map in the sense of Section~\ref{stoch-fusion}, since the $(x,\omega)$-th entry of the matrix is the probability of splitting one $\gamma$-site which is in state $x$ into two $\Gamma$-sites which are in state~$\omega$.

Using the independence of fission on each $\gamma$-site, we can represent the process on multiple lattice sites via the tensor product.
\begin{Lemma} The matrix \smash{$\Lambda^{(K)}$}, defined as follows, represents the fission map on $K/2$ $\gamma$-lattice sites.
\smash{$
 \Lambda^{(K)} \coloneqq \bigl(\bigotimes_{k=1}^{K/2} \Lambda^{(2)}\bigr)$}.
\end{Lemma}
\begin{proof}
 Fission is performed independently on each $\gamma$-lattice site, so the probability measure for fission on $n$ sites is the $n$-times product of the probability measure for fission on one site. The $n$-times tensor product of the fission map on one site induces this product measure, as can be seen from the following computation:
\[
\bigl(\Lambda^{(2)} \otimes \Lambda^{(2)}\bigr)_{(x_1,x_2),(\omega_1,\omega_2)} = \Lambda^{(2)}_{(x_1,\omega_1)}\Lambda^{(2)}_{(x_2,\omega_2)}.
\]
As additional intuition for the computation above, we recall that $\dim(V \otimes W) = \dim(V) \dim(W)$ for any maps~$V$ and~$W$. The fission map which takes one $\gamma$-site to two $\Gamma$-sites, denoted by \smash{$\Lambda^{(2)}$}, has 9 rows (the number of possible states at one $\gamma$-site) and 16 columns (the number of possible states at two $\Gamma$-sites). It follows that
$\Lambda^{(2)} \otimes \Lambda^{(2)}$, which is the fission map that takes two $\gamma$-sites to four $\Gamma$-sites, has $9\cdot 9$ rows (the number of possible states at two $\gamma$-sites) and $16 \cdot 16$ columns (the number of possible states at four $\Gamma$-sites).
\end{proof}

Similarly, we can represent the fusion map on two $\Gamma$-lattice sites, which was first defined in Section~\ref{stoch-fusion}, using the matrix \smash{$\Phi^{(2)}$}.

\begin{Lemma}\label{PhiTwo}
 The matrix \smash{$\Phi^{(2)}$}, defined as follows, represents the fusion map on two $\Gamma$-sites.
 Let the rows of \smash{$\Phi^{(2)}$} be indexed by $\omega \in \Omega$ and the columns be indexed by $x \in \chi$. We define
 \[
 \bigl(\Phi^{(2)}\bigr)_{\omega, x} :=
 \begin{cases}
 1 ,& \phi(\omega) = x, \\
 0, & \phi(\omega) \neq x.
 \end{cases}
 \]
\end{Lemma}

Note that \smash{$\Phi^{(2)}$} is a~deterministic map and the matrix representation of $\phi$.
We can extend the fusion map to multiple lattice sites in the same way as we did for the fission map.

\begin{Lemma}
The matrix \smash{$\Phi^{(K)}$}, defined as follows, represents the fusion map on $K$ $\Gamma$-lattice sites.
We define \smash{$\Phi^{(K)} \coloneqq \bigl(\bigotimes_{k=1}^{K/2} \Phi^{(2)}\bigr)$}.
\end{Lemma}

Now we can define the fused Type D ASEP process by taking a~particle configuration on $\gamma$-lattice sites, applying fission, applying the usual Type D ASEP, then applying fusion. In particular, we can write an explicit formula for the generator matrix for stochastically fused Type D ASEP on an even number of lattice sites, using $\pi_i(k)$ to denote the $i$-th $\gamma$-lattice site of~${k \in \prod \chi}$ and using $\pi_i(j)$ to denote the $2i$-th $\Gamma$-lattice site of $j \in \prod \Omega$, as follows.
\begin{Theorem}\label{generator-function}
 Let \smash{$L_Q^{(2b)}$} be the generator for stochastically fused Type D ASEP on $2b$ $\Gamma$-lattice sites, where $b \in \mathbb{Z}_{> 0}$. Let \smash{$i,j\in \prod_{i=1}^b \chi$} be two fused states. We have that
 \[L_{Q_{i,j}} =\frac{1}{N_i^{(b)}}\sum_{(k,m) \in \mathcal{C}^{(b)}_i \times \mathcal{C}^{(b)}_j}\prod_{v=1}^b G^2(\pi_v(k)) L^{(2b)}_{p_{k,m}},
 \]
 where
 \[
 N_i^{(b)} = \prod_{v=1}^b \bigg(\sum_{m \in \phi^{-1}(\pi_v(i))}G^2(m)\bigg)
 \]
 is a~normalization constant and
 \[
 \mathcal{C}^{(b)}_i = \left\{k \in \prod_{v=1}^{b}\Omega \colon \bigwedge_{v=1}^{b}\bigl(\phi(\pi_v(k)) = \pi_v(i)\bigr)\right\}\]
 is the set of states in $\prod\limits_{v=1}^{b}\Omega$ that fuse to the desired state $i$.
\end{Theorem}
\begin{proof}
We have, by the definition of stochastic fusion, that
\begin{align*}
 L^{(2b)}_{Q_{i,j}} &= \Lambda^{(2b)}_{i,} \cdot \bigl(L^{(2b)}_p \Phi^{(2b)}\bigr)_{,j} = \sum\limits_{k \in \prod_{v=1}^b \Omega} \Lambda^{(2b)}_{i,k} \cdot \bigl(L^{(2b)}_{p_{k,}} \cdot \Phi^{(2b)}_{,j}\bigr) \\
 &= \sum\limits_{k \in \mathcal{C}^{(b)}_i} \frac{\prod_{v=1}^b G^2(\pi_v(k))}{N_i^{(b)}} \cdot \bigl(L^{(2b)}_{p_{k,}} \cdot \Phi^{(2b)}_{,j}\bigr) =\frac{1}{N_i^{(b)}}\sum\limits_{k \in \mathcal{C}^{(b)}_i} \prod_{v=1}^b G^2(\pi_v(k)) \cdot \bigl(L^{(2b)}_{p_{k,}} \cdot \Phi^{(2b)}_{,j}\bigr) \\
 &=\frac{1}{N_i^{(b)}}\sum\limits_{k \in \mathcal{C}^{(b)}_i} \prod_{v=1}^b G^2(\pi_v(k)) \sum\limits_{m \in \prod_{v=1}^b \Omega} \bigl(L^{(2b)}_{p_{k,m}} \cdot \Phi^{(2b)}_{m,j}\bigr)\\
 &=\frac{1}{N_i^{(b)}}\sum\limits_{k \in \mathcal{C}^{(b)}_i} \prod_{v=1}^b G^2(\pi_v(k)) \sum\limits_{m \in \mathcal{C}_j^{(b)}} L^{(2b)}_{p_{k,m}}
 =\frac{1}{N_i^{(b)}}\sum_{(k,m) \in \mathcal{C}^{(b)}_i \times \mathcal{C}^{(b)}_j}\prod_{v=1}^b G^2(\pi_v(k)) L^{(2b)}_{p_{k,m}},
\end{align*}
where the third equality follows from the observation that \smash{$\Lambda_{i,k}^{(2b)}$} is $0$ if \smash{$k \not \in \mathcal{C}_i^{(b)}$}, then applying the definition of \smash{$\Lambda^{(2b)}$}. Moreover, the sixth equality follows from the observation that \smash{$\Phi^{(2b)}_{m,j} = 0$} if \smash{$m \not \in \mathcal{C}_j^{(b)}$} and is 1 otherwise.
\end{proof}

We can also express the generator matrix for stochastically fused Type D ASEP as a~matrix product. In the following theorem, we define $L_m$ as the generator matrix for the Type D ASEP on the middle two of $K$ $\Gamma$-lattice sites.

\begin{Theorem}
\label{MiddleSwap}
The generator for fused Type D ASEP on $K/2$ $\gamma$-lattice sites, denoted \smash{$L^{(K)}_Q$}, is equal to \smash{$\Lambda^{(K)} L^{(K)}_m \Phi^{(K)}$}.
\end{Theorem}
\begin{proof}
 \label{MiddleSwapProof}
 We can express \smash{$L_p = L^{1,2}+L^{2,3}+L^{3,4}+\dots+L^{K-1,K}$}. Now we have that
 \begin{align*}
 L^{(K)}_Q\coloneqq{}& \Lambda^{(K)} L^{(K)}_p \Phi^{(K)} = \Lambda^{(K)} \bigl(L^{1,2}+L^{2,3}+L^{3,4}+\dots+L^{K-1,K}\bigr) \Phi^{(K)} \\
={}& \sum_{i=1}^{K-1}\bigl(\Lambda^{(K)}L^{i,i+1}\Phi^{(K)}\bigr) = \sum_{i=1}^{K/2} \bigl(\Lambda^{(K)}L^{2i-1,2i}\Phi^{(K)}\bigr) + \sum_{i=1}^{\frac{K}{2}-1} \bigl(\Lambda^{(K)}L^{2i,2i+1}\Phi^{(K)}\bigr) \\
={}&  \sum_{i=1}^{K/2} \left(\left(\bigotimes_{k=1}^{K/2} \Lambda^{(2)}\right)
 L^{2i-1,2i}\left(\bigotimes_{i=1}^{K/2} \Phi^{(2)}\right)\right) + \sum_{i=1}^{\frac{K}{2}-1} \bigl(\Lambda^{(K)}L^{2i,2i+1}\Phi^{(K)}\bigr) \\
={}&  0+\sum_{i=1}^{\frac{K}{2}-1} \bigl(\Lambda^{(K)}L^{2i,2i+1}\Phi^{(K)}\bigr) \hspace{2em} (*) \\
={}&  \sum_{i=1}^{\frac{K}{2}-1} \bigl(\Lambda^{(K)}L^{2i,2i+1}\Phi^{(K)}\bigr) = \Lambda^{(K)}\left(\sum_{i=1}^{\frac{K}{2}-1} L^{2i,2i+1}\right)\Phi^{(K)} = \Lambda^{(K)}L_m^{(K)}\Phi^{(K)}.
 \end{align*}
 The third equality applies \smash{$\Lambda^{(2)}$} and \smash{$\Phi^{(2)}$} to each pair of lattice sites $2i-1$ and $2i$ where $i \in \{1,\dots ,K/2\}$. Equation $(*)$ follows because
 for $i \in \{1,\dots ,K/2\}$, the $i$-th $\gamma$-site is
fissioned by~\smash{$\Lambda^{(2)}$} to a state on $\Gamma$-sites $2i-1$
and~$2i$, then the sites undergo the Type~D ASEP and are fused down
to $\gamma$-site $i$ by \smash{$\Phi^{(2)}$}. The fused state of the Type D ASEP on the $2i-1$ and $2i$ $\Gamma$-lattice sites which are being fused cannot change because the stochastically fused state is completely dependent on the number of class 1 and class~2 particles on those two $\Gamma$-lattice sites. This cannot change when only the $2i-1$ and $2i$ lattice sites interact.
\end{proof}

In particular, using the notation of Section~\ref{transition-gen}, we can write $L_Q = \Lambda L_m \Phi$. In fact, we can simplify $L_Q$ by writing it as a~permutation of the direct sum of the generators of its communicating classes (i.e., the communicating classes of $\chi \times \chi$). Note that each communicating class with the same number of states has the same generator.

\begin{Proposition}
\label{DiagonalBlockOrdering}
Let $L^D_Q = \mathcal{L}_9 \oplus \bigoplus_{i=1}^4 \mathcal{L}_6 \oplus \bigoplus_{i=1}^4 \mathcal{L}_4 \oplus \bigoplus_{i=1}^4 \mathcal{L}_3 \oplus \bigoplus_{i=1}^8 \mathcal{L}_2 \oplus \bigoplus_{i=1}^4 [0]$, where~$\mathcal{L}_9$ is the generator for the $9$-state communicating class, $\mathcal{L}_6$ is the generator for each $6$-state communicating class, and so on for $\mathcal{L}_4$, $\mathcal{L}_3$ and $\mathcal{L}_2$. There exists a~permutation matrix $C$ such that~${L_Q = CL^D_QC^{-1}}$.
\end{Proposition}
\begin{proof} \label{DiagonalBlockOrderingProof}
 Permute $L_Q$ according to the diagonal ordering for two $\gamma$ given in Appendix~\ref{PermutationsAppendix}. This is a~permutation of $L_Q$ according to its communicating classes.
\end{proof}

The states defining the permutation matrix $C$, which itself defines the permutation of states from $L_Q$ to the block diagonal $L^D_Q$, are listed in Appendix~\ref{PermutationsAppendix}.
Since the $4$-state communicating class will play a~central role in the proof of Theorem $\ref{thm-bigalgebra}$, we explicitly write
\smash{$\mathcal{L}_4 = \frac{1}{(q^4+2q^2+1)q^{2n}}M$},
where
\[M =
\left[
\begin{matrix}
 * & D_1 & D_1 & D_6\\
 D_2 & * & D_3 & D_7\\
 D_2 & D_3 & * & D_7\\
 D_4 & D_5 & D_5 & * \\
\end{matrix}
\right]
\]
and
\begin{alignat*}{3}
& D_1 = \frac{q^6+\bigl(q^4+q^2-1\bigr)q^{4n}+2q^{2n+2}}{q^4}, \qquad &&
 D_2 = q^6+\bigl(q^4+q^2-1\bigr)q^{4n}+2q^{2n+2},& \\
& D_3 = q^4+q^{4n}-2q^{2n+2},\qquad&&
 D_4 = q^8+q^{4n+4}-2q^{2n+6},& \\
& D_5 = q^{4n+2}+2q^{2n+6}+q^6+q^4-q^8,\qquad&&
 D_6 = \frac{q^4+q^{4n}-2q^{2n+2}}{q^4},&\\
 &D_7 = \frac{2q^{2n+4}+q^{4n}+q^4+q^2-q^6}{q^2},&&&
\end{alignat*}
and $*$ in row $i$ denotes the scalar such that row $i$ sums to 0.

\subsection{Producing the algebraic generator} \label{section--AlgebraGenerator}

Recall the definitions from Section~\ref{algebra_defns}. In this section, we construct a~$20$-dimensional highest-weight irreducible representation $W$ of $\mathcal{U}_q(\mathfrak{so}_6)$, which we use to construct the matrix for~\cite{kuan2020interactingparticlesystemstype}'s central element $C$ in the quantized second symmetric tensor power $\mathrm{Sym}^2_q(\mathfrak{so}_6)$. Following the path of previous research~\cite{kuan2019stochastic}, we perform a~ground state transformation so that the resulting matrix has rows which sum to $0$, allowing us to use this transformed matrix as a~generator for a~Markov process.

We will prove multiple statements partially using the theory of crystal bases as stated in~\cite{QuantumBases}. We use the correspondence between a~weight module $V(\lambda)$ and its crystal base $(\mathcal{L}(\lambda), \mathcal{B}(\lambda))$, formed of the crystal lattice and crystal basis respectively. We also use the tensor product rule, \cite[Theorem~4.4.1 on page~83 ]{QuantumBases}.

\begin{Definition}
 We denote the weight lattice of $\mathcal{U}_q(\mathfrak{so}_6)$ by $P$ and weights by $\lambda \in P$. A~vector~$v_\lambda$ is called a~highest weight vector if $E_1v_\lambda = E_2 v_\lambda = E_3v_\lambda = 0$. Then, a~$\mathcal{U}_q(\mathfrak{so}_6)$-module is a~highest weight module if it is of the form $\mathcal{U}_q(\mathfrak{so}_6)v_\lambda$, and in this case will be denoted $V(\lambda)$.
\end{Definition}
Since $e_1$ has weight $L_1$, we can use the following equation from~\cite{FultonHarris} to determine the dimension of $V(L_1)$:
\begin{align}
\label{dimensionformula}
\text{dim}(V(\lambda)) = \frac{1}{12}\prod_{1 \leq i < j \leq 3} (\lambda_i - \lambda_j +j-i)(\lambda_i+\lambda_j+6-i-j),
\end{align}
which holds for $\lambda = \lambda_1L_1 + \lambda_2L_2 + \lambda_3L_3$. Therefore, $V(L_1)$ is a~6-dimensional $\mathcal{U}_q(\mathfrak{so}_6)$-module, and $V(L_1)$ is thus the entire fundamental representation.

\begin{Lemma} \label{VxV-decomp} $V(L_1) \otimes V(L_1) \cong V(2L_1) \oplus V(L_1 + L_2) \oplus V(0)$.
\end{Lemma}
\begin{proof}
 $V(L_1)$ has a~crystal base as shown in Figure~\ref{tikz-V}, with arrow index $i$ corresponding to Kashiwara operator $\Tilde{f_i}$. In the picture below, define $1 := e_1$, $2 := e_2$, $3 := e_3$, $\bar{1} := e_4$, $\bar{2} := -e_5$, $\bar{3} := e_6$.

 \begin{figure}[!ht]\centering
 \includegraphics[width=0.55\textwidth]{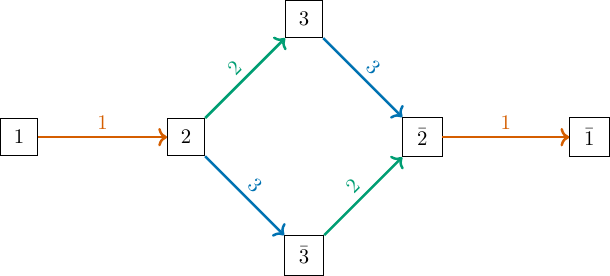}
 \caption{Crystal graph of fundamental representation $V(L_1)$ of $\mathcal{U}_q(\mathfrak{so}_6)$.}\label{tikz-V}
 \end{figure}

 Then, by the tensor product rule, we can draw the crystal graph of $V(L_1) \otimes V(L_1)$ as shown in Figure~\ref{tikz--BigCrystalGraph}. By reading off the disjoint connected components in the crystal graph of Figure~\ref{tikz--BigCrystalGraph}, we see that $V(L_1) \otimes V(L_1) \cong V(2L_1) \oplus V(L_1 + L_2) \oplus V(0)$.
\end{proof}

From the crystal graph in Figure~\ref{tikz--BigCrystalGraph--JustW}, $V(2L_1)$ is a~$20$-dimensional vector space. Note that it is irreducible as a~highest weight representation. Denote $V(2L_1)$ as $W$; a~further investigation into this space, including its explicit basis, can be found in Lemma~\ref{lemma--Wis20dimensional}.

In 2020, \cite{kuan2020interactingparticlesystemstype} found the following central element of $\mathcal{U}_q(\mathfrak{so}_6)$ (note that the fundamental representation of $\mathfrak{so}_{2n}$ presented in that paper is incorrect, leading to sign errors in its central elements. The $E_3$ used in that paper is $-E_3$ in our notation). We write $F_{ij}$ to shorten $ F_iF_j$ and~$r$ to represent $q+q^{-1}$.
\begin{align}
 C = {}& q^{-4-2H_1-H_2-H_3}
 +q^{-2-H_2-H_3}
 +q^{H_2-H_3}
 +q^{H_3-H_2}
 +q^{2+H_2+H_3}\nonumber \\
 &+q^{4+2H_1+H_2+H_3}
 +\frac{r^2}{q^3}F_1q^{-H_1-H_2-H_3}E_1
 +\frac{r^2}{q}F_2q^{-H_3}E_2 \nonumber\\
 &+\frac{r^2}{q}F_3q^{-H_2}E_3
 +r^2qF_2q^{H_3}E_2
 +r^2qF_3q^{H_2}E_3
 +r^2q^3F_1q^{H_1+H_2+H_3}E_1\nonumber \\
 &+\frac{r^2}{q^3}(qF_{12}-F_{21})q^{-H_1-H_3}(qE_{21}-E_{12})
 +\frac{r^2}{q^3}(qF_{13}-F_{31})q^{-H_1-H_2}(qE_{31}-E_{13}) \nonumber\\
 &+r^2q(qF_{21}-F_{12})q^{H_1+H_3}(qE_{12}-E_{21})
 +r^2q(qF_{31}-F_{13})q^{H_1+H_2}(qE_{13}-E_{31})\nonumber \\
&+\frac{r^2}{q^3}\bigl(q^2F_{123}-qF_{213}-qF_{312}+F_{231}\bigr)q^{-H_1}\bigl(q^2E_{231}-qE_{312}-qE_{213}+E_{123}\bigr) \nonumber\\
&+\frac{r^2}{q}\bigl(q^2F_{231}-qF_{312}-qF_{213}+F_{123}\bigr)q^{H_1}\bigl(q^2E_{123}-qE_{213}-qE_{312}+E_{231}\bigr)\nonumber \\
&+\frac{r^4}{q^2}\bigl(\bigl(q^2+1\bigr)F_{1231}-qF_{1312}-qF_{2131}\bigr)\bigl(\bigl(q^2+1\bigr)E_{1231}-qE_{1312}-qE_{2131}\bigr) \nonumber\\
 &+r^4F_2F_3E_2E_3.\label{equation--C}
\end{align}

By some tedious $36\times 36$ matrix computations, one can find $\pi_V(C)$, the matrix corresponding to the element $C$ in the fundamental representation $V$. Constructing the change of basis matrix~$Y$ with first 20 columns being Lemma~\ref{lemma--Wis20dimensional}'s basis of $W$ and the last 16 columns extending to a~basis of $V(L_1)$, the matrix $Y^{-1}\pi_{V(L_1)}(C)Y$ has an upper left $20 \times 20$ block of $\bigl(q^8+q^2+2+q^{-2}+q^{-8}\bigr) \mathrm{Id}_{20 \times 20} = \pi_{V(L_1)}(C)|_W$. As expected, then, $\pi_{V(L_1)}(C)|_W$ is a~constant times $\mathrm{Id}_{20 \times 20}$.

 \begin{figure}[hbt!]
 \centering
 \includegraphics[width=0.83\textwidth]{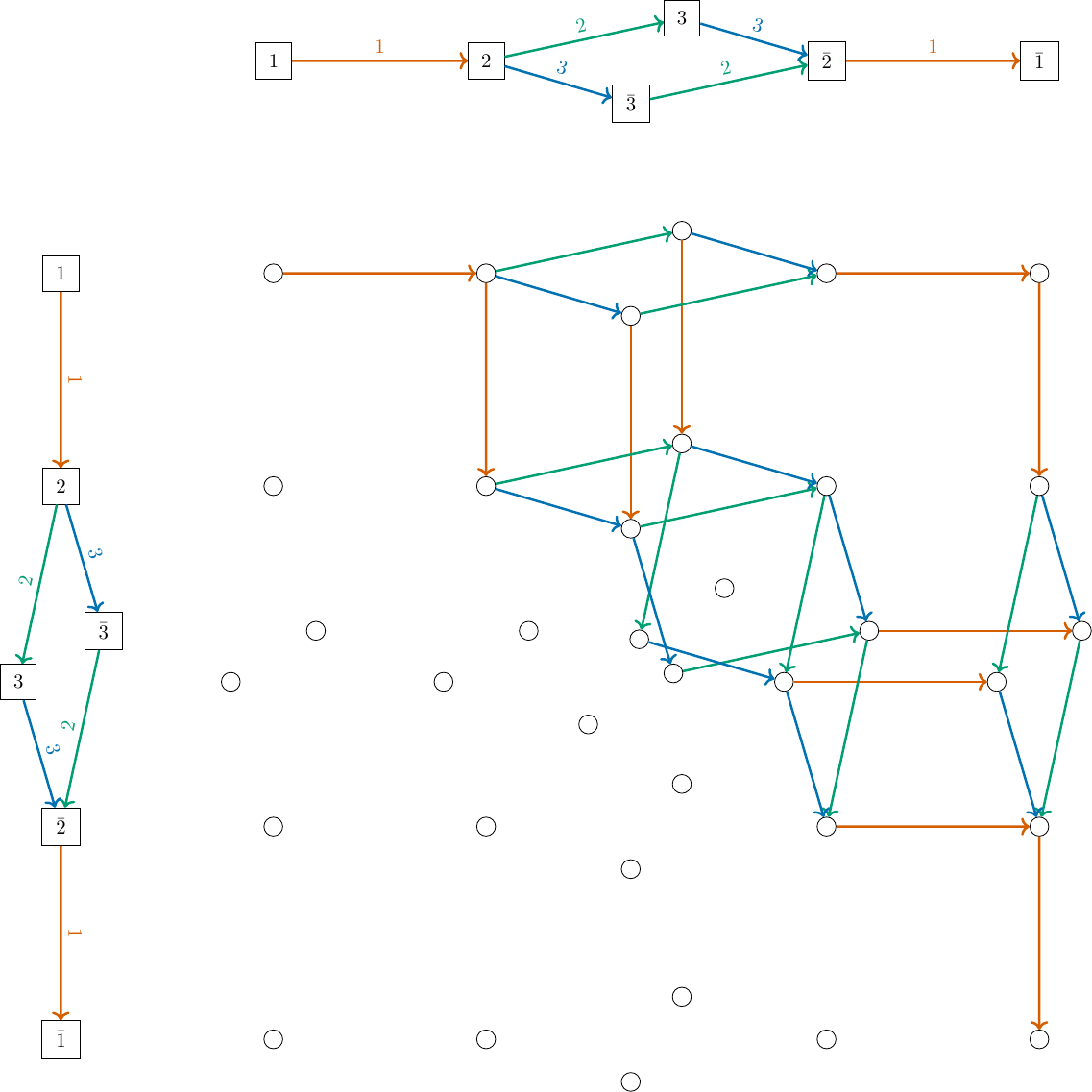}
 \caption{Crystal graph of $V(2L_1)$ in the representation $V(L_1) \otimes V(L_1)$.}\label{tikz--BigCrystalGraph--JustW}
 \end{figure}

 \begin{figure}[hbt!]
 \centering
 \includegraphics[width=0.8\textwidth]{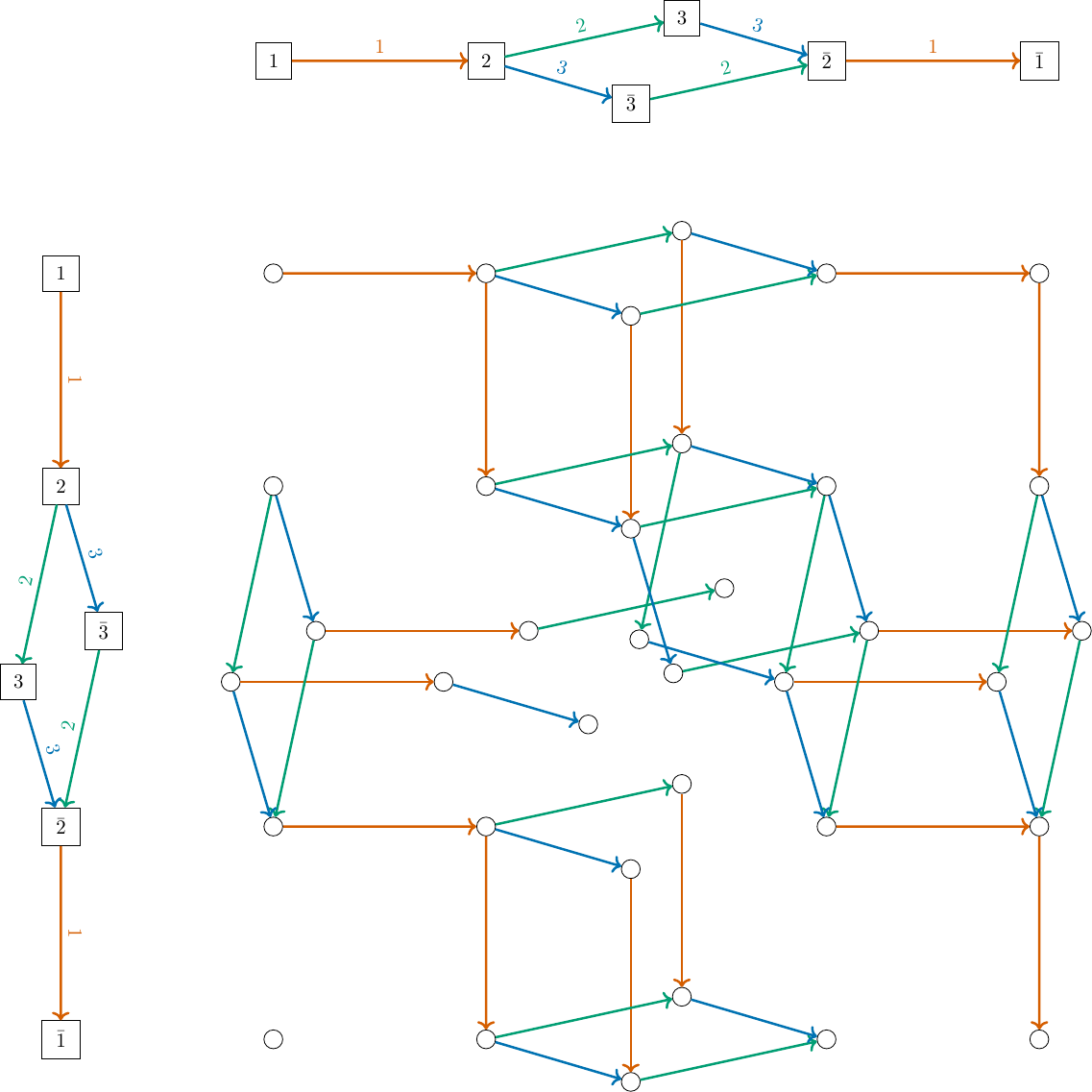}
 \caption{Crystal graph of $\mathcal{U}_q(\mathfrak{so}_6)$ in the representation $V(L_1) \otimes V(L_1)$.}\label{tikz--BigCrystalGraph}
 \end{figure}

\subsubsection[Decomposition of W tensor W into irreducible representations]{Decomposition of $\boldsymbol{W \otimes W}$ into irreducible representations}
In order to decompose $W \otimes W$ into a~direct sum of irreducible representations, we would like to use some analog of the tensor product rule as done in Lemma~\ref{VxV-decomp}. However, since our current crystal graph of $W$ is complicated, creating the crystal graph of $W \otimes W$ would be very difficult. Therefore, we introduce Young tableaux to be able to visualize $W \otimes W$, and use the corresponding tensor product rule to decompose $W \otimes W$. An interested reader may reference chapters 7 and 8 of~\cite{QuantumBases} for all appropriate background.

The chain of correspondence is as follows: $\mathcal{B}(2L_1)$ describes $W$, as seen in Lemma~\ref{VxV-decomp}, and~\smash{$
 \mathcal{B} \bigl( \vcenter{\hbox{\text{
 \begin{tikzpicture}
 \coordinate (11) at (.5,.25);
 \coordinate (01) at (0,.25);
 \coordinate (00) at (0,-.25);
 \coordinate (10) at (.5,-.25);
 \coordinate (20) at (1,-.25);
 \coordinate (21) at (1,.25);
 \draw[thick, -] (00) -- (20) -- (21) -- (01) -- (00);
 \draw[thick, -] (10) -- (11);
 \end{tikzpicture}
 }}}\bigr) \coloneqq \mathcal{B}(\mathcal{Y})
$}
describes $\mathcal{B}(2L_1)$. Thus, the decomposition of $\mathcal{B}(\mathcal{Y}) \otimes \mathcal{B}(\mathcal{Y})$ dictates the structure of $W \otimes W$. This brings us to Lemma~\ref{lemma-young}.
\begin{Lemma} \label{lemma-young} $W \otimes W$ decomposes as follows into irreducible representations
\begin{gather*}
 W \otimes W \cong V(4L_1) \oplus V(3L_1+L_2) \oplus V(2L_1+2L_2) \oplus V(2L_1) \oplus V(L_1+L_2) \oplus V(0).
\end{gather*}
\end{Lemma}
\begin{proof}
 As can be read off of Figure~\ref{tikz--BigCrystalGraph}, \smash{$
 \mathcal{B}(2L_1) = \bigl\{\vcenter{\hbox{\text{
 \begin{tikzpicture}
 \coordinate (11) at (.5,.25);
 \coordinate (01) at (0,.25);
 \coordinate (00) at (0,-.25);
 \coordinate (10) at (.5,-.25);
 \coordinate (20) at (1,-.25);
 \coordinate (21) at (1,.25);
 \draw[thick, -] (00) -- (20) -- (21) -- (01) -- (00);
 \draw[thick, -] (10) -- (11);
 \node[circle] (ablah) at (.25,-.05) {$a$};
 \node[circle] (ablah) at (.75,0) {$b$};
 \end{tikzpicture}
 }}} \big| b \succeq a \bigr\}$}.
 We choose the English notation and convention of reading Young diagrams down columns from right to left. By the tensor product rule for Young diagrams (\cite[Theorem 8.6.6, p.~206]{QuantumBases}),
 \begin{gather*}
 \mathcal{B}(\mathcal{Y}) \otimes \mathcal{B}(\mathcal{Y}) \cong \bigoplus_{b_1 \otimes b_2 \in W} \mathcal{B}(\mathcal{Y}[b_1, b_2]).
 \end{gather*}
 To compute the right-hand side, note that
 \begin{alignat*}{3}
& \mathcal{B}(\mathcal{Y}[v_1, v_1]) = \mathcal{B} \bigl( \vcenter{\hbox{\text{
 \begin{tikzpicture}
 \coordinate (11) at (.5,.25);
 \coordinate (01) at (0,.25);
 \coordinate (00) at (0,-.25);
 \coordinate (10) at (.5,-.25);
 \coordinate (20) at (1,-.25);
 \coordinate (21) at (1,.25);
 \coordinate (30) at (1.5, -.25);
 \coordinate (31) at (1.5,.25);
 \coordinate (40) at (2, -.25);
 \coordinate (41) at (2,.25);
 \draw[thick, -] (00) -- (40) -- (41) -- (01) -- (00);
 \draw[thick, -] (10) -- (11);
 \draw[thick, -] (20) -- (21);
 \draw[thick, -] (30) -- (31);
 \end{tikzpicture}
 }}}\bigr) ,
 \qquad&&
 \mathcal{B}(\mathcal{Y}[v_2, v_1]) = \mathcal{B} \left( \vcenter{\hbox{\text{
 \begin{tikzpicture}
 \coordinate (11) at (.5,.25);
 \coordinate (01) at (0,.25);
 \coordinate (00) at (0,-.25);
 \coordinate (10) at (.5,-.25);
 \coordinate (20) at (1,-.25);
 \coordinate (21) at (1,.25);
 \coordinate (30) at (1.5, -.25);
 \coordinate (31) at (1.5,.25);
 \coordinate (02) at (0, -.75);
 \coordinate (12) at (.5, -.75);
 \draw[thick, -] (00) -- (30) -- (31) -- (01) -- (00);
 \draw[thick, -] (10) -- (11);
 \draw[thick, -] (20) -- (21);
 \draw[thick, -] (02) -- (12);
 \draw[thick, -] (01) -- (02);
 \draw[thick, -] (11) -- (12);
 \end{tikzpicture}
 }}}\right) ,&
 \\
 &\mathcal{B}(\mathcal{Y}[v_2, v_2]) = \mathcal{B} \left( \vcenter{\hbox{\text{
 \begin{tikzpicture}
 \coordinate (11) at (.5,.25);
 \coordinate (01) at (0,.25);
 \coordinate (00) at (0,-.25);
 \coordinate (10) at (.5,-.25);
 \coordinate (20) at (1,-.25);
 \coordinate (21) at (1,.25);
 \coordinate (30) at (1.5, -.25);
 \coordinate (31) at (1, -.75);
 \coordinate (02) at (0, -.75);
 \coordinate (12) at (.5, -.75);
 \draw[thick, -] (01) -- (02) -- (31) -- (21) -- (01);
 \draw[thick, -] (00) -- (20);
 \draw[thick, -] (11) -- (12);
 \end{tikzpicture}
 }}}\right) ,
 \qquad&&
 \mathcal{B}(\mathcal{Y}[v_{\bar{1}}, v_1]) = \mathcal{B} \bigl( \vcenter{\hbox{\text{
 \begin{tikzpicture}
 \coordinate (11) at (.5,.25);
 \coordinate (01) at (0,.25);
 \coordinate (00) at (0,-.25);
 \coordinate (10) at (.5,-.25);
 \coordinate (20) at (1,-.25);
 \coordinate (21) at (1,.25);
 \draw[thick, -] (00) -- (20) -- (21) -- (01) -- (00);
 \draw[thick, -] (10) -- (11);
 \end{tikzpicture}
 }}}\bigr) ,&
 \\
 &\mathcal{B}(\mathcal{Y}[v_{\bar{1}},v_2]) = \mathcal{B} \left( \vcenter{\hbox{\text{
 \begin{tikzpicture}
 \coordinate (11) at (.5,.25);
 \coordinate (01) at (0,.25);
 \coordinate (00) at (0,-.25);
 \coordinate (10) at (.5,-.25);
 \coordinate (20) at (1,-.25);
 \coordinate (21) at (1,.25);
 \coordinate (30) at (1.5, -.25);
 \coordinate (31) at (1, -.75);
 \coordinate (02) at (0, -.75);
 \coordinate (12) at (.5, -.75);
 \draw[thick, -] (01) -- (02) -- (12) -- (11) -- (01);
 \draw[thick, -] (00) -- (10);
 \end{tikzpicture}
 }}}\right) ,\qquad&&
 \mathcal{B}(\mathcal{Y}[v_{\bar{1}}, v_{\bar{1}}]) = \mathcal{B}(\varnothing).&
 \end{alignat*}
 Thus, we have that
 \begin{align*}
 \mathcal{B}(\mathcal{Y}) \otimes \mathcal{B}(\mathcal{Y}) \cong{}&
 \mathcal{B} \left( \vcenter{\hbox{\text{
 \begin{tikzpicture}
 \coordinate (11) at (.5,.25);
 \coordinate (01) at (0,.25);
 \coordinate (00) at (0,-.25);
 \coordinate (10) at (.5,-.25);
 \coordinate (20) at (1,-.25);
 \coordinate (21) at (1,.25);
 \coordinate (30) at (1.5, -.25);
 \coordinate (31) at (1.5,.25);
 \coordinate (40) at (2, -.25);
 \coordinate (41) at (2,.25);
 \draw[thick, -] (00) -- (40) -- (41) -- (01) -- (00);
 \draw[thick, -] (10) -- (11);
 \draw[thick, -] (20) -- (21);
 \draw[thick, -] (30) -- (31);
 \end{tikzpicture}
 }}}\right)
 \oplus
 \mathcal{B} \left( \vcenter{\hbox{\text{
 \begin{tikzpicture}
 \coordinate (11) at (.5,.25);
 \coordinate (01) at (0,.25);
 \coordinate (00) at (0,-.25);
 \coordinate (10) at (.5,-.25);
 \coordinate (20) at (1,-.25);
 \coordinate (21) at (1,.25);
 \coordinate (30) at (1.5, -.25);
 \coordinate (31) at (1.5,.25);
 \coordinate (02) at (0, -.75);
 \coordinate (12) at (.5, -.75);
 \draw[thick, -] (00) -- (30) -- (31) -- (01) -- (00);
 \draw[thick, -] (10) -- (11);
 \draw[thick, -] (20) -- (21);
 \draw[thick, -] (02) -- (12);
 \draw[thick, -] (01) -- (02);
 \draw[thick, -] (11) -- (12);
 \end{tikzpicture}
 }}}\right)
 \oplus
 \mathcal{B} \left( \vcenter{\hbox{\text{
 \begin{tikzpicture}
 \coordinate (11) at (.5,.25);
 \coordinate (01) at (0,.25);
 \coordinate (00) at (0,-.25);
 \coordinate (10) at (.5,-.25);
 \coordinate (20) at (1,-.25);
 \coordinate (21) at (1,.25);
 \coordinate (30) at (1.5, -.25);
 \coordinate (31) at (1, -.75);
 \coordinate (02) at (0, -.75);
 \coordinate (12) at (.5, -.75);
 \draw[thick, -] (01) -- (02) -- (31) -- (21) -- (01);
 \draw[thick, -] (00) -- (20);
 \draw[thick, -] (11) -- (12);
 \end{tikzpicture}
 }}}\right)
 \oplus
 \mathcal{B} \bigl( \vcenter{\hbox{\text{
 \begin{tikzpicture}
 \coordinate (11) at (.5,.25);
 \coordinate (01) at (0,.25);
 \coordinate (00) at (0,-.25);
 \coordinate (10) at (.5,-.25);
 \coordinate (20) at (1,-.25);
 \coordinate (21) at (1,.25);
 \draw[thick, -] (00) -- (20) -- (21) -- (01) -- (00);
 \draw[thick, -] (10) -- (11);
 \end{tikzpicture}
 }}}\bigr)\\
 &
 \oplus
 \mathcal{B} \left( \vcenter{\hbox{\text{
 \begin{tikzpicture}
 \coordinate (11) at (.5,.25);
 \coordinate (01) at (0,.25);
 \coordinate (00) at (0,-.25);
 \coordinate (10) at (.5,-.25);
 \coordinate (20) at (1,-.25);
 \coordinate (21) at (1,.25);
 \coordinate (30) at (1.5, -.25);
 \coordinate (31) at (1, -.75);
 \coordinate (02) at (0, -.75);
 \coordinate (12) at (.5, -.75);
 \draw[thick, -] (01) -- (02) -- (12) -- (11) -- (01);
 \draw[thick, -] (00) -- (10);
 \end{tikzpicture}
 }}}\right) \oplus \mathcal{B}(\varnothing),
 \end{align*}
 which implies that $
 W \otimes W \cong V(4L_1) \oplus V(3L_1+L_2) \oplus V(2L_1+2L_2) \oplus V(2L_1) \oplus V(L_1+L_2) \oplus V(0)$.
 Note that \smash{$
 \operatorname{dim}(W \otimes W) = 400 = 105+175+84+20+15+1 = \sum_\lambda \operatorname{dim}V(\lambda)$}
 using the dimension formula equation~\eqref{dimensionformula}. The order of dimensions in the above sum is the same order in which the irreducible representations are listed in the decomposition above.
\end{proof}

\subsubsection[Decomposition of W tensor W into weight spaces]{Decomposition of $\boldsymbol{W \otimes W}$ into weight spaces}
Each highest-weight $\mathcal{U}_q(\mathfrak{so}_6)$-module has a~weight space decomposition, so we can express $W \otimes W$ as a~direct sum of weight spaces. We are most interested in the dimension of each weight space, which will help deduce the shape of $\pi_{W \otimes W}(C)$. \label{subsect-weightspacedecomp}

\begin{Lemma}\label{WeightSpaceLemma}
 $W \otimes W$ admits a~decomposition into eighty-five weight spaces, one with dimension~$22$, twelve with dimension $12$, six with dimension $8$, twenty-four with dimension $4$, twelve with dimension $3$, twenty-four with dimension $2$, and six with dimension $1$.
\end{Lemma}
\begin{proof}
 For ease of notation, define $L_{-i} \coloneqq -L_i$ for $1 \le i \le 3$. Recall that the weights associated with the standard basis vectors of the fundamental representation $\mathbb{R}^6$ of $\mathfrak{so}_6$ are as follows:
 \begin{center}
 \begin{tabular}{|c|c|}
 \hline
 $e_1$ & $L_1$ \\
 \hline
 $e_2$ & $L_2$ \\
 \hline
 $e_3$ & $L_3$ \\
 \hline
 $e_4$ & $L_{-1}$ \\
 \hline
 $e_5$ & $L_{-2}$ \\
 \hline
 $e_6$ & $L_{-3}$ \\
 \hline
 \end{tabular}
 \end{center}

 It is a~property of weights that if $v$ and $v'$ are vectors in a~representation $V$ with respective weights $\lambda$ and $\lambda'$, then $v \otimes v'$ is a~vector in $V \otimes V$ with weight $\lambda + \lambda'$. Thus, we can observe from Figure~\ref{tikz--BigCrystalGraph--JustW} that the weights of basis vectors in $W$ take the following forms, for~${i,j \in \{1,2,3,-1,-2,-3\}}$ with $i$, $-i$, $j$, $-j$ all unique $
 \{0,2L_i,L_i+L_j\}$.
 Note that two basis vectors in $W$ ($e_1 \otimes e_4$ and $e_2 \otimes e_5$) have weight $0$, and each of the other eighteen basis vectors in $W$ has a~distinct weight of the form $2L_i$ or $L_i+L_j$. By considering possible sums of the weights above, the weights of vectors in $W \otimes W$ take one of the following forms, for $i,j,k \in \{1,2,3,-1,-2,-3\}$ with $i$, $-i$, $j$, $-j$, $k$, $-k$ all unique
 \begin{align*}
 \{0,L_i + L_j,2L_i,2L_i + L_j + L_k,2L_i + 2L_j,3L_i + L_j,4L_i\}.
 \end{align*}
 Now we count the number of weight spaces corresponding to each of the forms above.
 \begin{itemize}\itemsep=0pt
 \item There is evidently one weight space with weight $0$.
 \item There are $6 \cdot 4 = 24$ possible pairs $(i,j)$. Since the pairs $(i,j)$ and $(j,i)$ yield the same weight, there are twelve weight spaces with a~weight of the form $L_i + L_j$.
 \item Since there are $6$ possible values for $i$, there are six weight spaces with a~weight of the form $2L_i$.
 \item There are $6 \cdot 4 \cdot 2 = 48$ possible triplets $(i,j,k)$. Since the triplets $(i,j,k)$ and $(i,k,j)$ yield the same weight, there are twenty-four weight spaces with a~weight of the form~${2L_i + L_j + L_k}$.
 \item There are $6 \cdot 4 = 24$ possible pairs $(i,j)$. Since the pairs $(i,j)$ and $(j,i)$ yield the same weight, there are twelve weight spaces with a~weight of the form $2L_i + 2L_j$.
 \item There are $6 \cdot 4 = 24$ possible pairs $(i,j)$. Since each pair yields a~distinct weight, there are twenty-four weight spaces with a~weight of the form $3L_i + L_j$.
 \item Since there are $6$ possible values for $i$, there are six weight spaces with a~weight of the form $4L_i$.
 \end{itemize}
 Now we count the dimension of each type of weight space. We can do so by counting the number of basis vectors $w_1 \otimes w_2 \in W \otimes W$ such that the sum of the weights corresponding to $w_1$ and~$w_2$ in $W$ equals the desired weight. Denote the weights in $W$ corresponding to $w_1$ and $w_2$ by $\lambda_1$ and $\lambda_2$, respectively.
 \begin{itemize}\itemsep=0pt
 \item If $\lambda_1+\lambda_2 = 0$, then either
 \begin{itemize}\itemsep=0pt
\item[$*$] $\lambda_1 = 0$. Then $\lambda_2 = 0$ as well. There are thus two options for each of $w_1$ and $w_2$, and thus four total possibilities in this case.
\item[$*$] $\lambda_1 \ne 0$. Then there are eighteen possible values for $w_1$ with nonzero weights, and exactly one possible value for $w_2$ such that $\lambda_2 = -\lambda_1$.
 \end{itemize}
 Thus, the dimension of the weight space of $W \otimes W$ with weight $0$ is twenty-two.
 \item If $\lambda_1+\lambda_2$ takes the form $L_i+L_j$, then either
 \begin{itemize}\itemsep=0pt
\item[$*$] $\lambda_1 = L_i + L_j$ and $\lambda_2 = 0$; or $\lambda_1 = 0$ and $\lambda_2 = L_i + L_j$. Since the weight space of $W$ with weight $0$ has dimension $2$, there are four total possibilities in this case.
\item[$*$] $\lambda_1 = L_i - L_j$ and $\lambda_2 = 2L_j$; or $\lambda_1 = 2L_j$ and $\lambda_2 = L_i - L_j$; or $\lambda_1 = -L_i + L_j$ and~${\lambda_2 = 2L_i}$; or $\lambda_1 = 2L_i$ and $\lambda_2 = -L_i + L_j$. There are four total possibilities in this case.
\item[$*$] $\lambda_1 = L_i + L_k$ and $\lambda_2 = -L_j - L_k$; or $\lambda_1 = -L_j - L_k$ and $\lambda_2 = L_i + L_k$. Since there are two possibilities for $k$ once $i$ and $j$ are fixed, there are four total possibilities in this case.
 \end{itemize}
 Thus, each weight space of $W \otimes W$ with a~weight of the form $L_i + L_j$ has dimension twelve.
 \item If $\lambda_1 + \lambda_2$ takes the form $2L_i$, then either
 \begin{itemize}\itemsep=0pt
\item[$*$] $\lambda_1 = 2L_i$ and $\lambda_2 = 0$; or $\lambda_1 = 0$ and $\lambda_2 = 2L_i$. Since the weight space of $W$ with weight $0$ has dimension $2$, there are four total possibilities in this case.
\item[$*$] $\lambda_1 = L_i + L_j$ and $\lambda_2 = L_i - L_j$. Since there are four possibilities for $j$ once $i$ is fixed, there are four total possibilities in this case.
 \end{itemize}
 Thus, each weight space of $W \otimes W$ with a~weight of the form $2L_i$ has dimension eight.
 \item If $\lambda_1 + \lambda_2$ takes the form $2L_i + L_j + L_k$, then either
 \begin{itemize}
\itemsep=0pt
 \item[$*$] $\lambda_1 = 2L_i$ and $\lambda_2 = L_j + L_k$.
 \item[$*$] $\lambda_1 = L_j + L_k$ and $\lambda_2 = 2L_i$.
 \item[$*$] $\lambda_1 = L_i + L_j$ and $\lambda_2 = L_i + L_k$.
 \item[$*$] $\lambda_1 = L_i + L_k$ and $\lambda_2 = L_i + L_j$.
 \end{itemize}
 Thus, each weight space of $W \otimes W$ with a~weight of the form $2L_i + L_j + L_k$ has dimension four.
 \item If $\lambda_1 + \lambda_2$ takes the form $2L_i + 2L_j$, then either
 \begin{itemize}
\itemsep=0pt
 \item[$*$] $\lambda_1 = 2L_i$ and $\lambda_2 = 2L_j$.
 \item[$*$] $\lambda_1 = 2L_j$ and $\lambda_2 = 2L_i$.
 \item[$*$] $\lambda_1 = L_i + L_j$ and $\lambda_2 = L_i + L_j$.
 \end{itemize}
 Thus, each weight space of $W \otimes W$ with a~weight of the form $2L_i + 2L_j$ has dimension three.
 \item If $\lambda_1 + \lambda_2$ takes the form $3L_i + L_j$, then either
 \begin{itemize}
\itemsep=0pt
 \item[$*$] $\lambda_1 = 2L_i$ and $\lambda_2 = L_i + L_j$.
 \item[$*$] $\lambda_1 = L_i + L_j$ and $\lambda_2 = 2L_i$.
 \end{itemize}
 Thus, each weight space of $W \otimes W$ with a~weight of the form $3L_i + L_j$ has dimension two.
\item If $\lambda_1 + \lambda_2$ takes the form $4L_i$, then $\lambda_1 = 2L_i$ and $\lambda_2 = 2L_i$, so each weight space of $W \otimes W$ with a~weight of the form $4L_i$ has dimension one. \hfill $\qed$
 \end{itemize}\renewcommand{\qed}{}
\end{proof}

\subsubsection[Block form of pi\_\{W otimes W\}(C)]{Block form of $\boldsymbol{\pi_{W \otimes W}(C)}$} \label{subsection-piwwc-entries}

Since $W \otimes W$ can be decomposed into a~direct sum of weight spaces, any matrix in $W \otimes W$ can be expressed as a~direct sum of matrices in the weight spaces of $W \otimes W$. This yields the following corollary of Lemma~\ref{WeightSpaceLemma}.

\begin{Corollary}
 $\pi_{W \otimes W}(C)$ admits a~block diagonal decomposition into a~direct sum of one $22 \times 22$ block, twelve $12 \times 12$ blocks, six $8 \times 8$ blocks, twenty-four $4 \times 4$ blocks, twelve $3 \times 3$ blocks, twenty-four $2 \times 2$ blocks, and six $1 \times 1$ blocks. The basis in respect to which results in $\pi_{W \otimes W}(C)$ being blocked this way is included in Appendix~{\rm \ref{appendix-basisWxW}}.
\end{Corollary}

\begin{Remark}
Note that since $\mathfrak{so}_6$ and $\mathfrak{sl}_4$ are isomorphic as Lie algebras, computing the Kostka numbers for $\mathfrak{sl}_4$ would also yield the block sizes found above. For example, the $22\times 22$ block can be calculated by $22=3+6+7+2+3+1$ where each summand appears as $K_{\lambda\mu}$ for the following choices of $\lambda$ and $\mu$:
\begin{alignat*}{4}
&K_{\lambda\mu} = 7,  \qquad&& \lambda = (4,3,1,0), \qquad&&\mu = (2,2,2,2),&\\
&K_{\lambda\mu} = 6, \qquad&& \lambda = (4,2,2,0),\qquad&& \mu = (2,2,2,2),& \\
&K_{\lambda\mu} = 3, \qquad&& \lambda = (4,4,0,0),\qquad&& \mu = (2,2,2,2),& \\
&K_{\lambda\mu} = 2, \qquad&& \lambda = (2,2,0,0),\qquad&& \mu = (1,1,1,1),& \\
&K_{\lambda\mu} = 3, \qquad&& \lambda = (2,1,1,0),\qquad&& \mu = (1,1,1,1),& \\
&K_{\lambda\mu} = 1, \qquad&& \lambda = (0,0,0,0),\qquad&& \mu = (0,0,0,0).&
\end{alignat*}
\end{Remark}

\subsubsection[Type A ground state transformation]{Type A ground state transformation}
\begin{Definition}
 \label{defn-groundstate}
 A~ground state transformation is a~map from a~Hamiltonian matrix, $H$, to a~matrix whose rows sum to $0$.
We define such a map to be of the form $H \mapsto G^{-1}HG -a\, \mathrm{Id}$, where
$G$ is a diagonal matrix and $a$ is in the underlying field.
\end{Definition}

Appendix~\ref{subsection-groundstatetransformation} provides a~ground state transformation that worked with Type~A Lie algebras. However, as can be seen using the respective crystal graphs, Type~D Lie algebras have more complicated behavior: the crystal graph for the fundamental representation $Z$ of $\mathcal{U}_q(\mathfrak{sl}_4)$ is shown in Figure~\ref{figure-fund-sl4}.
\begin{figure}[h!]
 \centering
 \includegraphics[width=0.4\textwidth]{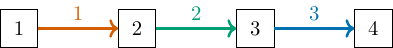}
 \caption{The crystal graph of the fundamental representation $Z$ of $\mathcal{U}_q(\mathfrak{sl}_4)$.} \label{figure-fund-sl4}
 \end{figure}

As a~result, the crystal graph for $Z \otimes Z$ of $\mathcal{U}_q(\mathfrak{sl}_4)$ is the graph shown in Figure~\ref{figure-vv-sl4}. Note that the cycle from the highest weight vector to the lowest weight vector is straightforward.
\begin{figure}[ht!]
 \centering
 \includegraphics[width=0.4\textwidth]{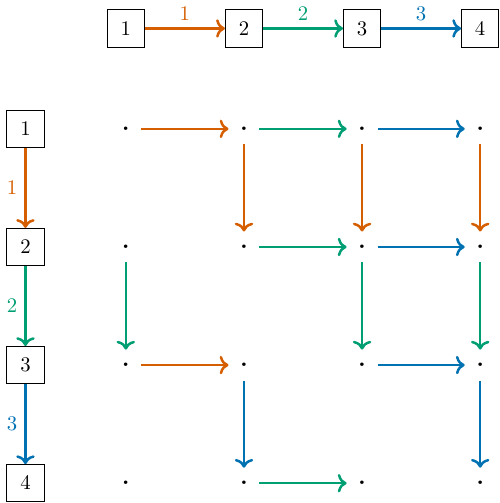}
 \caption{Crystal graph of $\mathcal{U}_q(\mathfrak{sl}_4)$ in $Z \otimes Z$.} \label{figure-vv-sl4}
 \end{figure}

Instead, in $\mathcal{U}_q(\mathfrak{so}_6)$, the crystal graph for $V(L_1) \otimes V(L_1)$ shown in Figure~\ref{tikz--BigCrystalGraph} was much less simple due to the more complex crystal graph for its fundamental representation $V(L_1)$ shown in Figure~\ref{tikz-V}. In this case, the shortest path from the highest weight vector to the lowest weight vector requires eight applications of $F_i$s as shown in the third example within the proof of Lemma~\ref{lemma--Wis20dimensional}.

Thus, intuitively, when we apply the tensor product rule again to $W$, the path from the highest weight vector to the lowest weight vector should be even more intricate. Therefore, the ground state transformation in Appendix~\ref{subsection-groundstatetransformation} should include more complicated structures than solely permutations of $F_1^{k_1}F_2^{k_2}F_3^{k_3}$ in order to generate a~complete matrix $G$. This result is formalized below, where we show that the procedure of Appendix~\ref{subsection-groundstatetransformation} does not extend to the Type~D case.

\begin{Lemma}
\label{lemma--groundstate}
 For any nonnegative integers $k_1$, $k_2$, $k_3$, there exist basis vectors $u_i, u_j \in W$ such that for any permutation $\sigma \in S_3$,
 \[\big\langle u_i \otimes u_j| \Delta(F_{\sigma(1)})^{k_{\sigma(1)}}\Delta(F_{\sigma(2)})^{k_{\sigma(2)}}\Delta(F_{\sigma(3)})^{k_{\sigma(3)}}| e_1 \otimes e_1 \otimes e_1 \otimes e_1 \big\rangle = 0.
 \]
 The procedure for finding a~canonical ground state transformation in the Type A~case, described in Appendix~{\rm\ref{subsection-groundstatetransformation}}, therefore results in a~non-canonical ground state transformation of $\pi_{W \otimes W}(C)$.
\end{Lemma}
\begin{proof}[Proof]
 By the crystal graph for $V(L_1)$ shown in Figure~\ref{tikz-V}, which is based on the fundamental representation defined in Definition~\ref{definition-fundrep}, we see that we will not be able to compute $e_4$ without applying either $F_1F_2F_3F_1$ or $F_1F_3F_2F_1$.

Now, note that $W$ is a~subspace of $\mathrm{Sym}^2_q\bigl(\mathbb{R}^6\bigr)$; in particular, as shown in Lemma~\ref{lemma--Wis20dimensional}, $W$ has basis elements of the form
\begin{gather*}
 \{ p_i(q)e_i \otimes e_i, p_i(q)e_i \otimes e_j+p_j(q)e_j \otimes e_i  \mid\\
 \qquad i \neq j, (i,j) \neq (1,4),(2,5),(3,6), p_j \text{ Laurent polynomial}\} \\
 \qquad \cup\bigl\{ -e_1 \otimes e_4 +q^{-1}e_2 \otimes e_5 - q^{-2}e_4 \otimes e_1+q^{-1}e_5 \otimes e_2\bigr\} \\
 \qquad \cup\bigl\{e_2 \otimes e_5 - q^{-1}e_3 \otimes e_6+q^{-2}e_5 \otimes e_2 -\bigl(q^2+1\bigr)e_6 \otimes e_3\bigr\}.
\end{gather*}
Thus, applying a~sequence of $F_i$s to obtain, for example, $\left\langle e_1 \otimes e_2 \middle| \pi_W(\Delta(F_1)) \middle| e_1 \otimes e_1 \right\rangle = 1$
will also yield
$
 \langle e_2 \otimes e_1| \pi_W(\Delta(F_1))| e_1 \otimes e_1 \rangle = q^{-1}
$
because $e_1 \otimes e_2+qe_2 \otimes e_1$ is a~basis vector of $W$.

Now, note that $e_4 \otimes e_4$ does not appear in any other basis vector of $W$; therefore, since $e_4 \otimes e_4$ lies in the bottom right corner of Figure~\ref{tikz--BigCrystalGraph}, it is clear that in order to obtain a~nonzero value corresponding to $e_4 \otimes e_4$, we must follow the arrows and apply a~composition of the form
\[\langle e_4 \otimes e_4| \pi_W(\Delta(F_1) \cdot \Pi_i (\Delta(F_i)) \cdot \Delta(F_1))| e_1 \otimes e_1 \rangle\]
with, in particular,
\begin{gather*}
\langle e_4 \otimes e_4| \pi_W(\Delta(F_1)\Delta(F_1)\Delta(F_3)\Delta(F_3)\Delta(F_2)\Delta(F_2)\Delta(F_1) \Delta(F_1))| e_1 \otimes e_1 \rangle\\
\qquad = q^4 + 4q^2 + 6 + 4q^{-2} + q^{-4}.
\end{gather*}
Finally, since $W$ is a~subset of a~symmetric tensor, $W \otimes W$ will be as well. Therefore, we must again find a~composition $\Pi_i(\Delta (F_i))$ that produces exactly $e_4^{\otimes 4}$, and it will not be enough to produce any other vector. By the tensor product rule, this will yet again result in a~nonzero value to an expression only in the structure $
\big\langle e_4^{\otimes 4}| \pi_W(\Delta(F_1) \cdot \Pi_i (\Delta(F_i)) \cdot \Delta(F_1))| e_1^{\otimes 4}\big\rangle$.
\end{proof}

\subsubsection{Finding a~Markov process generator} \label{section--FindingAlgebraGenerator}

By Lemma~\ref{lemma--groundstate}, the procedure for algebraically creating a~Markov process generator described in Appendix~\ref{subsection-groundstatetransformation} will produce an incomplete, or non-canonical, transformation. In this section, we show that such an incomplete transformation can be completed to produce a~Markov generator (with some arbitrary choices).

\begin{Definition} \label{definition--PartialTypeAGST}
 Fix nonnegative integers $k_1$, $k_2$, $k_3$. Let $K \subset W \otimes W$ be the subspace spanned by basis vectors $u_i \otimes u_j$ for which there exists $\sigma \in S_3$ such that
 \begin{align*}
 \big\langle u_i \otimes u_j| \Delta(F_{\sigma(1)})^{k_{\sigma(1)}}\Delta(F_{\sigma(2)})^{k_{\sigma(2)}}\Delta(F_{\sigma(3)})^{k_{\sigma(3)}}| e_1 \otimes e_1 \otimes e_1 \otimes e_1 \big\rangle \ne 0.
 \end{align*}

 Define $a \coloneqq q^{12} + q^2 + q^{-2} + q^{-12}$. We define a~partial Type~A ground state transformation of~${\pi_{W \otimes W}(C)}$ to be a~map \smash{$\pi_{W \otimes W}(C) \mapsto \bigl( a^{-1}G^{-1}\pi_{W \otimes W}(C)G - \mathrm{Id} \bigr) \big| _K$},
 where ${G \in \textrm{End}(W \otimes W)}$ is a~diagonal matrix (with respect to the basis $\{u_i \otimes u_j\mid 1 \le i,\allowbreak j \le 20\}$) such that, for any~${u_i \otimes u_j \in K}$, the diagonal entry of $G$ corresponding to $u_i \otimes u_j$ is equal to
 \begin{align*}
 \langle u_i \otimes u_j| \Delta(F_{\sigma(1)})^{k_{\sigma(1)}}\Delta(F_{\sigma(2)})^{k_{\sigma(2)}}\Delta(F_{\sigma(3)})^{k_{\sigma(3)}}| e_1 \otimes e_1 \otimes e_1 \otimes e_1 \rangle.
 \end{align*}
 Note that for any $u_i \otimes u_j$, there is at most one $\sigma$ such that the expression above is nonzero, so~$G|_K$ is unique for fixed $k_1$, $k_2$, $k_3$.

 We call $G$ a~partial Type A~ground state transformation matrix.
\end{Definition}
Note by Lemma~\ref{lemma--groundstate} that $K \ne W$ for any choice of $k_1$, $k_2$, $k_3$. Note that $G$ also depends on this choice.

\begin{Lemma}
\label{lemma--GSTeigenvalue}
 Let $G$ be a~ground state transformation matrix for $\pi_{W \otimes W}(C)$, i.e., a~diagonal matrix such that each row of $a^{-1}G^{-1}\pi_{W \otimes W}(C)G - \mathrm{Id}$ sums to $0$, and let $g_i$ denote the $i$-th diagonal entry of $G$ for $1 \le i \le 400$. Then the vector
 \begin{align*}
 \vec{g} \coloneqq
 \begin{pmatrix}
 g_1 \\
 g_2 \\
 \vdots \\
 g_{400}
 \end{pmatrix}
 \end{align*}
 satisfies $\pi_{W \otimes W}(C) \vec{g} = a\vec{g}$. This condition is both sufficient and necessary, i.e., if $\vec{g}$ is an eigenvector of $\pi_{W \otimes W}(G)$ with eigenvalue $a$, the diagonal matrix $G$ with entries corresponding to~$\vec{g}$ as above is such that the rows of $a^{-1}G^{-1}\pi_{W \otimes W}G - \mathrm{Id}$ sum to $0$.
\end{Lemma}

\begin{proof}
 Define the matrix $L$ by $ L \coloneqq a^{-1}G^{-1}\pi_{W \otimes W}(C)G - \mathrm{Id}$.
 By definition, each row of $L$ sums to $0$. Then we have
 \begin{gather*}
 G^{-1}(\pi_{W \otimes W}(C))G - a\,\mathrm{Id} = aL ,\qquad
 G^{-1}(\pi_{W \otimes W}(C))G = aL + a\,\mathrm{Id} \\
 \pi_{W \otimes W}(C)G = aGL + aG.
 \end{gather*}
 Since $G$ is diagonal and each row of $L$ sums to $0$, each row of $aGL$ sums to $0$. Thus, the sum of the elements in each row of $\pi_{W \otimes W}(C) G$ equals the sum of the elements in the same row of $aG$. Defining $g_i$ as above, we can write the sum of the $i$-th row of $\pi_{W \otimes W}(C)G$ as $\sum_{j=1}^{400} \pi_{W \otimes W}(C)_{i,j} g_j$, and the sum of the $i$-th row of $aG$ as $ag_i$. Thus, for all $1 \le i \le 400$,
 \begin{align*}
 \sum_{j=1}^{400} \pi_{W \otimes W}(C)_{i,j} g_j &= ag_i.
 \end{align*}
 This is equivalent to saying that $\pi_{W \otimes W}(C) \vec{g} = a\vec{g}$.

 Reversing the steps above shows that the condition $\pi_{W \otimes W}(C)\vec{g} = a\vec{g}$ is sufficient as well as necessary.
\end{proof}

Thus, any partial Type A~ground state transformation matrix $G$ of $\pi_{W \otimes W}(C)$ which satisfies the condition in Lemma~\ref{lemma--GSTeigenvalue} is a~ground state transformation matrix of $\pi_{W \otimes W}(C)$. Finding diagonal entries of $G$ as described in Definition~\ref{definition--PartialTypeAGST} and solving the condition in Lemma~\ref{lemma--GSTeigenvalue} as a~system of linear equations yields a~solution space $\mathcal{S}$ with $37$ unknowns.

\begin{Lemma}
 Given any $G \in \mathcal{S}$, $a^{-1}G^{-1}HG - \mathrm{Id}$ is a~generator for a~Markov process.
\end{Lemma}
\begin{proof}
 Immediate by the definition of $\mathcal{S}$.
\end{proof}
We give one such generator in Appendix~\ref{subsection-MarkovGen} by setting all unknowns to $0$ for ease of computation. One detail to note is that this generator has dimension $196 = 14^2$, which suggests that, when interpreted as an interacting particle system, the process allows for $14$ different states at each site.

\subsection{Difference between generators}
Note that the $2 \times 2$ and $3 \times 3$ blocks of the algebraically-produced Markov generator in Appendix~\ref{subsection-MarkovGen} are a~constant multiple of the $2 \times 2$ and $3 \times 3$ blocks of the probabilistically-produced Markov generator. However, we notice that any generator matrix given by a~ground state transformation as described in Definition~\ref{defn-groundstate} will not produce a~generator matrix for a~Markov process matching the probabilistically-produced generator matrix. We now prove Theorem~\ref{thm-bigalgebra}.

\begin{proof}
 We will show that no block of any ground state transformation of $\pi_{W \otimes W}(C)$ can be a~constant multiple of the $4\times 4$ block of $L^D_Q$ derived in Section~\ref{fusion-gen} corresponding to the communicating class $( \langle 3, 0\rangle, \langle1, 2\rangle, \langle2, 1\rangle, \langle0, 3\rangle)$.
 Thus, $L^D_Q$ must be different from any Markov generator produced from $\pi_{W \otimes W}(C)$.

 Note that the above communicating class contains one class 1 particle and one class 2 particle per Appendix~\ref{subsection-commClasses}. As proved in~\cite[Proposition 2.5]{kuan2020interactingparticlesystemstype}, this communicating class must correspond to the weight $2 \lambda_i + \lambda_j +\lambda_k$
 found in Section~\ref{subsect-weightspacedecomp}. As discussed in Lemma~\ref{WeightSpaceLemma}, there are $24$ weight spaces of dimension $4$ associated with $2 \lambda_i + \lambda_j +\lambda_k$, and thus $2 \lambda_i + \lambda_j +\lambda_k$ corresponds to the twenty-four $4 \times 4$ blocks in $\pi_{W \otimes W}$. It is therefore sufficient to show that none of the $4 \times 4$ blocks of $\pi_{W \otimes W}(C)$, which are listed in Appendix~\ref{4x4blocks}, can be mapped to the $4 \times 4$ block in $L^D_Q$ corresponding to the above communicating class.

 By referencing Section~\ref{DiagonalBlockOrderingProof}, we notice that each $4 \times 4$ block of the stochastic fusion generator does not contain any zero entries. However, every $4 \times 4$ block listed in Appendix~\ref{4x4blocks} has an off-diagonal entry which takes the value of $0$.

 Since conjugation by a~diagonal matrix will not change a~$0$-entry in a~matrix and the $0$-entries are off-diagonal, Definition~\ref{defn-groundstate} confirms that no ground state transformation of $\pi_{W \otimes W}(C)$ will result in $L^D_Q$.
\end{proof}

The computationally-heavy verification which involves disregarding the relationship between weights and particles and instead checking $4 \times 4$ blocks after trying all possible ground state transformations on $\pi_{W \otimes W}(C)$ will yield the same result.

\section{Further probabilistic results} \label{section--ProbProofs}

The construction of the process of stochastically fused Type D ASEP on two $\gamma$-sites is enough to notice fundamental differences with the algebraically constructed particle system. In this section, we outline further properties of stochastically fused Type D ASEP for the convenience of future investigations of this particle system. In particular, we explore drift limits of the generator, spectral gaps, reversible measures, and Markov duality.

\subsection{Properties of the generator}
We show in the following lemma that the generator matrix $L_m$ can be expressed as a~permutation of \smash{$\mathrm{Id}_{16 \times 16} \otimes L_p^{(2)}$}.

\begin{Lemma}
 \label{MiddleRearrange}
 There exists a~permutation matrix $J$ such that $L_m = J\bigl(\mathrm{Id}_{16 \times 16} \otimes L_p^{(2)}\bigr)J^{-1}$.
\end{Lemma}

\begin{proof}
 \label{MiddleRearrangeProof}
 For a~short proof by existence, consider the diagonal ordering of the possible states of four $\Gamma$-lattice sites as listed in Appendix~\ref{PermutationsAppendix}, and rearrange $L_m$ according to these states to get~\smash{$\textrm{Id}_{16 \times 16} \otimes L_p^{(2)}$}.

 For a~more intuitive proof, note that any state in the Type D ASEP on four $\Gamma$-lattice sites can be represented as $(x_1,x_2,x_3,x_4)$ where $x_1,x_2,x_3,x_4 \in \{0,1,2,3\}$. Letting only the middle $\Gamma$-sites interact corresponds to $x_1$ and $x_4$ being held constant, while the interaction of $x_2$ and $x_3$ is governed by \smash{$L_p^{(2)}$}. Now consider the bijective function $h\colon(x_1,x_2,x_3,x_4) \mapsto ((x_1,x_4),(x_2,x_3))$. Since $(x_1,x_4)$ is held constant while $(x_2,x_3)$ interact, it follows that $L_m$ can be permuted into~\smash{$\bigoplus_{i=1}^{16} L_p^{(2)}$}, since there are 16 possibilities for $(x_1,x_4)$. Finally, it follows that~\smash{$\bigoplus_{i=1}^{16} L_p^{(2)} = \textrm{Id}_{16 \times 16} \otimes L_p^{(2)}$}. This proves that $L_m$ can be permuted into \smash{$\textrm{Id}_{16 \times 16} \otimes L_p^{(2)}$}.
\end{proof}

\subsection{Taking drift speed to infinity}

In the following theorem, we examine the limiting behavior of stochastically fused Type D ASEP with respect to drift speed. Recall that $n$ affects drift speed, with larger values of $n$ causing faster drift in the Type D ASEP.

\begin{Theorem}\label{finiteLimitTheorem}
The following are finite: \smash{$\lim_{n \to \infty} \bigl(q^{-2n} L_Q\bigr)$} for $q > 1$, and \smash{$\lim_{n \to \infty} \bigl(q^{2n} L_Q\bigr)$} for $0 < q < 1$.
\end{Theorem}

The theorem above makes the intuitive statement that, as the drift speed $n$ of the Type D ASEP on four $\Gamma$-lattice sites increases, so does the drift speed of the stochastically fused process on two $\gamma$-sites. Now we list two lemmas necessary to prove the theorem.

First, we prove that the limit of drift speed commutes with the products that produce $L_Q$. In particular, the limit of $L_Q$ depends only on that of $L_m$.

\begin{Lemma}
\label{commutativityProposition}
The limit as $n \to \infty$ is commutative on $L_Q$. Precisely, \[\lim_{n \to \infty} L_Q = \lim_{n \to \infty} (\Lambda L_m \Phi) = \Lambda (\lim_{n \to \infty} L_m) \Phi.\]
\end{Lemma}

\begin{proof} \label{commutativity_proof}
We remind the reader that $\Lambda$ and $\Phi$ do not depend on $n$. Therefore,
\[\lim_{n \to \infty} L_Q = \lim_{n \to \infty} (\Lambda L_m \Phi) = \Lambda (\lim_{n \to \infty} L_m) \Phi.\tag*{\qed}
\]\renewcommand{\qed}{}
\end{proof}

Now we must show that the limit of $L_m$ is finite.

\begin{Lemma}
\label{finiteProposition}
The following are finite: $\lim_{n \to \infty} \bigl(q^{-2n} L_m\bigr)$ for $q >1$ and $\lim_{n \to \infty} \bigl(q^{2n} L_m\bigr)$ for $0 < q < 1$.
\end{Lemma}

\begin{proof}
\label{finite_proof}
Take $q >1$. Using the notation of Section~\ref{TypeDASEPIntro}, we have that
\[
q^{-2n} L_1 =\left[\begin{smallmatrix} * & \frac{(- q^{1 - 2n} + q^{- 1})^{2}}{q^{2}} & - q^{- 4} + q^{ - 2} + \frac{2}{q^{2+2n}} & - q^{- 4} + q^{ - 2} + \frac{2}{q^{2+2n}}
\\1 - q^{ - 2} + 2q^{-2n} & * & - q^{2 - 4 n} + 2q^{-2n} + q^{- 4 n} & \left(- q^{1 - 2n} + q^{- 1}\right)^{2}
\\ q^{2} \left(- q^{1 - 2n} + q^{- 1}\right)^{2} & 2 q^{2-2n} + q^{2 - 4 n} - q^{4 - 4 n} & * & 2 q^{2-2n} + q^{2 - 4 n} - q^{4 - 4 n}
\\ 1 - q^{- 2} + 2q^{-2n} & \left(- q^{1 - 2n} + q^{- 1}\right)^{2} & - q^{2 - 4 n} + 2q^{-2n} + q^{- 4 n} & *
\end{smallmatrix}\right],
\]
so
\[\begin{aligned}
 \lim_{n \to \infty} q^{-2n} L_1= \begin{bmatrix} * & q^{-4} & - q^{- 4} + q^{ - 2} & - q^{- 4} + q^{ - 2}
\\1 - q^{ - 2} & * & 0 & q^{-2}
\\ 1 & 0 & * & 0
\\ 1 - q^{- 2} & q^{-2} & 0 & *
\end{bmatrix}
\end{aligned}.\]
Since the rows of generator matrices sum to zero, $*$ in row $i$ is the negation of the sum of the rest of the elements in row $i$, so it is finite. This is true for all rows $i$ of $L_1$.

Similarly,
\[\begin{aligned}
 q^{-2n}L_2
 &=\begin{bmatrix} * & \frac{q^{1 - 4 n} + q^{- 1}}{q}\\q \left(q^{1 - 4n} + q^{- 1}\right) & * \end{bmatrix}
\end{aligned}, \]
so $
 \lim_{n \to \infty} q^{-2n} L_2 =\big[
\begin{smallmatrix} * & q^{-2} \\1 & * \end{smallmatrix}\big]$
and each $*$ is finite since the rows of a~generator matrix sum to zero.

Now by Section~\ref{TypeDASEPIntro}, we have that
\begin{align*}
 \lim_{n \to \infty} q^{-2n} L_p^{(2)} &=\lim_{n \to \infty} q^{-2n} \left(L_1 \oplus \bigoplus_{i=1}^4 L_2 \oplus \bigoplus_{i=1}^4 [0]\right) \\
 &=
\lim_{n \to \infty} \left(q^{-2n} L_1 \oplus \bigoplus_{i=1}^4 q^{-2n}L_2 \oplus \bigoplus_{i=1}^4 [0]\right) \\
&= \lim_{n \to \infty} \bigl(q^{-2n} L_1\bigr) \oplus \bigoplus_{i=1}^4 \lim_{n \to \infty} \bigl(q^{-2n}L_2\bigr) \oplus \bigoplus_{i=1}^4 [0],
\end{align*}
which is finite by the computations above. A~similar argument shows that \smash{$\lim_{n \to \infty} q^{2n} L^{(2)}_{p}$} is finite for $0 < q< 1$.

Finally, using Lemma~\ref{MiddleRearrange}, we obtain
\[q^{\pm 2n} L_m = q^{\pm 2n}J\bigl(\textrm{Id}_{16 \times 16} \otimes L^{(2)}_p\bigr)J^{-1} = J\bigl(\textrm{Id}_{16 \times 16} \otimes \bigl(q^{\pm 2n} L^{(2)}_p\bigr)\bigr)J^{-1}\] for respective boundaries on $q$, by the properties of the Kronecker product and standard matrix multiplication. Noting that $J$ and $J^{-1}$ do not depend on $n$, we take the limit as $n \to \infty$ of both sides and see that \[\lim_{n \to \infty} \bigl(q^{\pm 2n} L_m\bigr) = \lim_{n \to \infty} \bigl(J\bigl(\textrm{Id}_{16 \times 16} \otimes \bigl(q^{\pm 2n} L^{(2)}_p\bigr)\bigr)J^{-1}\bigr) = J\bigl(\textrm{Id}_{16 \times 16} \otimes \bigl(\lim_{n \to \infty} q^{\pm 2n} L^{(2)}_p\bigr)\bigr)J^{-1},\] where the right (and therefore left) side of the equation is finite since $\lim_{n \to \infty} \bigl(q^{\pm 2n} L^{(2)}_p\bigr)$ is finite for respective boundaries on $q$. This implies that the rate matrix grows on the order of $q^{2n}$ irrespective of drift direction.
\end{proof}

Now we can prove the primary theorem of this section.

\begin{proof}[Proof of Theorem~\ref{finiteLimitTheorem}]
\label{finite_theorem_proof}
Note that Lemma~\ref{finiteProposition} suggests that multiplying $q^{\pm 2n} L_m$ (for respective boundaries on $q$) by any constant preserves the finite limit, since the constant can be moved outside the limit. Note that
\[
q^{\pm 2n} L_Q = q^{\pm 2n} \Lambda L_m \Phi = \Lambda q^{\pm 2n} L_m \Phi
\]
 by the commutativity of scalar multiplication. Taking the limit of both sides of the equation, we have
\[
\lim_{n \to \infty} \bigl(q^{\pm 2n} L_Q\bigr) = \lim_{n \to \infty} \bigl(\Lambda q^{\pm 2n} L_m \Phi\bigr) = \Lambda \lim_{n \to \infty} \bigl(q^{\pm 2n} L_m\bigr) \Phi
\]
 by Lemma~\ref{commutativityProposition}. We know that $\lim_{n \to \infty} \bigl(q^{\pm 2n} \!L_m\bigr)$ is finite for respective values of $q$, so $\lim_{n \to \infty}\! q^{\pm 2n}\allowbreak\times L_Q$ must also be finite for such values of $q$.
\end{proof}

\subsection{Spectral gaps}

Recall that the spectral gap of a~Markov process is the absolute value of the second-largest (i.e., least negative) eigenvalue of the generator matrix, where all eigenvalues of generator matrices are less than or equal to zero. From our calculations, the spectral gap corresponding to the generator $\mathcal{L}_2$ simplifies to a~reasonable expression, which we state below. However, the spectral gaps for $\mathcal{L}_3$, $\mathcal{L}_4$, $\mathcal{L}_6$, and $\mathcal{L}_9$ do not simplify to expressions of a~reasonable length.

\begin{Proposition}
 \label{SpectralStatement}
 The spectral gap for the $\mathcal{L}_2$ communicating class, denoted as $\lvert \lambda_{\mathcal{L}_2} \rvert$ is
 \[\lvert \lambda_{\mathcal{L}_2} \rvert = \frac{\bigl(q^{4n-2}+1\bigr)\bigl(q^4+1\bigr)}{q^{2n}\bigl(q^2+1\bigr)}.\]
\end{Proposition}
\begin{proof}
\label{SpectralStatement_proof}
We use the code in Appendix~\ref{section--Appendix} to directly calculate the eigenvalues for $\mathcal{L}_2$. We then take the absolute value of the second largest eigenvalue of each generator to obtain the spectral gap for $\mathcal{L}_2$.
\end{proof}

For $q > 1$ and $0 < q < 1$, notice that the spectral gap $|\lambda_{\mathcal{L}_2}|$ goes to infinity as $n$ goes to infinity, meaning the relaxation time $1/\lambda_{\mathcal{L}_2}$ goes to zero. This implies that as the drift speed of the Type D ASEP goes to infinity, the time to convergence to the stationary distribution approaches zero for the 2-state communicating classes. The authors suspect the same is true for the 3, 4, 6, and 9-state communicating classes.

Also observe that the eigenvalues of $q^{\pm 2n} \mathcal{L}_2$ are those of $\mathcal{L}_2$ multiplied by $q^{\pm 2n}$. Thus, the limit as $n$ goes to infinity of the spectral gap of $q^{-2n} \mathcal{L}_2$ for $q > 1$ is 1, so its relaxation time is 1. Moreover, the limit as $n$ goes to infinity of the spectral gap of $q^{2n} \mathcal{L}_2$ for $0 < q < 1$ is~${\bigl(q^4 + 1\bigr)/\bigl(q^2 + 1\bigr)}$, so its relaxation time is $\bigl(q^2 + 1\bigr)/\bigl(q^4 + 1\bigr)$. These results suggest that the time to convergence to the stationary distribution increases (relative to the time to convergence in the previous paragraph) as $n$ goes to infinity for the 2-state communicating class. The authors suspect the same is true for the 3, 4, 6, and 9-state communicating classes.

\begin{Remark}
As $n$ goes to infinity, some of the jump rates in the (time-rescaled) fused Type D ASEP converge to $0$, which can be seen by taking the limit as $n$ goes to infinity of entries in the fusion generator. Thus, as $n$ approaches infinity, the fused Type D ASEP (where $\delta = 0$ and time is rescaled by a~factor of $q^{\pm 2n}$ for respective values of $q$) degenerates to the usual ASEP.
\end{Remark}

\subsection{Reversible measures} \label{subsection--StationaryDistr}

First, notice that in our particle system the non-normalized versions of the stationary distributions are reversible measures for the communicating classes. This is intuitive since the fused Type D ASEP is positive recurrent and irreducible on its communicating classes, so the stationary distributions for the generators on its communicating classes are left eigenvectors of those generators. As the left eigenvectors corresponding to eigenvalue 0 are nonzero across all communicating classes, they are candidates for reversible measures. One can verify that these eigenvectors, evaluated at $n=2$, are reversible measures for arbitrary $n$.

Now we state the reversible measures explicitly, using the communicating classes listed in Appendix~\ref{subsection-commClasses} and $q$-deformed integer notation \smash{$[n]_q = \frac{q^n - q^{-n}}{q-q^{-1}}$}.

\begin{Lemma}\label{allStationary} The following is true for all $q > 0$, where $q \neq 1$.
For the $9$-state communicating class, a~reversible measure is
\[q^8 [2]^4_q \left[1, \frac{q^{-8}}{[2]^4_q}, \frac{q^{8}}{[2]^4_q}, \frac{q^{-4}}{[2]^2_q}, \frac{q^{4}}{[2]^2_q}, \frac{q^{-4}}{[2]^2_q}, \frac{q^{4}}{[2]^2_q}, \frac{1}{[2]^4_q}, \frac{1}{[2]^4_q}\right].\]

For the $6$-state communicating class, a~reversible measure is
\[q^6 [2]^2_q \left[q^{-2}, q^2, \frac{q^{-6}}{[2]^2_q},\frac{q^{6}}{[2]^2_q}, \frac{q^{-2}}{[2]^2_q},\frac{q^{2}}{[2]^2_q} \right].\]

For the $4$-state communicating class, a~reversible measure is \smash{$q^4 \left[q^4, 1, 1, q^{-4} \right]$}.

For the $3$-state communicating class, a~reversible measure is
\[q^4 [2]^2_q \left[1, \frac{q^{-4}}{[2]^2_q},\frac{q^{4}}{[2]^2_q} \right].\]

For the $2$-state communicating class, a~reversible measure is $q^2 \left[q^{-2}, q^2\right]$.
\end{Lemma}

Now we generalize our reversible measures to arbitrary system sizes. In the following proposition, $\langle a_i, a_j\rangle$ denotes a~particle configuration on two $\gamma$-lattice sites, where $a_i, a_j \in \{0,1,2,3\}$ denote the particle configuration at site $i$ and $j$, respectively.

\begin{Proposition}
\label{generalizeStationary}
There exists a~reversible measure $\pi^{(N)}$ for the generator $L$ on $N$ $\gamma$-sites such that \smash{$\pi^{(N)}(\langle a_1, \dots, a_N \rangle) = \prod_{i < j} \pi(\langle a_i, a_j \rangle)$},
where each $\pi$ is a~reversible measure from Lemma~{\rm\ref{allStationary}}.
\end{Proposition}

\begin{proof}\label{generalizeStationaryProof}
Notice that the interacting particle system does not differentiate amongst the two classes of particles, so we can use occupation variables to write a~particle configuration $A$ as $(A_1, A_2)$ where $A_i= \langle a_{i,1}, a_{i,2} \rangle$ such that $a_{i,j} = i$ if there is a~particle of class $i$ at the $j$-th site and~${a_{i,j}=0}$ otherwise. We use the $\langle ~\rangle$ notation because the position of class 1 particles can be viewed as a~particle configuration on two $\gamma$-sites which contain no class 2 particles. Then $\pi(A_1)$ is a~reversible measure from Lemma~\ref{allStationary} for a~configuration on two $\gamma$-sites with no class 2 particles and with class 1 particles in the positions given by $A$. We can interpret $\pi(A_2)$ in a~similar fashion.

By casework, one sees that $\pi(A) = \pi(A_1)\pi(A_2)$. We compute some cases below, and the rest follow similarly.

\textbf{Case 1:} Let $A = \langle 3,1 \rangle$, so $A_1 = \langle 1,1 \rangle$ and $A_2 = \langle 2,0 \rangle$. Then
\begin{align*}
 \frac{ q^6\bigl(q^2+1\bigr)^2}{q^{12} + q^8 + q^2\bigl(q^2+1\bigr)^2+q^6\bigl(q^2+1\bigr)^2+ q^4 +1} &= \pi(\langle 3,1 \rangle)= \pi(\langle 1,1 \rangle) \pi(\langle 2,0 \rangle) \\
 &= \frac{q^2\bigl(q^2+1\bigr)^2}{q^8 + q^2\bigl(q^2+1\bigr)^2 + 1} \times \frac{q^4}{q^4 + 1}.
\end{align*}

\textbf{Case 2:} Let $A = \langle 3,3 \rangle$, so $A_1 = \langle 1,1 \rangle$ and $A_2 = \langle 2,2\rangle$. Then
\[q^8 \left(\frac{q^2 - q^{-2}}{q-q^{-1}} \right)^4 = \pi(\langle 3,3 \rangle) = \pi(\langle 1,1 \rangle) \pi(\langle 2,2\rangle) = q^4\left(\frac{q^2 - q^{-2}}{q-q^{-1}} \right)^2 \times q^4\left(\frac{q^2 - q^{-2}}{q-q^{-1}} \right)^2.\]

\textbf{Case 3:} Let $A = \langle 31,32\rangle$, so $A_1 = \langle 11,1\rangle$ and $A_2 = \langle 2,22\rangle$. Then
\[q^4 = \pi(\langle 31,32\rangle) = \pi(\langle 11,1\rangle) \pi(\langle 2,22\rangle) = q^4 \times 1.\]

\textbf{Case 4:} Let $A = \langle 3,0\rangle$, so $A_1 = \langle 1,0\rangle$ and $A_2 = \langle 2,0 \rangle$. Then
\[q^8 = \pi(\langle 3,0\rangle) = \pi(\langle 1,0\rangle) \pi(\langle 2,0 \rangle) = q^4 \times q^4.\]

Now, for $N$ sites, let \smash{$\pi^{(N)}(\langle a_1, \dots, a_N \rangle) = \prod_{i < j} \pi(\langle a_i, a_j \rangle)$}.

Then by~\cite[Theorem 2.1]{kuan2024qexchangeablemeasurestransformationsinteracting}, $\pi^{(N)}$ must be a~reversible measure for generator $L$ on $N$ sites. Furthermore, it must be the product measure in (16) of~\cite{Carinci2014} for $j = 1$ and some $\alpha$.
\end{proof}

Lemma~\ref{allStationary} and Proposition~\ref{generalizeStationary} emphasize that, although the jump rates of particles in the Type D ASEP depend on drift speed $n$, the reversible measures do not.

\subsection{Markov duality}\label{subsection--MarkovDuality}

\begin{Definition}
Let $X_t$ be a~time-homogeneous Markov process on a~discrete state space $\mathcal{X}$. Label the generator matrix of the Markov process as $L_X$. Let $D$ be a~matrix with rows and columns indexed by $\mathcal{X}$. Then $X_t$ is self-dual with respect to the matrix $D$ if $L_X D = D L_X^{\mathsf{T}}$
\cite{10.1214/aop/1176990628}.
\end{Definition}

The first possible approach to duality would be through $q$-Krawtchouk polynomials. In~\cite[Theorem 3.1]{Blyschak_2023}, the authors propose a~self-duality function for the non-fused Type D ASEP with parameters $(q, n = 2, \delta = 0)$ and $(q, n = 3, \delta = 0)$. Using the parameters $L = 2$ and $n = 2$, we find a~potential self-duality function (given by \cite[Theorem~3.1]{Blyschak_2023}) for this two-site Type D ASEP, where $D$ is the matrix of this proposed self-duality function. However, using the matrix~$D$ that we computed, \smash{$L_Q^D D[1,2] \neq D\bigl(L_Q^D\bigr)^{\mathsf{T}}[1,2]$}, so \smash{$L_Q^D D \neq D\bigl(L_Q^D\bigr)^{\mathsf{T}}$} and thus $D$ is not the matrix of the self-duality function of the Type D ASEP with two $\gamma$-lattice sites. In our case (unlike in~\cite{Blyschak_2023}),
constructing a~particle system via probabilistic and algebraic processes results in different duality functions, so we need to consider different methods of constructing duals.

We were able to construct a~nontrivial dual for $L_Q^D$ when $n=2$ using the eigenvectors of the operator and a~set of (arbitrary) nonzero complex weights $\alpha_i$. We let $n=2$ because the eigenvalues of $L_Q^D$ (which are the eigenvalues of the block generators $\mathcal{L}_i$, where $i$ ranges over communicating classes) simplify nicely in this case. Furthermore, we extend this dual to one on~$L_Q$ using a~symmetry $S$ that commutes with $L_p$ and works nicely with fission and fusion.

Recall the fact that the generator of a~reversible Markov process is diagonalizable. Now let $\mathcal{A}_i$ be the diagonal matrix of the eigenvalues of $\mathcal{L}_i$ such that $\mathcal{L}_i = \mathcal{P}_i \mathcal{A}_i \mathcal{P}^{-1}_i$. Define $\alpha \coloneqq \{\alpha_i | \alpha_i \in \mathbb{C} \setminus \{0\} \}$. Also define
\[\mathcal{P}_\alpha := \bigoplus_{i} \alpha_i \mathcal{P}_i,\qquad\mathcal{A} := \bigoplus_{i} \mathcal{A}_i.\]

Observe that, by definition, $L_Q^D$ equals $\mathcal{P}_\alpha \mathcal{A}\mathcal{P}_\alpha^{-1} \oplus \mathcal{Z}$, where $\mathcal{Z} := \bigoplus_{i=1}^4 [0]$.

We can now relate $\mathcal{P}_\alpha$ to our operator defining Markov self-duality.

\begin{Theorem}
\label{DiagonalStrategy}
The operator $\mathcal{D}_\alpha := \mathcal{P}_\alpha \mathcal{P}_{\alpha}^{\mathsf{T}} \oplus \mathcal{Z}$ defines a~nontrivial dual for $L_Q^D$, in the sense that \smash{$L_Q^D \mathcal{D}_\alpha = \mathcal{D}_\alpha \bigl(L_Q^D\bigr)^{\mathsf{T}}$}.
\end{Theorem}
\begin{proof} \label{DiagonalStrategyProof}

This is true because
 \begin{align*}
 L_Q^D \mathcal{D}_\alpha &= \bigl(\mathcal{P}_\alpha \mathcal{A} \mathcal{P}_\alpha^{-1} \oplus \mathcal{Z}\bigr)\bigl(\mathcal{P}_\alpha \mathcal{P}_\alpha^{\mathsf{T}} \oplus \mathcal{Z}\bigr)
 = \bigl(\mathcal{P}_\alpha \mathcal{A} \mathcal{P}_\alpha^{-1} \mathcal{P}_\alpha \mathcal{P}_\alpha^{\mathsf{T}}\bigr) \oplus \mathcal{Z}
 = \mathcal{P}_\alpha \mathcal{A} \mathcal{P}_\alpha^{\mathsf{T}} \oplus \mathcal{Z}\\
 &= \mathcal{P}_\alpha \bigl(\mathcal{P}_\alpha^{\mathsf{T}} \bigl(\mathcal{P}_\alpha^{\mathsf{T}}\bigr)^{-1}\bigr)\mathcal{A} \mathcal{P}_\alpha^{\mathsf{T}} \oplus \mathcal{Z} = \bigl(\mathcal{P}_\alpha \mathcal{P}_\alpha^{\mathsf{T}}\bigr)\bigl(\bigl(\mathcal{P}_\alpha^{-1}\bigr)^{T}\mathcal{A} \mathcal{P}_\alpha^{\mathsf{T}} \bigr) \oplus \mathcal{Z}
 \\
 & = \bigl(\mathcal{P}_\alpha \mathcal{P}_\alpha^{\mathsf{T}} \bigr)\bigl(\mathcal{P}_\alpha \mathcal{A} \mathcal{P}_\alpha^{-1} \bigr)^{\mathsf{T}} \oplus \mathcal{Z}
 = \bigl(\mathcal{P}_\alpha \mathcal{P}_\alpha^{\mathsf{T}} \oplus \mathcal{Z} \bigr)\bigl(\bigl(\mathcal{P}_\alpha \mathcal{A} \mathcal{P}_\alpha^{-1}\bigr)^{\mathsf{T}} \oplus \mathcal{Z}\bigr) \\
 &
 = \bigl(\mathcal{P}_\alpha \mathcal{P}_\alpha^{\mathsf{T}} \oplus \mathcal{Z}\bigr)\bigl(\bigl(\mathcal{P}_\alpha \mathcal{A} \mathcal{P}_\alpha^{-1}\bigr) \oplus \mathcal{Z}\bigr)^{\mathsf{T}}
 = \mathcal{D}_\alpha \bigl(L_Q^D\bigr)^{\mathsf{T}}.\tag*{\qed}
 \end{align*}\renewcommand{\qed}{}
\end{proof}

Note that since $\mathcal{P}_\alpha$ is a~block diagonal matrix, so is $\mathcal{P}_\alpha \mathcal{P}_\alpha^{\mathsf{T}}$, which implies that $\mathcal{D}_\alpha$ is also block diagonal.

Now we can extend the results above to provide a~duality function for $L_Q$ when $n=2$.

\begin{Lemma} \label{permutation-dual} If $C$ is the permutation matrix described in Proposition~{\rm\ref{DiagonalBlockOrdering}}, then $C \mathcal{D}_\alpha C^{\mathsf{T}}$ is a~dual for $L_Q$.
\end{Lemma}
\begin{proof}
 By Proposition~\ref{DiagonalBlockOrdering}, we have that $C^{-1} L_Q C = L_Q^D$. By the definition of duality and by substitution,
 \smash{$
 C^{-1} L_Q C \mathcal{D}_\alpha = \mathcal{D}_\alpha \bigl(C^{-1} L_Q C\bigr)^{\mathsf{T}}$}.
 The right side of the equality is equivalent to~\smash{$\mathcal{D}_\alpha C^{\mathsf{T}} L_Q^{\mathsf{T}} \bigl(C^{\mathsf{T}}\bigr)^{-1}$}, so we can multiply both sides of the equality by $C$ and $C^{\mathsf{T}}$ to obtain \smash{$L_Q C \mathcal{D}_\alpha C^{\mathsf{T}} = C \mathcal{D}_\alpha C^{\mathsf{T}} L_Q^{\mathsf{T}}$}, which concludes the proof.
\end{proof}

Now we use symmetry and the result above to describe additional duals for $L_Q$.

\begin{Lemma} \label{commutator} If the matrix $S_Q$ satisfies $[S_Q, L_Q] =0$, then $S_Q \bigl(C \mathcal{D}_\alpha C^{\mathsf{T}}\bigr)$ is a~dual for $L_Q$.
\end{Lemma}
\begin{proof} We have by assumption that $S_Q L_Q = L_Q S_Q$. Since \smash{$L_Q \bigl(C \mathcal{D}_\alpha C^{\mathsf{T}}\bigr) = \bigl(C \mathcal{D}_\alpha C^{\mathsf{T}}\bigr) L_Q^{\mathsf{T}}$} by Lemma~\ref{permutation-dual}, we have by substitution that
\[L_Q \bigl(S_Q C \mathcal{D}_\alpha C^{\mathsf{T}}\bigr) = S_Q L_Q \bigl(C \mathcal{D}_\alpha C^{\mathsf{T}}\bigr) =\bigl( S_Q C \mathcal{D}_\alpha C^{\mathsf{T}}\bigr) L_Q^{\mathsf{T}}, \] which concludes the proof.
\end{proof}

Finally, we describe a~sufficient condition for a~dual of $L_Q$.

\begin{Theorem} \label{s-lemma} If the matrix $S$ satisfies $S L_p = L_p S$, $\Phi \Lambda S \Phi = S \Phi$, and $\Lambda S \Phi \Lambda = \Lambda S$, then $(\Lambda S \Phi) \bigl( C \mathcal{D}_\alpha C^{\mathsf{T}}\bigr)$ is a~dual for $L_Q$.
\end{Theorem}
\begin{proof} Define $S_Q := \Lambda S \Phi$. We want to show that $[S_Q, L_Q]=0$.

It follows from our assumptions and the definitions of $L_Q$ and $S_Q$ that \[L_Q S_Q = (\Lambda L_p \Phi) (\Lambda S \Phi) = \Lambda L_p S \Phi\] and \[S_Q L_Q = (\Lambda S \Phi)(\Lambda L_p \Phi) = \Lambda S L_p \Phi = \Lambda L_p S \Phi\] so $[S_Q, L_Q] = 0$.

Now Lemma~\ref{commutator} implies that $S_Q \bigl( C \mathcal{D}_\alpha C^{\mathsf{T}}\bigr) = (\Lambda S \Phi)\bigl( C \mathcal{D}_\alpha C^{\mathsf{T}}\bigr)$ is a~dual for $L_Q$.
\end{proof}

Next, we explain (intuitively, in words) why there are natural nontrivial examples of $S$ that can satisfy the conditions of Theorem~\ref{s-lemma}. In the context of the Type D ASEP, initial conditions of the form $v\Lambda$ are called $q$--exchangeable measures. In particular, if the interacting particle system is $q$--exchangeable on two lattice sites, then it is $q$--exchangeable on arbitrarily many lattice sites. See~\cite{kuan2024qexchangeablemeasurestransformationsinteracting} for more details. Certain symmetries $S$ will preserve $q$--exchangeable measures.

More explicitly, we can leverage previous methods to find such an $S$.
\begin{Proposition}
Given a~particle configuration $\eta$, we let $A_1(\eta)$ and $A_2(\eta)$ denote the occupation variables, which are the sites which class $1$ and class $2$ particles, respectively, occupy. Let~$\{x\}$ denote the particle configuration with a~single class $2$ particle at lattice site $x$. Then the operator $S$ defined by $GDG$, where
 \[D(\eta, \{x\}) = \mathds{1}_{x \in q A_2(\eta)} q^{2x - 2\overline{N}_x^2(\eta)}\]
 and $\overline{N}$ denotes the number of class $2$ particles in $\eta$ to the left of $x$, and
 \[
G(\eta,\eta)=\prod_{x_1 \in A_1(\eta)} q^{-x_1} \prod_{x_2 \in A_2(\eta)} q^{-x_2}
 \]
 satisfies the conditions of Theorem~{\rm\ref{s-lemma}}.
\end{Proposition}
\begin{proof}
 The operator $S$ has an algebraic definition from~\cite{kuan2020interactingparticlesystemstype}, where it was also explicitly computed to have the stated form. Using the algebraic form, the commutation $SL_p=L_pS$ holds. The other relations follow from direct computation.
\end{proof}

From here, the dual function in Theorem~\ref{s-lemma} can be written as an explicit series. Writing duality functions as series has a~precedence in orthogonal polynomial duality functions, but we do not explore this connection here.

The authors also note that the dual function described above can easily be generalized to $K$ lattice sites, where $K$ is an even positive integer. We provided the case of $K=4$ above to aid reader comprehension.

\appendix

\section{Appendix} \label{section--Appendix}
\subsection{Code for probabilistic sections}
The SymPy code used to compute the generator matrix, limits of component generator matrices, stationary distributions, spectral gaps, and matrices of Markov duality can be found at the authors' \href{https://github.com/evaengel/stochastic-fusion}{GitHub repository}. Since the largest communicating classes are of size 9, most calculations can be done by hand. We have omitted such calculations from the paper, and the SymPy code serves only as independent verification.

\subsection{List of communicating classes by size} \label{subsection-commClasses}
We list the communicating classes for the Type D ASEP on two $\gamma$-lattice sites below.

The ordered 9-state communicating class is
\[( \langle 3, 3 \rangle,\langle 0, 33 \rangle, \langle 33,0 \rangle,\langle 1,32 \rangle,\langle 32,1 \rangle,\langle 2,31 \rangle, \langle 31,2 \rangle,\langle 11, 22 \rangle,\langle 22, 11 \rangle ).\]

The ordered 6-state communicating classes are
\begin{gather*}
( \langle 1, 3\rangle, \langle3, 1\rangle, \langle0, 31\rangle, \langle31, 0\rangle, \langle2, 11\rangle, \langle11, 2 \rangle ),\\
( \langle 2, 3\rangle, \langle3, 2\rangle, \langle0, 32\rangle, \langle32, 0\rangle, \langle1, 22\rangle, \langle22, 1\rangle ),\\
( \langle 3, 31\rangle, \langle31, 3\rangle, \langle1, 33\rangle, \langle33, 1\rangle, \langle11, 32\rangle, \langle32, 11\rangle),\\
( \langle 3, 32\rangle, \langle32, 3\rangle, \langle2, 33\rangle, \langle33, 2\rangle, \langle22, 31\rangle, \langle31, 22\rangle).
\end{gather*}

The ordered 4-state communicating classes are
\begin{gather*}
( \langle 3, 0\rangle, \langle1, 2\rangle, \langle2, 1\rangle, \langle0, 3\rangle),\qquad
(\langle 3, 11\rangle, \langle1, 31\rangle, \langle11, 3\rangle, \langle31, 1\rangle),\\
(\langle 3, 22\rangle, \langle2, 32\rangle, \langle22, 3\rangle, \langle32, 2\rangle),\qquad
(\langle 33, 3\rangle, \langle31, 32\rangle, \langle32, 31\rangle, \langle3, 33\rangle).
\end{gather*}

The ordered 3-state communicating classes are
\begin{gather*}
(\langle 1, 1\rangle, \langle0, 11\rangle, \langle11, 0\rangle),\qquad
(\langle 2, 2\rangle, \langle0, 22\rangle, \langle22, 0\rangle),\\
(\langle 31, 31\rangle, \langle11, 33\rangle, \langle33, 11\rangle),\qquad
(\langle 32, 32\rangle, \langle22, 33\rangle, \langle33, 22\rangle).
\end{gather*}

The ordered 2-state communicating classes are
\begin{gather*}
(\langle 0, 1\rangle, \langle1, 0\rangle),\qquad
(\langle 0, 2\rangle, \langle2, 0\rangle),\qquad
(\langle 1, 11\rangle, \langle11, 1\rangle),\qquad
(\langle 2, 22\rangle, \langle22, 2\rangle),\\
(\langle 11, 31\rangle, \langle31, 11\rangle),\qquad
(\langle 22, 32\rangle, \langle32, 22\rangle),\qquad
(\langle 31, 33\rangle, \langle33, 31\rangle),\qquad
(\langle 32, 33\rangle, \langle33, 32\rangle).
\end{gather*}

The ordered 1-state communicating classes are
\begin{gather*}
(\langle 0, 0\rangle),\qquad
(\langle 11, 11\rangle),\qquad
(\langle 22, 22\rangle),\qquad
(\langle 33, 33\rangle).
\end{gather*}

\subsection[Ordering of states and permutations of L\_m and L\_q]{Ordering of states and permutations of $\boldsymbol{L_m}$ and $\boldsymbol{L_Q}$}
\label{PermutationsAppendix}
\subsubsection[Gamma state ordering]{$\boldsymbol{\Gamma}$ state ordering}
The assumed ordering of states for four $\Gamma$-lattice sites is $\Omega \otimes \Omega$, which we state below. To be more concise, we represent the state $(x_1,x_2,x_3,x_4)$ as $x_1x_2x_3x_4$.

Initial ordering of $\Omega \otimes \Omega$:\\
(3030, 3021, 3003, 3012, 3010, 3001, 3020, 3002, 3031, 3013, 3032, 3023, 3000, 3011, 3022, 3033, 2130, 2121, 2103, 2112, 2110, 2101, 2120, 2102, 2131, 2113, 2132, 2123, 2100, 2111, 2122, 2133, 0330, 0321, 0303, 0312, 0310, 0301, 0320, 0302, 0331, 0313, 0332, 0323, 0300, 0311, 0322, 0333, 1230, 1221, 1203, 1212, 1210, 1201, 1220, 1202, 1231, 1213, 1232, 1223, 1200, 1211, 1222, 1233, 1030, 1021, 1003, 1012, 1010, 1001, 1020, 1002, 1031, 1013, 1032, 1023, 1000, 1011, 1022, 1033, 0130, 0121, 0103, 0112, 0110, 0101, 0120, 0102, 0131, 0113, 0132, 0123, 0100, 0111, 0122, 0133, 2030, 2021, 2003, 2012, 2010, 2001, 2020, 2002, 2031, 2013, 2032, 2023, 2000, 2011, 2022, 2033, 0230, 0221, 0203, 0212, 0210, 0201, 0220, 0202, 0231, 0213, 0232, 0223, 0200, 0211, 0222, 0233, 3130, 3121, 3103, 3112, 3110, 3101, 3120, 3102, 3131, 3113, 3132, 3123, 3100, 3111, 3122, 3133, 1330, 1321, 1303, 1312, 1310, 1301, 1320, 1302, 1331, 1313, 1332, 1323, 1300, 1311, 1322, 1333, 3230, 3221, 3203, 3212, 3210, 3201, 3220, 3202, 3231, 3213, 3232, 3223, 3200, 3211, 3222, 3233, 2330, 2321, 2303, 2312, 2310, 2301, 2320, 2302, 2331, 2313, 2332, 2323, 2300, 2311, 2322, 2333, 0030, 0021, 0003, 0012, 0010, 0001, 0020, 0002, 0031, 0013, 0032, 0023, 0000, 0011, 0022, 0033, 1130, 1121, 1103, 1112, 1110, 1101, 1120, 1102, 1131, 1113, 1132, 1123, 1100, 1111, 1122, 1133, 2230, 2221, 2203, 2212, 2210, 2201, 2220, 2202, 2231, 2213, 2232, 2223, 2200, 2211, 2222, 2233, 3330, 3321, 3303, 3312, 3310, 3301, 3320, 3302, 3331, 3313, 3332, 3323, 3300, 3311, 3322, 3333).\looseness=1

Diagonal ordering of $\Omega \otimes \Omega$:\\
(3300, 3210, 3030, 3120, 3100, 3010, 3200, 3020, 3310, 3130, 3320, 3230, 3000, 3110, 3220, 3330, 2301, 2211, 2031, 2121, 2101, 2011, 2201, 2021, 2311, 2131, 2321, 2231, 2001, 2111, 2221, 2331, 0303, 0213, 0033, 0123, 0103, 0013, 0203, 0023, 0313, 0133, 0323, 0233, 0003, 0113, 0223, 0333, 1302, 1212, 1032, 1122, 1102, 1012, 1202, 1022, 1312, 1132, 1322, 1232, 1002, 1112, 1222, 1332, 1300, 1210, 1030, 1120, 1100, 1010, 1200, 1020, 1310, 1130, 1320, 1230, 1000, 1110, 1220, 1330, 0301, 0211, 0031, 0121, 0101, 0011, 0201, 0021, 0311, 0131, 0321, 0231, 0001, 0111, 0221, 0331, 2300, 2210, 2030, 2120, 2100, 2010, 2200, 2020, 2310, 2130, 2320, 2230, 2000, 2110, 2220, 2330, 0302, 0212, 0032, 0122, 0102, 0012, 0202, 0022, 0312, 0132, 0322, 0232, 0002, 0112, 0222, 0332, 3301, 3211, 3031, 3121, 3101, 3011, 3201, 3021, 3311, 3131, 3321, 3231, 3001, 3111, 3221, 3331, 1303, 1213, 1033, 1123, 1103, 1013, 1203, 1023, 1313, 1133, 1323, 1233, 1003, 1113, 1223, 1333, 3302, 3212, 3032, 3122, 3102, 3012, 3202, 3022, 3312, 3132, 3322, 3232, 3002, 3112, 3222, 3332, 2303, 2213, 2033, 2123, 2103, 2013, 2203, 2023, 2313, 2133, 2323, 2233, 2003, 2113, 2223, 2333, 0300, 0210, 0030, 0120, 0100, 0010, 0200, 0020, 0310, 0130, 0320, 0230, 0000, 0110, 0220, 0330, 1301, 1211, 1031, 1121, 1101, 1011, 1201, 1021, 1311, 1131, 1321, 1231, 1001, 1111, 1221, 1331, 2302, 2212, 2032, 2122, 2102, 2012, 2202, 2022, 2312, 2132, 2322, 2232, 2002, 2112, 2222, 2332, 3303, 3213, 3033, 3123, 3103, 3013, 3203, 3023, 3313, 3133, 3323, 3233, 3003, 3113, 3223, 3333).\looseness=1

\subsubsection[gamma state ordering]{$\boldsymbol{\gamma}$ state ordering}
The assumed ordering of states for two $\gamma$-lattice sites is $\chi \otimes \chi$, which we state below. To be more concise, we represent the state $\langle x_1,x_2 \rangle$ as $x_1\_x_2$.

Initial ordering of $\chi \otimes \chi$:\\
(0\_0, 0\_1, 0\_2, 0\_11, 0\_3, 0\_22, 0\_31, 0\_32, 0\_33, 1\_0, 1\_1, 1\_2, 1\_11, 1\_3, 1\_22, 1\_31, 1\_32, 1\_33, 2\_0, 2\_1, 2\_2, 2\_11, 2\_3, 2\_22, 2\_31, 2\_32, 2\_33, 11\_0, 11\_1, 11\_2, 11\_11, 11\_3, 11\_22, 11\_31, 11\_32, 11\_33, 3\_0, 3\_1, 3\_2, 3\_11, 3\_3, 3\_22, 3\_31, 3\_32, 3\_33, 22\_0, 22\_1, 22\_2, 22\_11, 22\_3, 22\_22, 22\_31, 22\_32, 22\_33, 31\_0, 31\_1, 31\_2, 31\_11, 31\_3, 31\_22, 31\_31, 31\_32, 31\_33, 32\_0, 32\_1, 32\_2, 32\_11, 32\_3, 32\_22, 32\_31, 32\_32, 32\_33, 33\_0, 33\_1, 33\_2, 33\_11, 33\_3, 33\_22, 33\_31, 33\_32, 33\_33).\looseness=1

Diagonal ordering of $\chi \otimes \chi$:\\
 (3\_3, 0\_33, 33\_0, 1\_32, 32\_1, 2\_31, 31\_2, 11\_22, 22\_11, 1\_3, 3\_1, 0\_31, 31\_0, 2\_11, 11\_2, 2\_3, 3\_2, 0\_32, 32\_0, 1\_22, 22\_1, 3\_31, 31\_3, 1\_33, 33\_1, 11\_32, 32\_11, 3\_32, 32\_3, 2\_33, 33\_2, 22\_31, 31\_22, 3\_0, 1\_2, 2\_1, 0\_3, 31\_1, 11\_3, 3\_11, 1\_31, 32\_2, 22\_3, 3\_22, 2\_32, 33\_3, 31\_32, 32\_31, 3\_33, 1\_1, 0\_11, 11\_0, 2\_2, 0\_22, 22\_0, 31\_31, 11\_33, 33\_11, 32\_32, 22\_33, 33\_22, 0\_1, 1\_0, 0\_2, 2\_0, 1\_11, 11\_1, 2\_22, 22\_2, 11\_31, 31\_11, 22\_32, 32\_22, 31\_33, 33\_31, 32\_33, 33\_32, 0\_0, 11\_11, 22\_22, 33\_33).

\subsection[Investigating W tensor W]{Investigating $\boldsymbol{W\otimes W}$}

\begin{Lemma}\label{lemma--Wis20dimensional} The $\mathcal{U}_q(\mathfrak{so}_6)$-module $W=V(2L_1)$ is $20$-dimensional with basis
\begin{gather}
 \{ e_1 \otimes e_1, e_1 \otimes e_2 + q^{-1}e_2 \otimes e_1, e_1\otimes e_3 + q^{-1}e_3 \otimes e_1,\nonumber \\
 \qquad-e_1 \otimes e_4 +q^{-1}e_2 \otimes e_5 - q^{-2}e_4 \otimes e_1+q^{-1}e_5 \otimes e_2, e_1\otimes e_5 +q^{-1}e_5 \otimes e_1, \nonumber\\
 \qquad- e_1 \otimes e_6-q^{-1}e_6 \otimes e_1, \bigl(q+q^{-1}\bigr)e_2 \otimes e_2,e_2 \otimes e_3 + q^{-1} e_3 \otimes e_2, \nonumber\\
 \qquad-e_2 \otimes e_4 -q^{-1}e_4 \otimes e_2, e_2 \otimes e_5 - q^{-1}e_3 \otimes e_6+q^{-2}e_5 \otimes e_2 -\bigl(q^2+1\bigr)e_6 \otimes e_3,\nonumber\\
 \qquad-e_2 \otimes e_6 -q^{_1}e_6\otimes e_2,\bigl(q^{-1}+q\bigr)e_3 \otimes e_3, \nonumber\\
 \qquad-\bigl(q^2 +2 +q^{-2}\bigr)e_3 \otimes e_4 - \bigl(q+2q^{-1}+q^{-3}\bigr)e_4 \otimes e_3, \nonumber\\
 \qquad\phantom{-}{}\bigl(q^2 + 2+q^{-2}\bigr)e_3 \otimes e_5 + \bigl(q+2q^{-1}+q^{-3}\bigr)e_5 \otimes e_3, \nonumber \\
 \qquad\phantom{-}{}
 \bigl(q^4 + 4q^2 + 6 + 4q^{-2} + q^{-4}\bigr)e_4 \otimes e_4,\nonumber \\
 \qquad-\bigl(q^2+3+3q^{-2}+q^{-4}\bigr)e_4 \otimes e_5 - \bigl(q^3 +3q+3q^{-1}+q^{-3}\bigr)e_5 \otimes e_4, \nonumber \\
 \qquad\phantom{-}{}\bigl(q^{-2}+1\bigr)e_4 \otimes e_6 +\bigl(q+q^{-1}\bigr)e_6 \otimes e_4,\bigl(q^2+2+q^{-2}\bigr)e_5 \otimes e_5,\nonumber \\
 \qquad-\bigl(q^{-2}+1\bigr)e_5 \otimes e_6 - \bigl(q^{-1}+q\bigr)e_6 \otimes e_5, \bigl(q^{-1}+q\bigr)e_6 \otimes e_6\}. \label{Wbasis}
\end{gather}
\end{Lemma}
\begin{proof}
 Define $\pi_V(F_i)$ to be the projection of $\Delta(F_i)$ into $V = \mathbb{R}^6 \otimes \mathbb{R}^6$. The tensor product assumes the role of the Kronecker product in this 36-dimensional representation $V$ of $\mathcal{U}_q(\mathfrak{so}_6)$. Note that $\{\pi_V(E_i)\}$ annihilates $e_1 \otimes e_1$, making $e_1 \otimes e_1$ a~highest weight vector with highest weight $2L_1$. To generate $V(2L_1)$, we use Lemma~\ref{VxV-decomp} to obtain the correct compositions of $F_i$s applied to $e_1 \otimes e_1$ to form each basis vector: we apply $F_i$s corresponding to the arrows in the path from $e_1 \otimes e_1$ to each $e_i \otimes e_j$ contained in the upper right-hand cycle as displayed in Figure~\ref{tikz--BigCrystalGraph}. A~few examples are done below, and the rest are left to the reader. The notation $F_{ijk}$ is used to shorten $F_iF_jF_k$,
 \begin{gather*}
 \pi_V(F_1)(e_1 \otimes e_1) = (1 \otimes F_1)(e_1 \otimes e_1) + \bigl(F_1 \otimes q^{-H_1}\bigr)(e_1 \otimes e_1)
 = e_1 \otimes e_2 + q^{-1}e_2 \otimes e_1, \\
 \pi_V(F_{11}) \coloneqq \pi_V(F_1F_1)(e_1 \otimes e_1) = \bigl(q+q^{-1}\bigr)e_2 \otimes e_2 , \\
 \pi_V(F_{11332211}) (e_1 \otimes e_1) = \bigl(q^4 + 4q^2 + 6 + 4q^{-2} + q^{-4}\bigr)e_4 \otimes e_4.
 \end{gather*}
 After complete evaluation, we obtain the basis shown in equation~\eqref{Wbasis} for $V(2L_1)$. Note that the number of basis vectors generated agrees with the dimension formula in equation~\eqref{dimensionformula}.
\end{proof}

\subsubsection[Entries of pi(W tensor W)(C)]{Entries of $\boldsymbol{\pi_{W \otimes W}(C)}$}
Since the coproduct $\Delta$ is a~homomorphism, $\pi_{\mathbb{R}^6 \otimes \mathbb{R}^6 \otimes \mathbb{R}^6 \otimes \mathbb{R}^6}(C)$ can be computed from Definition~\ref{definition-fundrep} by computing $\Delta^3(E_i)$, $\Delta^3(F_i)$, and $\Delta^3\bigl(q^{H_i}\bigr)$ for $1 \le i \le 3$ and substituting them for $E_i$, $F_i$, and $q^{H_i}$, respectively, in equation~\eqref{equation--C}. By Definition~\ref{definition--coproduct}, for all $1 \le i \le 3$,
\begin{gather*}
 \Delta^3 (E_i) = E_i \otimes 1 \otimes 1 \otimes 1 + q^{H_i} \otimes E_i \otimes 1 \otimes 1 + q^{H_1} \otimes q^{H_1} \otimes E_i \otimes 1\\
 \phantom{ \Delta^3 (E_i) =}{}+ q^{H_1} \otimes q^{H_1} \otimes q^{H_1} \otimes E_1 ,\\
 \Delta^3 (F_i) = 1 \otimes 1 \otimes 1 \otimes F_i + 1 \otimes 1 \otimes F_i \otimes q^{-H_i} + 1 \otimes F_i \otimes q^{-H_i} \otimes q^{-H_i}\\
 \phantom{\Delta^3 (F_i) =}{}+ F_i \otimes q^{-H_i} \otimes q^{-H_i} \otimes q^{-H_i} ,\\
 \Delta^3 \bigl(q^{H_i}\bigr) = q^{H_i} \otimes q^{H_i} \otimes q^{H_i} \otimes q^{H_i}.
\end{gather*}
Denote the matrix resulting from the above procedure by $\pi_{\mathbb{R}^{1296}} (C)$.

Denote the basis of $W$ established in the proof of Lemma~\ref{lemma--Wis20dimensional} by $B_W \coloneqq (w_1,w_2,\dots,w_{20})$. Recall the $36 \times 36$ change of basis matrix $Y$ defined in equation~\eqref{equation--C}, and note that a~change of basis by $Y$ maps $w_i$ to the standard basis vector $e_i$ of $\mathbb{R}^6 \otimes \mathbb{R}^6$ for $1 \le i \le 20$. Thus, a~change of basis by $Y \otimes Y$ in $\mathbb{R}^6 \otimes \mathbb{R}^6 \otimes \mathbb{R}^6 \otimes \mathbb{R}^6$ maps $B_W \otimes B_W$ to $\{e_i \otimes e_j \mid 1 \le i,j \le 20\}$. Since for any $1 \le i$, $j \le 36$, $e_i \otimes e_j$ is the $(36(i-1)+j)$-th basis vector of $\mathbb{R}^6 \otimes \mathbb{R}^6 \otimes \mathbb{R}^6 \otimes \mathbb{R}^6$, $\pi_{W \otimes W}(C) = \pi_{\mathbb{R}^{1296}}(C) |_{W \otimes W}$ can be found by taking the $(36(i-1)+j)$-th row and column of~${\bigl(Y^{-1} \otimes Y^{-1}\bigr)\pi_{\mathbb{R}^{1296}}(C)(Y \otimes Y)}$ for all $1 \le i$, $j \le 20$. Since the $e_i \otimes e_j$-th row or column of $\pi_{W \otimes W}(C)$ corresponds directly to the~${(36(i-1)+j)}$-th row or column of $\bigl(Y^{-1} \otimes Y^{-1}\bigr)\pi_{\mathbb{R}^{1296}}(C)(Y \otimes Y)$, we only need to determine a~few of the values of $\bigl(Y^{-1} \otimes Y^{-1}\bigr)\pi_{\mathbb{R}^{1296}}(C)(Y \otimes Y)$ to compute some of the blocks of $\pi_{W \otimes W}(C)$. Explicit values of $\pi_{W \otimes W}(E_i),\pi_{W \otimes W}(F_i),\pi_{W \otimes W}\bigl(q^{H_i}\bigr)$ for $1 \le i \le 3$ are given in Appendix~\ref{subsection--EiFiqHi}.

\subsubsection[Eigenvalues of pi\_\{W otimes W\}(C)]{Eigenvalues of $\boldsymbol{\pi_{W \otimes W}(C)}$}
Given $W \otimes W$'s decomposition into highest weight $\mathcal{U}_q(\mathfrak{so}_6)$-modules, we are also able to deduce the eigenvalues of $\pi_{W \otimes W}(C)$.
\begin{Corollary}
 $\pi_{W \otimes W}(C)$ has eigenvalues of
 \begin{alignat*}{4}
 & q^{12}+q^2+2+q^{-2}+q^{-12}, \qquad&&
 q^{10}+q^4 + 2 + q^{-4}+q^{-10}, \qquad&&
 q^8 +q^6+2+q^{-6}+q^{-8},& \\
& q^8+q^2+2+q^{-2}+q^{-8}, \qquad&&
 q^6 + q^4 +2 + q^{-4}+q^{-6}, \qquad&&
 q^4+q^2+2+q^{-2}+q^{-4}&
 \end{alignat*}
 with respective multiplicities of $105$, $175$, $84$, $20$, $15$, $1$.
\end{Corollary}
\begin{proof}
 By Lemma~\ref{lemma-young},
 \[W \otimes W \cong V(4L_1) \oplus V(3L_1+L_2) \oplus V(2L_1+2L_2) \oplus V(2L_1) \oplus V(L_1+L_2) \oplus V(0)\]
 with respective dimensions of $105$, $175$, $84$, $20$, $15$, $1$. Since each of these representations is irreducible, $\pi_{W \otimes W}(C)|_{V(\lambda)} = c\,\mathrm{Id}$ for some $c \in \mathbb{R}\big[q,q^{-1}\big]$. Also, by construction of each $V(\lambda)$, $v_\lambda$ is a~basis vector of $V(\lambda)$. Therefore, in changing the basis of $W \otimes W$ to compute $\pi_{W \otimes W}(C)|_{V(\lambda)}$, we will have
 $\pi_{W \otimes W}(C)|_{V(\lambda)}(1,1) = \langle v_\lambda|\pi_{W \otimes W}(C)|_{V(\lambda)}|v_\lambda \rangle$
 since we may express each highest weight vector $v_\lambda$ as a~$400$-dimensional vector with a~$1$ in the $i$-th position and $0$s elsewhere. Thus, to find the $c$ associated with each $v_\lambda$, we must only find the $i$-th diagonal entry of $\pi_{W \otimes W}(C)$. We will then have that $c$ is an eigenvalue of multiplicity $\operatorname{dim}(V(\lambda))$ of this matrix with changed basis. Since a~change of basis does not impact the eigenvalues of the matrix and $\pi_{W \otimes W}(C)$ is a~block matrix by Lemma~\ref{WeightSpaceLemma}, $c$ is an eigenvalue of~$\pi_{W \otimes W}(C)$ with multiplicity $\operatorname{dim}(V(\lambda))$.
 Next, as $v_\lambda$ is a~highest weight vector of weight $\lambda$, by definition, $E_1v_\lambda = E_2v_\lambda = E_3v_\lambda = 0$.
 Thus,
 \begin{align*}
 \pi_{W \otimes W}(C)v_\lambda ={}& \pi_{W \otimes W}\bigl(q^{-4-2H_1-H_2-H_3}
 +q^{-2-H_2-H_3}
 +q^{H_2-H_3}
 +q^{H_3-H_2}
 +q^{2+H_2+H_3}\\
 &
 +q^{4+2H_1+H_2+H_3}\bigr)v_\lambda.
 \end{align*}
 Since each $q^{H_i}$ is a~diagonal matrix, we may use Section~\ref{subsection-piwwc-entries} to find the corresponding diagonal entry to each highest weight vector in a~feasible way computationally. For example, $v_{4L_1} = e_1^{\otimes 2}$ where $e_1$ is $20$-dimensional, and thus $\pi_{W \otimes W}(C)$ has an eigenvalue of
 \begin{align*}
 c ={}& \big\langle e_1^{\otimes 2}|\pi_{W \otimes W}(C)|e_1^{\otimes 2} \big\rangle \\
 ={}&\big\langle e_1^{\otimes 2}|\pi_{W \otimes W}\bigl(q^{-4-2H_1-H_2-H_3}
 +q^{-2-H_2-H_3}
 +q^{H_2-H_3}
 +q^{H_3-H_2}
 +q^{2+H_2+H_3}\\
 &
 +q^{4+2H_1+H_2+H_3}\bigr)| e_1^{\otimes 2}\big\rangle \\
 ={}& q^{-4}\cdot q^{-4} \cdot q^{-4} \cdot 1 \cdot 1 + q^{-2}\cdot 1 \cdot 1+1\cdot 1+1\cdot 1+q^2\cdot 1\cdot 1 +q^4\cdot q^4\cdot q^4\cdot 1\cdot 1 \\
 ={}& q^{12}+q^2+2+q^{-2}+q^{-12}
 \end{align*}
 with multiplicity $\operatorname{dim}(V(4L_1)) = 105$. The remaining eigenvalues can be found similarly.
\end{proof}

\begin{Remark}
 Recall that $\mathfrak{so}_6$ is a~rank 3 Lie algebra which is isomorphic to the rank 3 Lie algebra $\mathfrak{sl}_4$. In this case, previous research in~\cite{HarishChandraComment} allows us to express $C$ using generators from both $\mathfrak{so}_6$ and $\mathfrak{sl}_4$ in terms of the above eigenvalues. We choose not to explore this, but someone who wishes to relate~\cite{kuan2020interactingparticlesystemstype}'s central element $C$ to the Harish-Chandra isomorphism may reference~\cite{HarishChandraComment}.
\end{Remark}

\subsection[Ordering of basis vectors to block pi\_\{W otimes W\}(C)]{Ordering of basis vectors to block $\boldsymbol{\pi_{W\otimes W} (C)}$}

Denote $v_i$ to be the $i$-th basis vector of $W$. Then, the blocked matrix $\pi_{W \otimes W}(C)$ is with respect to ordered basis \label{appendix-basisWxW}
\begin{align*}
 &v_{1} \otimes v_{1},
v_{1} \otimes v_{20},
v_{2} \otimes v_{19},
v_{4} \otimes v_{17},
v_{5} \otimes v_{16},
v_{8} \otimes v_{13},
v_{9} \otimes v_{9},
v_{3} \otimes v_{18},
v_{6} \otimes v_{15}, \\
&v_{7} \otimes v_{14},
v_{9} \otimes v_{11},
v_{11} \otimes v_{9},
v_{11} \otimes v_{11},
v_{13} \otimes v_{8},
v_{10} \otimes v_{12},
v_{14} \otimes v_{7},
v_{16} \otimes v_{5},
v_{12} \otimes v_{10}, \\
&v_{15} \otimes v_{6},
v_{17} \otimes v_{4},
v_{18} \otimes v_{3},
v_{19} \otimes v_{2},
v_{20} \otimes v_{1},
v_{1} \otimes v_{13},
v_{2} \otimes v_{9},
v_{2} \otimes v_{11},
v_{4} \otimes v_{7}, \\
&v_{5} \otimes v_{6},
v_{8} \otimes v_{3},
v_{9} \otimes v_{2},
v_{3} \otimes v_{8},
v_{6} \otimes v_{5},
v_{7} \otimes v_{4},
v_{11} \otimes v_{2},
v_{13} \otimes v_{1},
v_{1} \otimes v_{16}, \\
&v_{2} \otimes v_{14},
v_{4} \otimes v_{9},
v_{4} \otimes v_{11},
v_{5} \otimes v_{10},
v_{8} \otimes v_{6},
v_{9} \otimes v_{4},
v_{6} \otimes v_{8},
v_{11} \otimes v_{4},
v_{10} \otimes v_{5}, \\
&v_{14} \otimes v_{2},
v_{16} \otimes v_{1},
v_{1} \otimes v_{17},
v_{2} \otimes v_{15},
v_{4} \otimes v_{12},
v_{5} \otimes v_{9},
v_{5} \otimes v_{11},
v_{8} \otimes v_{7},
v_{9} \otimes v_{5}, \\
&v_{7} \otimes v_{8},
v_{11} \otimes v_{5},
v_{12} \otimes v_{4},
v_{15} \otimes v_{2},
v_{17} \otimes v_{1},
v_{1} \otimes v_{19},
v_{2} \otimes v_{18},
v_{4} \otimes v_{15},
v_{5} \otimes v_{14}, \\
&v_{8} \otimes v_{9},
v_{8} \otimes v_{11},
v_{9} \otimes v_{8},
v_{11} \otimes v_{8},
v_{14} \otimes v_{5},
v_{15} \otimes v_{4},
v_{18} \otimes v_{2},
v_{19} \otimes v_{1},
v_{2} \otimes v_{16}, \\
&v_{3} \otimes v_{14},
v_{4} \otimes v_{13},
v_{6} \otimes v_{9},
v_{6} \otimes v_{11},
v_{7} \otimes v_{10},
v_{9} \otimes v_{6},
v_{11} \otimes v_{6},
v_{13} \otimes v_{4},
v_{14} \otimes v_{3}, \\
&v_{16} \otimes v_{2},
v_{10} \otimes v_{7},
v_{2} \otimes v_{17},
v_{3} \otimes v_{15},
v_{5} \otimes v_{13},
v_{6} \otimes v_{12},
v_{7} \otimes v_{9},
v_{7} \otimes v_{11},
v_{9} \otimes v_{7}, \\
&v_{11} \otimes v_{7},
v_{13} \otimes v_{5},
v_{12} \otimes v_{6},
v_{15} \otimes v_{3},
v_{17} \otimes v_{2},
v_{2} \otimes v_{20},
v_{3} \otimes v_{19},
v_{6} \otimes v_{17},
v_{7} \otimes v_{16}, \\
&v_{9} \otimes v_{13},
v_{11} \otimes v_{13},
v_{13} \otimes v_{9},
v_{13} \otimes v_{11},
v_{16} \otimes v_{7},
v_{17} \otimes v_{6},
v_{19} \otimes v_{3},
v_{20} \otimes v_{2},
v_{4} \otimes v_{19}, \\
&v_{6} \otimes v_{18},
v_{8} \otimes v_{16},
v_{9} \otimes v_{14},
v_{10} \otimes v_{15},
v_{11} \otimes v_{14},
v_{14} \otimes v_{9},
v_{14} \otimes v_{11},
v_{16} \otimes v_{8},
v_{18} \otimes v_{6}, \\
&v_{19} \otimes v_{4},
v_{15} \otimes v_{10},
v_{4} \otimes v_{20},
v_{6} \otimes v_{19},
v_{9} \otimes v_{16},
v_{10} \otimes v_{17},
v_{11} \otimes v_{16},
v_{14} \otimes v_{13},
v_{16} \otimes v_{9}, \\
&v_{13} \otimes v_{14},
v_{16} \otimes v_{11},
v_{17} \otimes v_{10},
v_{19} \otimes v_{6},
v_{20} \otimes v_{4},
v_{5} \otimes v_{19},
v_{7} \otimes v_{18},
v_{8} \otimes v_{17},
v_{9} \otimes v_{15}, \\
&v_{11} \otimes v_{15},
v_{12} \otimes v_{14},
v_{15} \otimes v_{9},
v_{15} \otimes v_{11},
v_{17} \otimes v_{8},
v_{14} \otimes v_{12},
v_{18} \otimes v_{7},
v_{19} \otimes v_{5},
v_{5} \otimes v_{20}, \\
&v_{7} \otimes v_{19},
v_{9} \otimes v_{17},
v_{11} \otimes v_{17},
v_{12} \otimes v_{16},
v_{15} \otimes v_{13},
v_{17} \otimes v_{9},
v_{13} \otimes v_{15},
v_{17} \otimes v_{11},
v_{16} \otimes v_{12}, \\
&v_{19} \otimes v_{7},
v_{20} \otimes v_{5},
v_{8} \otimes v_{20},
v_{9} \otimes v_{19},
v_{11} \otimes v_{19},
v_{14} \otimes v_{17},
v_{15} \otimes v_{16},
v_{18} \otimes v_{13},
v_{19} \otimes v_{9}, \\
&v_{13} \otimes v_{18},
v_{16} \otimes v_{15},
v_{17} \otimes v_{14},
v_{19} \otimes v_{11},
v_{20} \otimes v_{8},
v_{1} \otimes v_{9},
v_{2} \otimes v_{8},
v_{4} \otimes v_{5},
v_{5} \otimes v_{4}, \\
&v_{8} \otimes v_{2},
v_{9} \otimes v_{1},
v_{1} \otimes v_{11},
v_{11} \otimes v_{1},
v_{2} \otimes v_{13},
v_{3} \otimes v_{9},
v_{3} \otimes v_{11},
v_{6} \otimes v_{7},
v_{7} \otimes v_{6}, \\
&v_{9} \otimes v_{3},
v_{11} \otimes v_{3},
v_{13} \otimes v_{2},
v_{4} \otimes v_{16},
v_{6} \otimes v_{14},
v_{9} \otimes v_{10},
v_{10} \otimes v_{9},
v_{10} \otimes v_{11},
v_{11} \otimes v_{10}, \\
&v_{14} \otimes v_{6},
v_{16} \otimes v_{4},
v_{5} \otimes v_{17},
v_{7} \otimes v_{15},
v_{9} \otimes v_{12},
v_{11} \otimes v_{12},
v_{12} \otimes v_{9},
v_{12} \otimes v_{11},
v_{15} \otimes v_{7}, \\
&v_{17} \otimes v_{5},
v_{8} \otimes v_{19},
v_{9} \otimes v_{18},
v_{11} \otimes v_{18},
v_{14} \otimes v_{15},
v_{15} \otimes v_{14},
v_{18} \otimes v_{9},
v_{18} \otimes v_{11},
v_{19} \otimes v_{8}, \\
&v_{9} \otimes v_{20},
v_{13} \otimes v_{19},
v_{16} \otimes v_{17},
v_{17} \otimes v_{16},
v_{19} \otimes v_{13},
v_{20} \otimes v_{9},
v_{11} \otimes v_{20},
v_{20} \otimes v_{11},
v_{1} \otimes v_{6}, \\
&v_{2} \otimes v_{4},
v_{4} \otimes v_{2},
v_{6} \otimes v_{1},
v_{1} \otimes v_{7},
v_{2} \otimes v_{5},
v_{5} \otimes v_{2},
v_{7} \otimes v_{1},
v_{1} \otimes v_{14},
v_{4} \otimes v_{8}, \\
&v_{8} \otimes v_{4},
v_{14} \otimes v_{1},
v_{1} \otimes v_{15},
v_{5} \otimes v_{8},
v_{8} \otimes v_{5},
v_{15} \otimes v_{1},
v_{2} \otimes v_{6},
v_{3} \otimes v_{4},
v_{4} \otimes v_{3}, \\
&v_{6} \otimes v_{2},
v_{2} \otimes v_{7},
v_{3} \otimes v_{5},
v_{5} \otimes v_{3},
v_{7} \otimes v_{2},
v_{2} \otimes v_{10},
v_{4} \otimes v_{6},
v_{6} \otimes v_{4},
v_{10} \otimes v_{2}, \\
&v_{2} \otimes v_{12},
v_{5} \otimes v_{7},
v_{7} \otimes v_{5},
v_{12} \otimes v_{2},
v_{3} \otimes v_{16},
v_{6} \otimes v_{13},
v_{13} \otimes v_{6},
v_{16} \otimes v_{3},
v_{3} \otimes v_{17}, \\
&v_{7} \otimes v_{13},
v_{13} \otimes v_{7},
v_{17} \otimes v_{3},
v_{4} \otimes v_{14},
v_{8} \otimes v_{10},
v_{10} \otimes v_{8},
v_{14} \otimes v_{4},
v_{4} \otimes v_{18},
v_{8} \otimes v_{14}, \\
&v_{14} \otimes v_{8},
v_{18} \otimes v_{4},
v_{5} \otimes v_{15},
v_{8} \otimes v_{12},
v_{12} \otimes v_{8},
v_{15} \otimes v_{5},
v_{5} \otimes v_{18},
v_{8} \otimes v_{15},
v_{15} \otimes v_{8}, \\
&v_{18} \otimes v_{5},
v_{6} \otimes v_{16},
v_{10} \otimes v_{13},
v_{13} \otimes v_{10},
v_{16} \otimes v_{6},
v_{6} \otimes v_{20},
v_{13} \otimes v_{16},
v_{16} \otimes v_{13},
v_{20} \otimes v_{6}, \\
&v_{7} \otimes v_{17},
v_{12} \otimes v_{13},
v_{13} \otimes v_{12},
v_{17} \otimes v_{7},
v_{7} \otimes v_{20},
v_{13} \otimes v_{17},
v_{17} \otimes v_{13},
v_{20} \otimes v_{7},
v_{10} \otimes v_{19}, \\
&v_{14} \otimes v_{16},
v_{16} \otimes v_{14},
v_{19} \otimes v_{10},
v_{12} \otimes v_{19},
v_{15} \otimes v_{17},
v_{17} \otimes v_{15},
v_{19} \otimes v_{12},
v_{14} \otimes v_{19},
v_{16} \otimes v_{18}, \\
&v_{18} \otimes v_{16},
v_{19} \otimes v_{14},
v_{14} \otimes v_{20},
v_{16} \otimes v_{19},
v_{19} \otimes v_{16},
v_{20} \otimes v_{14},
v_{15} \otimes v_{19},
v_{17} \otimes v_{18},
v_{18} \otimes v_{17}, \\
&v_{19} \otimes v_{15},
v_{15} \otimes v_{20},
v_{17} \otimes v_{19},
v_{19} \otimes v_{17},
v_{20} \otimes v_{15},
v_{1} \otimes v_{3},
v_{2} \otimes v_{2},
v_{3} \otimes v_{1},
v_{1} \otimes v_{10}, \\
&v_{4} \otimes v_{4},
v_{10} \otimes v_{1},
v_{1} \otimes v_{12},
v_{5} \otimes v_{5},
v_{12} \otimes v_{1},
v_{1} \otimes v_{18},
v_{8} \otimes v_{8},
v_{18} \otimes v_{1},
v_{3} \otimes v_{10}, \\
&v_{6} \otimes v_{6},
v_{10} \otimes v_{3},
v_{3} \otimes v_{12},
v_{7} \otimes v_{7},
v_{12} \otimes v_{3},
v_{3} \otimes v_{20},
v_{13} \otimes v_{13},
v_{20} \otimes v_{3},
v_{10} \otimes v_{18}, \\
&v_{14} \otimes v_{14},
v_{18} \otimes v_{10},
v_{10} \otimes v_{20},
v_{16} \otimes v_{16},
v_{20} \otimes v_{10},
v_{12} \otimes v_{18},
v_{15} \otimes v_{15},
v_{18} \otimes v_{12},
v_{12} \otimes v_{20}, \\
&v_{17} \otimes v_{17},
v_{20} \otimes v_{12},
v_{18} \otimes v_{20},
v_{19} \otimes v_{19},
v_{20} \otimes v_{18},
v_{1} \otimes v_{2},
v_{2} \otimes v_{1},
v_{1} \otimes v_{4},
v_{4} \otimes v_{1}, \\
&v_{1} \otimes v_{5},
v_{5} \otimes v_{1},
v_{1} \otimes v_{8},
v_{8} \otimes v_{1},
v_{2} \otimes v_{3},
v_{3} \otimes v_{2},
v_{3} \otimes v_{6},
v_{6} \otimes v_{3},
v_{3} \otimes v_{7}, \\
&v_{7} \otimes v_{3},
v_{3} \otimes v_{13},
v_{13} \otimes v_{3},
v_{4} \otimes v_{10},
v_{10} \otimes v_{4},
v_{5} \otimes v_{12},
v_{12} \otimes v_{5},
v_{6} \otimes v_{10},
v_{10} \otimes v_{6}, \\
&v_{7} \otimes v_{12},
v_{12} \otimes v_{7},
v_{8} \otimes v_{18},
v_{18} \otimes v_{8},
v_{10} \otimes v_{14},
v_{14} \otimes v_{10},
v_{10} \otimes v_{16},
v_{16} \otimes v_{10},
v_{12} \otimes v_{15}, \\
&v_{15} \otimes v_{12},
v_{12} \otimes v_{17},
v_{17} \otimes v_{12},
v_{13} \otimes v_{20},
v_{20} \otimes v_{13},
v_{14} \otimes v_{18},
v_{18} \otimes v_{14},
v_{15} \otimes v_{18},
v_{18} \otimes v_{15}, \\
&v_{16} \otimes v_{20},
v_{20} \otimes v_{16},
v_{17} \otimes v_{20},
v_{20} \otimes v_{17},
v_{18} \otimes v_{19},
v_{19} \otimes v_{18},
v_{19} \otimes v_{20},
v_{20} \otimes v_{19},
v_{3} \otimes v_{3}, \\
&v_{10} \otimes v_{10},
v_{12} \otimes v_{12},
v_{18} \otimes v_{18},
v_{20} \otimes v_{20}.
\end{align*}

\subsection{Algebraically-produced Markov generator} \label{subsection-MarkovGen}
This generator is produced as described in Section~\ref{section--FindingAlgebraGenerator}, with the additional step of deleting the $i$-th row and column of the matrix $a^{-1}G^{-1}\pi_{W \otimes W}(C)G - \textrm{Id}$ if $g_i = 0$, or if row $i$ or column $i$ of the matrix contains a~negative off-diagonal entry, as those lack a~probabilistic interpretation.

The authors note that the polynomials below may be expressible as $q$-hypergeometric series, but we do not investigate this possibility in this paper. Recall that $a = q^{12} + q^2 + q^{-2} + q^{-12}$, and let the matrix entry $*$ denote, as elsewhere, $-1$ times the sum of the other elements in its row. We may express the Markov generator produced algebraically as the direct sum of one~${9 \times 9}$ block of the form
\begin{align*}
 a^{-2} \begin{bmatrix}
 * & B_1 & B_1 & B_2 & 0 & 0 & 0 & 0 & 0 \\
 B_3 & * & B_4 & B_5 & B_6 & B_7 & 0 & 0 & 0 \\
 B_3 & B_4 & * & B_5 & 0 & 0 & B_6 & B_7 & 0 \\
 B_8 & B_9 & B_9 & * & B_{10} & B_{11} & B_{10} & B_{11} & B_{12} \\
 0 & B_{13} & 0 & B_{14} & * & B_{15} & 0 & 0 & 0 \\
 0 & B_{16} & 0 & B_{17} & B_{18} & * & 0 & B_{19} & B_{20} \\
 0 & 0 & B_{13} & B_{14} & 0 & 0 & * & B_{15} & 0 \\
 0 & 0 & B_{16} & B_{17} & 0 & B_{19} & B_{18} & * & B_{20} \\
 0 & 0 & 0 & B_{21} & 0 & B_{22} & 0 & B_{22} & *
 \end{bmatrix},
\end{align*}
where
\begin{gather*}
 B_1 =
\flac{q^{22}}{a^2}-\flac{q^{20}}{a^2}-\flac{2q^{18}}{a^2}+\flac{2q^{16}}{a^2}+\flac{3q^{14}}{a^2}-\flac{q^{12}}{a^2}-\flac{4q^{10}}{a^2}+\flac{2q^{6}}{a^2}, \\
 B_2 =
\flac{q^{20}}{a^2}-\flac{2q^{16}}{a^2}-\flac{2q^{14}}{a^2}+\flac{q^{12}}{a^2}+\flac{4q^{10}}{a^2}+\flac{q^{8}}{a^2}-\flac{2q^{6}}{a^2}-\flac{2q^{4}}{a^2}+\flac{1}{a^2}, \\
 B_3 =
\flac{q^{24}}{a^2}-\flac{3q^{22}}{a^2}+\flac{3q^{20}}{a^2}-\flac{q^{18}}{a^2}+\flac{2q^{16}}{a^2}-\flac{4q^{14}}{a^2}+\flac{2q^{12}}{a^2}, \\
 B_4 =
\flac{q^{22}}{a^2}-\flac{2q^{20}}{a^2}+\flac{q^{18}}{a^2}-\flac{2q^{16}}{a^2}+\flac{4q^{14}}{a^2}-\flac{2q^{12}}{a^2}+\flac{q^{10}}{a^2}-\flac{2q^{8}}{a^2}+\flac{q^{6}}{a^2}, \\
 B_5 =
a^2 \cdot \frac{2q^{16}-3q^{14}-q^{12}+2q^{10}+2q^{8}-q^{6}-3q^{4}+2q^{2}}{q^{20}-q^{16}+q^{12}+q^{10}+q^{8}-q^{4}+1}, \\
 B_6 =
\flac{q^{18}}{a^2}-\flac{3q^{16}}{a^2}+\flac{3q^{14}}{a^2}-\flac{q^{12}}{a^2}+\flac{2q^{10}}{a^2}-\flac{4q^{8}}{a^2}+\flac{2q^{6}}{a^2}, \\
 B_7 =
\flac{q^{16}}{a^2}-\flac{2q^{14}}{a^2}+\flac{q^{12}}{a^2}-\flac{2q^{10}}{a^2}+\flac{4q^{8}}{a^2}-\flac{2q^{6}}{a^2}+\flac{q^{4}}{a^2}-\flac{2q^{2}}{a^2}+\flac{1}{a^2}, \\
 B_8 =
\flac{q^{24}}{a^2}-\flac{4q^{22}}{a^2}+\flac{6q^{20}}{a^2}-\flac{4q^{18}}{a^2}+\flac{q^{16}}{a^2}, \\
 B_9 =
\flac{2q^{22}}{a^2}-\flac{5q^{20}}{a^2}+\flac{3q^{18}}{a^2}+\flac{q^{16}}{a^2}+\flac{q^{14}}{a^2}-\flac{4q^{12}}{a^2}+\flac{2q^{10}}{a^2}, \\
 B_{10} =
\flac{q^{18}}{a^2}-\flac{2q^{16}}{a^2}-\flac{q^{14}}{a^2}+\flac{4q^{12}}{a^2}-\flac{q^{10}}{a^2}-\flac{2q^{8}}{a^2}+\flac{q^{6}}{a^2}, \\
 B_{11} =
\flac{2q^{14}}{a^2}-\flac{4q^{12}}{a^2}+\flac{q^{10}}{a^2}+\flac{q^{8}}{a^2}+\flac{3q^{6}}{a^2}-\flac{5q^{4}}{a^2}+\flac{2q^{2}}{a^2}, \\
 B_{12} =
\flac{q^{8}}{a^2}-\flac{4q^{6}}{a^2}+\flac{6q^{4}}{a^2}-\flac{4q^{2}}{a^2}+\flac{1}{a^2}, \\
 B_{13} =
\flac{q^{24}}{a^2}-\flac{q^{22}}{a^2}-\flac{2q^{20}}{a^2}+\flac{2q^{18}}{a^2}+\flac{3q^{16}}{a^2}-\flac{q^{14}}{a^2}-\flac{4q^{12}}{a^2}+\flac{2q^{8}}{a^2}, \\
 B_{14} =
\flac{q^{22}}{a^2}-\flac{2q^{18}}{a^2}-\flac{2q^{16}}{a^2}+\flac{q^{14}}{a^2}+\flac{4q^{12}}{a^2}+\flac{q^{10}}{a^2}-\flac{2q^{8}}{a^2}-\flac{2q^{6}}{a^2}+\flac{q^{2}}{a^2}, \\
 B_{15} =
\flac{2q^{16}}{a^2}-\flac{4q^{12}}{a^2}-\flac{q^{10}}{a^2}+\flac{3q^{8}}{a^2}+\flac{2q^{6}}{a^2}-\flac{2q^{4}}{a^2}-\flac{q^{2}}{a^2}+\flac{1}{a^2}, \\
 B_{16} =
\flac{q^{24}}{a^2}-\flac{2q^{22}}{a^2}+\flac{q^{20}}{a^2}-\flac{2q^{18}}{a^2}+\flac{4q^{16}}{a^2}-\flac{2q^{14}}{a^2}+\flac{q^{12}}{a^2}-\flac{2q^{10}}{a^2}+\flac{q^{8}}{a^2}, \\
 B_{17} =
a^2 \cdot \frac{2q^{18}-3q^{16}-q^{14}+2q^{12}+2q^{10}-q^{8}-3q^{6}+2q^{4}}{q^{20}-q^{16}+q^{12}+q^{10}+q^{8}-q^{4}+1},\\
 B_{18} =
\flac{2q^{18}}{a^2}-\flac{4q^{16}}{a^2}+\flac{2q^{14}}{a^2}-\flac{q^{12}}{a^2}+\flac{3q^{10}}{a^2}-\flac{3q^{8}}{a^2}+\flac{q^{6}}{a^2}, \\
 B_{19} =
\flac{q^{18}}{a^2}-\flac{2q^{16}}{a^2}+\flac{q^{14}}{a^2}-\flac{2q^{12}}{a^2}+\flac{4q^{10}}{a^2}-\flac{2q^{8}}{a^2}+\flac{q^{6}}{a^2}-\flac{2q^{4}}{a^2}+\flac{q^{2}}{a^2}, \\
 B_{20} =
\flac{2q^{12}}{a^2}-\flac{4q^{10}}{a^2}+\flac{2q^{8}}{a^2}-\flac{q^{6}}{a^2}+\flac{3q^{4}}{a^2}-\flac{3q^{2}}{a^2}+\flac{1}{a^2}, \\
 B_{21} =
\flac{q^{24}}{a^2}-\flac{2q^{20}}{a^2}-\flac{2q^{18}}{a^2}+\flac{q^{16}}{a^2}+\flac{4q^{14}}{a^2}+\flac{q^{12}}{a^2}-\flac{2q^{10}}{a^2}-\flac{2q^{8}}{a^2}+\flac{q^{4}}{a^2}, \\
 B_{22} =
\flac{2q^{18}}{a^2}-\flac{4q^{14}}{a^2}-\flac{q^{12}}{a^2}+\flac{3q^{10}}{a^2}+\flac{2q^{8}}{a^2}-\flac{2q^{6}}{a^2}-\flac{q^{4}}{a^2}+\flac{q^{2}}{a^2},
\end{gather*}
a~$6 \times 6$ block of the form
\begin{align*}
 a^{-2}\begin{bmatrix}
 * & B_1 & B_2 & B_3 & 0 & B_4 \\
 B_1 & * & B_2 & 0 & B_3 & B_4 \\
 B_5 & B_5 & * & 0 & 0 & B_6 \\
 B_7 & 0 & 0 & * & B_1 & B_8 \\
 0 & B_7 & 0 & B_1 & * & B_8 \\
 B_9 & B_9 & B_{10} & B_{11} & B_{11} & *
 \end{bmatrix},
\end{align*}
where
\begin{gather*}
 B_1 =
\flac{q^{20}}{a^2}-\flac{2q^{18}}{a^2}+\flac{q^{16}}{a^2}-\flac{2q^{14}}{a^2}+\flac{4q^{12}}{a^2}-\flac{2q^{10}}{a^2}+\flac{q^{8}}{a^2}-\flac{2q^{6}}{a^2}+\flac{q^{4}}{a^2},\\
 B_2 =
\flac{2q^{14}}{a^2}-\flac{4q^{12}}{a^2}+\flac{2q^{10}}{a^2}-\flac{q^{8}}{a^2}+\flac{3q^{6}}{a^2}-\flac{3q^{4}}{a^2}+\flac{q^{2}}{a^2}, \\
 B_3 =
\flac{q^{20}}{a^2}-\flac{2q^{18}}{a^2}+\flac{q^{16}}{a^2}+\flac{q^{6}}{a^2}-\flac{2q^{4}}{a^2}+\flac{q^{2}}{a^2}, \\
 B_4 =
\flac{q^{18}}{a^2}-\flac{2q^{16}}{a^2}+\flac{q^{14}}{a^2}+\flac{q^{4}}{a^2}-\flac{2q^{2}}{a^2}+\flac{1}{a^2}, \\
 B_5 =
\flac{2q^{18}}{a^2}-\flac{2q^{16}}{a^2}-\flac{2q^{14}}{a^2}+\flac{q^{12}}{a^2}+\flac{2q^{10}}{a^2}-\flac{2q^{6}}{a^2}+\flac{q^{4}}{a^2}, \\
 B_6 =
\flac{q^{20}}{a^2}-\flac{q^{18}}{a^2}-\flac{q^{16}}{a^2}+\flac{q^{14}}{a^2}+\flac{q^{6}}{a^2}-\flac{q^{4}}{a^2}-\flac{q^{2}}{a^2}+\flac{1}{a^2}, \\
 B_7 =
\flac{q^{22}}{a^2}-\flac{2q^{20}}{a^2}+\flac{q^{18}}{a^2}+\flac{q^{8}}{a^2}-\flac{2q^{6}}{a^2}+\flac{q^{4}}{a^2}, \\
 B_8 =
\flac{q^{18}}{a^2}-\flac{2q^{16}}{a^2}+\flac{3q^{14}}{a^2}-\flac{4q^{12}}{a^2}+\flac{2q^{10}}{a^2}-\flac{q^{8}}{a^2}+\flac{3q^{6}}{a^2}-\flac{2q^{4}}{a^2}-\flac{q^{2}}{a^2}+\flac{1}{a^2}, \\
 B_9 =
\flac{q^{22}}{a^2}-\flac{2q^{20}}{a^2}+\flac{2q^{18}}{a^2}-\flac{3q^{16}}{a^2}+\flac{3q^{14}}{a^2}-\flac{q^{12}}{a^2}+\flac{q^{10}}{a^2}-\flac{2q^{8}}{a^2}+\flac{q^{6}}{a^2}, \\
 B_{10} =
-\flac{q^{18}}{a^2}+\flac{3q^{16}}{a^2}-\flac{3q^{14}}{a^2}+\flac{q^{12}}{a^2}-\flac{q^{10}}{a^2}+\flac{3q^{8}}{a^2}-\flac{3q^{6}}{a^2}+\flac{q^{4}}{a^2}, \\
 B_{11} =
\flac{q^{20}}{a^2}-\flac{2q^{18}}{a^2}+\flac{2q^{16}}{a^2}-\flac{3q^{14}}{a^2}+\flac{3q^{12}}{a^2}-\flac{q^{10}}{a^2}+\flac{q^{8}}{a^2}-\flac{2q^{6}}{a^2}+\flac{q^{4}}{a^2},
\end{gather*}
two $6 \times 6$ blocks of the form
\begin{align*}
 a^{-2} \begin{bmatrix}
 * & B_1 & B_2 & B_3 & 0 & 0 \\
 B_4 & * & B_5 & B_6 & 0 & 0 \\
 B_7 & B_8 & * & B_9 & 0 & B_3 \\
 B_{10} & B_{11} & B_{12} & * & B_{13} & B_{14} \\
 0 & 0 & 0 & B_{15} & * & B_{16} \\
 0 & 0 & B_{17} & B_{18} & B_{19} & *
 \end{bmatrix},
\end{align*}
where
\begin{gather*}
 B_1 =
\flac{q^{20}}{a^2}-\flac{3q^{18}}{a^2}+\flac{3q^{16}}{a^2}-\flac{q^{14}}{a^2}+\flac{2q^{12}}{a^2}-\flac{4q^{10}}{a^2}+\flac{2q^{8}}{a^2}, \\
 B_2 =
\flac{q^{18}}{a^2}-\flac{2q^{16}}{a^2}+\flac{q^{14}}{a^2}-\flac{2q^{12}}{a^2}+\flac{4q^{10}}{a^2}-\flac{2q^{8}}{a^2}+\flac{q^{6}}{a^2}-\flac{2q^{4}}{a^2}+\flac{q^{2}}{a^2}, \\
 B_3 =
\flac{q^{18}}{a^2}-\flac{2q^{16}}{a^2}+\flac{q^{14}}{a^2}+\flac{q^{4}}{a^2}-\flac{2q^{2}}{a^2}+\flac{1}{a^2}, \\
 B_4 =
\flac{q^{24}}{a^2}-\flac{2q^{22}}{a^2}+\flac{2q^{18}}{a^2}+\flac{q^{16}}{a^2}-\flac{2q^{14}}{a^2}-\flac{2q^{12}}{a^2}+\flac{2q^{10}}{a^2}, \\
 B_5 =
\flac{2q^{16}}{a^2}-\flac{2q^{14}}{a^2}-\flac{2q^{12}}{a^2}+\flac{q^{10}}{a^2}+\flac{2q^{8}}{a^2}-\flac{2q^{4}}{a^2}+\flac{q^{2}}{a^2},\\
 B_6 =
\flac{q^{20}}{a^2}-\flac{q^{18}}{a^2}-\flac{q^{16}}{a^2}+\flac{q^{14}}{a^2}+\flac{q^{6}}{a^2}-\flac{q^{4}}{a^2}-\flac{q^{2}}{a^2}+\flac{1}{a^2}, \\
 B_7 =
\flac{q^{24}}{a^2}-\flac{2q^{22}}{a^2}+\flac{q^{20}}{a^2}-\flac{2q^{18}}{a^2}+\flac{4q^{16}}{a^2}-\flac{2q^{14}}{a^2}+\flac{q^{12}}{a^2}-\flac{2q^{10}}{a^2}+\flac{q^{8}}{a^2}, \\
 B_8 =
\flac{2q^{18}}{a^2}-\flac{4q^{16}}{a^2}+\flac{2q^{14}}{a^2}-\flac{q^{12}}{a^2}+\flac{3q^{10}}{a^2}-\flac{3q^{8}}{a^2}+\flac{q^{6}}{a^2}, \\
 B_9 =
\flac{q^{20}}{a^2}-\flac{2q^{18}}{a^2}+\flac{q^{16}}{a^2}+\flac{q^{6}}{a^2}-\flac{2q^{4}}{a^2}+\flac{q^{2}}{a^2}, \\
 B_{10} =
\flac{q^{24}}{a^2}-\flac{3q^{22}}{a^2}+\flac{3q^{20}}{a^2}-\flac{2q^{18}}{a^2}+\flac{3q^{16}}{a^2}-\flac{3q^{14}}{a^2}+\flac{q^{12}}{a^2}, \\
 B_{11} =
\flac{q^{22}}{a^2}-\flac{2q^{20}}{a^2}+\flac{2q^{18}}{a^2}-\flac{3q^{16}}{a^2}+\flac{3q^{14}}{a^2}-\flac{q^{12}}{a^2}+\flac{q^{10}}{a^2}-\flac{2q^{8}}{a^2}+\flac{q^{6}}{a^2},\\
 B_{12} =
-\flac{q^{18}}{a^2}+\flac{3q^{16}}{a^2}-\flac{3q^{14}}{a^2}+\flac{q^{12}}{a^2}-\flac{q^{10}}{a^2}+\flac{3q^{8}}{a^2}-\flac{3q^{6}}{a^2}+\flac{q^{4}}{a^2}, \\
 B_{13} =
\flac{q^{18}}{a^2}-\flac{2q^{16}}{a^2}+\flac{q^{12}}{a^2}+\flac{q^{10}}{a^2}-\flac{2q^{6}}{a^2}+\flac{q^{4}}{a^2}, \\
 B_{14} =
\flac{q^{14}}{a^2}-\flac{q^{12}}{a^2}-\flac{q^{10}}{a^2}+\flac{2q^{6}}{a^2}-\flac{2q^{2}}{a^2}+\flac{1}{a^2}, \\
 B_{15} =
\flac{q^{20}}{a^2}-\flac{q^{18}}{a^2}-\flac{q^{16}}{a^2}-\flac{q^{14}}{a^2}+\flac{2q^{12}}{a^2}+\flac{2q^{10}}{a^2}-\flac{q^{8}}{a^2}-\flac{q^{6}}{a^2}-\flac{q^{4}}{a^2}+\flac{q^{2}}{a^2}, \\
 B_{16} =
\flac{2q^{14}}{a^2}-\flac{2q^{12}}{a^2}-\flac{2q^{10}}{a^2}+\flac{q^{8}}{a^2}+\flac{2q^{6}}{a^2}-\flac{2q^{2}}{a^2}+\flac{1}{a^2}, \\
 B_{17} =
\flac{q^{22}}{a^2}-\flac{2q^{20}}{a^2}+\flac{q^{18}}{a^2}+\flac{q^{8}}{a^2}-\flac{2q^{6}}{a^2}+\flac{q^{4}}{a^2}, \\
 B_{18} =
\flac{q^{20}}{a^2}-\flac{2q^{18}}{a^2}+\flac{3q^{16}}{a^2}-\flac{4q^{14}}{a^2}+\flac{2q^{12}}{a^2}-\flac{q^{10}}{a^2}+\flac{3q^{8}}{a^2}-\flac{2q^{6}}{a^2}-\flac{q^{4}}{a^2}+\flac{q^{2}}{a^2}, \\
 B_{19} =
\flac{2q^{16}}{a^2}-\flac{4q^{14}}{a^2}+\flac{2q^{12}}{a^2}-\flac{q^{10}}{a^2}+\flac{3q^{8}}{a^2}-\flac{3q^{6}}{a^2}+\flac{q^{4}}{a^2},
\end{gather*}
a~$6 \times 6$ block of the form
\begin{align*}
a^{-2}
 \begin{bmatrix}
 * & B_1 & B_1 & 0 & 0 & 0 \\
 B_2 & * & B_3 & B_4 & 0 & 0 \\
 B_2 & B_5 & * & 0 & B_4 & 0 \\
 0 & B_6 & 0 & * & B_7 & B_8 \\
 0 & 0 & B_6 & B_7 & * & B_8 \\
 0 & 0 & 0 & B_9 & B_9 & *
 \end{bmatrix},
\end{align*}
where
\begin{gather*}
 B_1 =
\flac{q^{22}}{a^2}-\flac{2q^{20}}{a^2}+\flac{2q^{16}}{a^2}+\flac{q^{14}}{a^2}-\flac{2q^{12}}{a^2}-\flac{2q^{10}}{a^2}+\flac{2q^{8}}{a^2}, \\
 B_2 =
\flac{q^{24}}{a^2}-\flac{3q^{22}}{a^2}+\flac{3q^{20}}{a^2}-\flac{q^{18}}{a^2}+\flac{2q^{16}}{a^2}-\flac{4q^{14}}{a^2}+\flac{2q^{12}}{a^2}, \\
 B_3 =
\flac{q^{22}}{a^2}-\flac{2q^{20}}{a^2}+\flac{q^{18}}{a^2}-\flac{2q^{16}}{a^2}+\flac{4q^{14}}{a^2}-\flac{2q^{12}}{a^2}+\flac{q^{10}}{a^2}-\flac{2q^{8}}{a^2}+\flac{q^{6}}{a^2}, \\
 B_4 =
\flac{q^{18}}{a^2}-\flac{2q^{16}}{a^2}+\flac{q^{14}}{a^2}+\flac{q^{4}}{a^2}-\flac{2q^{2}}{a^2}+\flac{1}{a^2}, \\
 B_5 =
\flac{q^{22}}{a^2}-\flac{2q^{20}}{a^2}+\flac{q^{18}}{a^2}-\flac{2q^{16}}{a^2}+\flac{4q^{14}}{a^2}-\flac{2q^{12}}{a^2}+\flac{q^{10}}{a^2}-\flac{2q^{8}}{a^2}+\flac{q^{6}}{a^2}, \\
 B_6 =
\flac{q^{24}}{a^2}-\flac{2q^{22}}{a^2}+\flac{q^{20}}{a^2}+\flac{q^{10}}{a^2}-\flac{2q^{8}}{a^2}+\flac{q^{6}}{a^2}, \\
 B_7 =
\flac{q^{18}}{a^2}-\flac{2q^{16}}{a^2}+\flac{q^{14}}{a^2}-\flac{2q^{12}}{a^2}+\flac{4q^{10}}{a^2}-\flac{2q^{8}}{a^2}+\flac{q^{6}}{a^2}-\flac{2q^{4}}{a^2}+\flac{q^{2}}{a^2}, \\
 B_8 =
\flac{2q^{12}}{a^2}-\flac{4q^{10}}{a^2}+\flac{2q^{8}}{a^2}-\flac{q^{6}}{a^2}+\flac{3q^{4}}{a^2}-\flac{3q^{2}}{a^2}+\flac{1}{a^2}, \\
 B_9 =
\flac{2q^{16}}{a^2}-\flac{2q^{14}}{a^2}-\flac{2q^{12}}{a^2}+\flac{q^{10}}{a^2}+\flac{2q^{8}}{a^2}-\flac{2q^{4}}{a^2}+\flac{q^{2}}{a^2},
\end{gather*}
four $6 \times 6$ blocks of the form
\begin{align*}
 a^{-2} \begin{bmatrix}
 * & B_1 & B_2 & B_3 & 0 & 0 \\
 B_4 & * & B_5 & B_6 & 0 & 0 \\
 B_7 & B_8 & * & B_9 & B_{10} & B_{11} \\
 B_{12} & B_{13} & B_{14} & * & B_{15} & B_{16} \\
 0 & 0 & B_{17} & B_{18} & * & B_{19} \\
 0 & 0 & B_{20} & B_{21} & B_{22} & *
 \end{bmatrix},
\end{align*}
where
\begin{gather*}
 B_1 =
\flac{q^{20}}{a^2}-\flac{2q^{18}}{a^2}+\flac{2q^{14}}{a^2}+\flac{q^{12}}{a^2}-\flac{2q^{10}}{a^2}-\flac{2q^{8}}{a^2}+\flac{2q^{6}}{a^2},\\
 B_2 =
\flac{q^{22}}{a^2}-\flac{3q^{18}}{a^2}+\flac{q^{16}}{a^2}+\flac{2q^{14}}{a^2}+\flac{q^{12}}{a^2}-\flac{2q^{10}}{a^2}-\flac{q^{8}}{a^2}+\flac{q^{6}}{a^2}-\flac{q^{4}}{a^2}+\flac{q^{2}}{a^2}, \\
 B_3 =
\flac{q^{18}}{a^2}-\flac{q^{16}}{a^2}-\flac{q^{14}}{a^2}-\flac{q^{12}}{a^2}+\flac{2q^{10}}{a^2}+\flac{2q^{8}}{a^2}-\flac{q^{6}}{a^2}-\flac{q^{4}}{a^2}-\flac{q^{2}}{a^2}+\flac{1}{a^2}, \\
 B_4 =
\flac{q^{24}}{a^2}-\flac{2q^{22}}{a^2}+\flac{2q^{18}}{a^2}+\flac{q^{16}}{a^2}-\flac{2q^{14}}{a^2}-\flac{2q^{12}}{a^2}+\flac{2q^{10}}{a^2}, \\
 B_5 =
\flac{q^{22}}{a^2}-\flac{q^{20}}{a^2}-\flac{q^{18}}{a^2}-\flac{q^{16}}{a^2}+\flac{2q^{14}}{a^2}+\flac{2q^{12}}{a^2}-\flac{q^{10}}{a^2}-\flac{q^{8}}{a^2}-\flac{q^{6}}{a^2}+\flac{q^{4}}{a^2}, \\
 B_6 =
\flac{q^{20}}{a^2}-\flac{q^{18}}{a^2}+\flac{q^{16}}{a^2}-\flac{q^{14}}{a^2}-\flac{2q^{12}}{a^2}+\flac{q^{10}}{a^2}+\flac{2q^{8}}{a^2}+\flac{q^{6}}{a^2}-\flac{3q^{4}}{a^2}+\flac{1}{a^2}, \\
 B_7 =
\flac{q^{24}}{a^2}-\flac{2q^{22}}{a^2}+\flac{2q^{18}}{a^2}-\flac{q^{14}}{a^2}-\flac{q^{12}}{a^2}+\flac{q^{10}}{a^2}, \\
 B_8 =
\flac{q^{20}}{a^2}-\flac{2q^{18}}{a^2}+\flac{q^{14}}{a^2}+\flac{q^{12}}{a^2}-\flac{2q^{8}}{a^2}+\flac{q^{6}}{a^2}, \\
 B_9 =
\flac{q^{18}}{a^2}-\flac{q^{16}}{a^2}-\flac{2q^{12}}{a^2}+\flac{3q^{10}}{a^2}-\flac{q^{8}}{a^2}+\flac{2q^{6}}{a^2}-\flac{4q^{4}}{a^2}+\flac{2q^{2}}{a^2}, \\
 B_{10} =
\flac{q^{18}}{a^2}-\flac{2q^{16}}{a^2}+\flac{q^{14}}{a^2}-\flac{q^{12}}{a^2}+\flac{3q^{10}}{a^2}-\flac{3q^{8}}{a^2}+\flac{2q^{6}}{a^2}-\flac{2q^{4}}{a^2}+\flac{q^{2}}{a^2}, \\
 B_{11} =
\flac{q^{12}}{a^2}-\flac{3q^{10}}{a^2}+\flac{3q^{8}}{a^2}-\flac{2q^{6}}{a^2}+\flac{3q^{4}}{a^2}-\flac{3q^{2}}{a^2}+\flac{1}{a^2}, \\
 B_{12} =
\flac{q^{24}}{a^2}-\flac{3q^{22}}{a^2}+\flac{3q^{20}}{a^2}-\flac{2q^{18}}{a^2}+\flac{3q^{16}}{a^2}-\flac{3q^{14}}{a^2}+\flac{q^{12}}{a^2}, \\
 B_{13} =
\flac{q^{22}}{a^2}-\flac{2q^{20}}{a^2}+\flac{2q^{18}}{a^2}-\flac{3q^{16}}{a^2}+\flac{3q^{14}}{a^2}-\flac{q^{12}}{a^2}+\flac{q^{10}}{a^2}-\flac{2q^{8}}{a^2}+\flac{q^{6}}{a^2}, \\
 B_{14} =
\flac{2q^{22}}{a^2}-\flac{4q^{20}}{a^2}+\flac{2q^{18}}{a^2}-\flac{q^{16}}{a^2}+\flac{3q^{14}}{a^2}-\flac{2q^{12}}{a^2}-\flac{q^{8}}{a^2}+\flac{q^{6}}{a^2}, \\
 B_{15} =
\flac{q^{18}}{a^2}-\flac{2q^{16}}{a^2}+\flac{q^{12}}{a^2}+\flac{q^{10}}{a^2}-\flac{2q^{6}}{a^2}+\flac{q^{4}}{a^2}, \\
 B_{16} =
\flac{q^{14}}{a^2}-\flac{q^{12}}{a^2}-\flac{q^{10}}{a^2}+\flac{2q^{6}}{a^2}-\flac{2q^{2}}{a^2}+\flac{1}{a^2}, \\
 B_{17} =
\flac{q^{24}}{a^2}-\flac{3q^{20}}{a^2}+\flac{q^{18}}{a^2}+\flac{2q^{16}}{a^2}+\flac{q^{14}}{a^2}-\flac{2q^{12}}{a^2}-\flac{q^{10}}{a^2}+\flac{q^{8}}{a^2}-\flac{q^{6}}{a^2}+\flac{q^{4}}{a^2}, \\
 B_{18} =
\flac{q^{20}}{a^2}-\flac{q^{18}}{a^2}-\flac{q^{16}}{a^2}-\flac{q^{14}}{a^2}+\flac{2q^{12}}{a^2}+\flac{2q^{10}}{a^2}-\flac{q^{8}}{a^2}-\flac{q^{6}}{a^2}-\flac{q^{4}}{a^2}+\flac{q^{2}}{a^2}, \\
 B_{19} =
\flac{2q^{14}}{a^2}-\flac{2q^{12}}{a^2}-\flac{2q^{10}}{a^2}+\flac{q^{8}}{a^2}+\flac{2q^{6}}{a^2}-\flac{2q^{2}}{a^2}+\flac{1}{a^2}, \\
 B_{20} =
\flac{q^{24}}{a^2}-\flac{q^{22}}{a^2}-\flac{q^{20}}{a^2}-\flac{q^{18}}{a^2}+\flac{2q^{16}}{a^2}+\flac{2q^{14}}{a^2}-\flac{q^{12}}{a^2}-\flac{q^{10}}{a^2}-\flac{q^{8}}{a^2}+\flac{q^{6}}{a^2}, \\
 B_{21} =
\flac{q^{22}}{a^2}-\flac{q^{20}}{a^2}+\flac{q^{18}}{a^2}-\flac{q^{16}}{a^2}-\flac{2q^{14}}{a^2}+\flac{q^{12}}{a^2}+\flac{2q^{10}}{a^2}+\flac{q^{8}}{a^2}-\flac{3q^{6}}{a^2}+\flac{q^{2}}{a^2}, \\
 B_{22} =
\flac{2q^{18}}{a^2}-\flac{2q^{16}}{a^2}-\flac{2q^{14}}{a^2}+\flac{q^{12}}{a^2}+\flac{2q^{10}}{a^2}-\flac{2q^{6}}{a^2}+\flac{q^{4}}{a^2},
\end{gather*}
a~$4 \times 4$ block of the form
\begin{align*}
 a^{-2} \begin{bmatrix}
 * & B_1 & B_2 & B_3 \\
 B_1 & * & B_2 & B_3 \\
 B_4 & B_4 & * & B_3 \\
 B_5 & B_5 & B_6 & *
 \end{bmatrix},
\end{align*}
where
\begin{gather*}
 B_1 =
\flac{q^{20}}{a^2}-\flac{2q^{18}}{a^2}+\flac{q^{16}}{a^2}-\flac{2q^{14}}{a^2}+\flac{4q^{12}}{a^2}-\flac{2q^{10}}{a^2}+\flac{q^{8}}{a^2}-\flac{2q^{6}}{a^2}+\flac{q^{4}}{a^2}, \\
 B_2 =
\flac{2q^{14}}{a^2}-\flac{4q^{12}}{a^2}+\flac{2q^{10}}{a^2}-\flac{q^{8}}{a^2}+\flac{3q^{6}}{a^2}-\flac{3q^{4}}{a^2}+\flac{q^{2}}{a^2}, \\
 B_3 =
\flac{q^{18}}{a^2}-\flac{2q^{16}}{a^2}+\flac{q^{14}}{a^2}+\flac{q^{4}}{a^2}-\flac{2q^{2}}{a^2}+\flac{1}{a^2}, \\
 B_4 =
\flac{2q^{16}}{a^2}-\flac{4q^{14}}{a^2}+\flac{2q^{12}}{a^2}-\flac{q^{10}}{a^2}+\flac{3q^{8}}{a^2}-\flac{3q^{6}}{a^2}+\flac{q^{4}}{a^2}, \\
 B_5 =
\flac{q^{22}}{a^2}-\flac{q^{20}}{a^2}-\flac{q^{16}}{a^2}+\flac{2q^{12}}{a^2}-\flac{q^{8}}{a^2}-\flac{q^{6}}{a^2}+\flac{q^{4}}{a^2}, \\
 B_6 =
-\flac{q^{18}}{a^2}+\flac{2q^{16}}{a^2}-\flac{2q^{12}}{a^2}+\flac{2q^{8}}{a^2}-\flac{2q^{4}}{a^2}+\flac{q^{2}}{a^2},
\end{gather*}
two $4 \times 4$ blocks of the form
\begin{align*}
 a^{-2} \begin{bmatrix}
 * & B_1 & B_1 & B_2 \\
 B_3 & * & B_4 & B_5 \\
 B_3 & B_4 & * & B_5 \\
 B_6 & B_7 & B_7 & *
 \end{bmatrix},
\end{align*}
where
\begin{gather*}
 B_1 =
\flac{q^{20}}{a^2}-\flac{q^{18}}{a^2}-\flac{q^{16}}{a^2}+\flac{2q^{12}}{a^2}-\flac{q^{8}}{a^2}-\flac{q^{4}}{a^2}+\flac{q^{2}}{a^2}, \\
 B_2 =
\flac{q^{14}}{a^2}-\flac{2q^{12}}{a^2}+\flac{q^{8}}{a^2}+\flac{q^{6}}{a^2}-\flac{2q^{2}}{a^2}+\flac{1}{a^2}, \\
 B_3 =
\flac{q^{24}}{a^2}-\flac{q^{22}}{a^2}-\flac{2q^{20}}{a^2}+\flac{3q^{18}}{a^2}-\flac{q^{16}}{a^2}+\flac{2q^{14}}{a^2}-\flac{4q^{12}}{a^2}+\flac{3q^{10}}{a^2}-\flac{2q^{8}}{a^2}+\flac{q^{6}}{a^2}, \\
 B_4 =
\flac{q^{20}}{a^2}-\flac{2q^{18}}{a^2}+\flac{q^{16}}{a^2}-\flac{2q^{14}}{a^2}+\flac{4q^{12}}{a^2}-\flac{2q^{10}}{a^2}+\flac{q^{8}}{a^2}-\flac{2q^{6}}{a^2}+\flac{q^{4}}{a^2},\\
 B_5 =
\flac{q^{18}}{a^2}-\flac{2q^{16}}{a^2}+\flac{3q^{14}}{a^2}-\flac{4q^{12}}{a^2}+\flac{2q^{10}}{a^2}-\flac{q^{8}}{a^2}+\flac{3q^{6}}{a^2}-\flac{2q^{4}}{a^2}-\flac{q^{2}}{a^2}+\flac{1}{a^2}, \\
 B_6 =
\flac{q^{24}}{a^2}-\flac{2q^{22}}{a^2}+\flac{q^{18}}{a^2}+\flac{q^{16}}{a^2}-\flac{2q^{12}}{a^2}+\flac{q^{10}}{a^2}, \\
 B_7 =
\flac{q^{22}}{a^2}-\flac{q^{20}}{a^2}-\flac{q^{16}}{a^2}+\flac{2q^{12}}{a^2}-\flac{q^{8}}{a^2}-\flac{q^{6}}{a^2}+\flac{q^{4}}{a^2},
\end{gather*}
two $4 \times 4$ blocks of the form
\begin{align*}
 a^{-2}\begin{bmatrix}
 * & B_1 & B_1 & B_2 \\
 B_3 & * & B_4 & B_5 \\
 B_3 & B_4 & * & B_5 \\
 B_6 & B_7 & B_7 & *
 \end{bmatrix},
\end{align*}
where
\begin{gather*}
 B_1 = \flac{q^{20}}{a^2}-\flac{q^{18}}{a^2}-\flac{2q^{16}}{a^2}+\flac{3q^{14}}{a^2}-\flac{q^{12}}{a^2}+\flac{2q^{10}}{a^2}-\flac{4q^{8}}{a^2}+\flac{3q^{6}}{a^2}-\flac{2q^{4}}{a^2}+\flac{q^{2}}{a^2}, \\
 B_2 = \flac{q^{16}}{a^2}-\flac{2q^{14}}{a^2}+\flac{q^{12}}{a^2}-\flac{2q^{10}}{a^2}+\flac{4q^{8}}{a^2}-\flac{2q^{6}}{a^2}+\flac{q^{4}}{a^2}-\flac{2q^{2}}{a^2}+\flac{1}{a^2}, \\
 B_3 = \flac{q^{24}}{a^2}-\flac{q^{22}}{a^2}-\flac{2q^{20}}{a^2}+\flac{2q^{18}}{a^2}+\flac{2q^{16}}{a^2}-\flac{q^{14}}{a^2}-\flac{2q^{12}}{a^2}+\flac{q^{8}}{a^2}, \\
 B_4 = \flac{q^{20}}{a^2}-\flac{q^{18}}{a^2}-\flac{2q^{16}}{a^2}+\flac{q^{14}}{a^2}+\flac{2q^{12}}{a^2}+\flac{q^{10}}{a^2}-\flac{2q^{8}}{a^2}-\flac{q^{6}}{a^2}+\flac{q^{4}}{a^2}, \\
 B_5 = \flac{q^{16}}{a^2}-\flac{2q^{12}}{a^2}-\flac{q^{10}}{a^2}+\flac{2q^{8}}{a^2}+\flac{2q^{6}}{a^2}-\flac{2q^{4}}{a^2}-\flac{q^{2}}{a^2}+\flac{1}{a^2}, \\
 B_6 = \flac{q^{24}}{a^2}-\flac{2q^{22}}{a^2}+\flac{q^{20}}{a^2}-\flac{2q^{18}}{a^2}+\flac{4q^{16}}{a^2}-\flac{2q^{14}}{a^2}+\flac{q^{12}}{a^2}-\flac{2q^{10}}{a^2}+\flac{q^{8}}{a^2}, \\
 B_7 = \flac{q^{22}}{a^2}-\flac{2q^{20}}{a^2}+\flac{3q^{18}}{a^2}-\flac{4q^{16}}{a^2}+\flac{2q^{14}}{a^2}-\flac{q^{12}}{a^2}
 +\flac{3q^{10}}{a^2}-\flac{2q^{8}}{a^2}-\flac{q^{6}}{a^2}+\flac{q^{4}}{a^2},
\end{gather*}
eight $4 \times 4$ blocks of the form
\begin{align*}
 a^{-2} \begin{bmatrix}
 * & B_1 & B_2 & 0 \\
 B_3 & * & B_4 & B_5 \\
 B_3 & B_6 & * & B_5 \\
 0 & B_7 & B_1 & *
 \end{bmatrix},
\end{align*}
where
\begin{gather*}
 B_1 = \flac{q^{22}}{a^2}-\flac{q^{20}}{a^2}-\flac{q^{18}}{a^2}+\flac{q^{16}}{a^2}+\flac{q^{8}}{a^2}-\flac{q^{6}}{a^2}-\flac{q^{4}}{a^2}+\flac{q^{2}}{a^2}, \\
 B_2 = \flac{q^{20}}{a^2}-\flac{q^{18}}{a^2}-\flac{q^{16}}{a^2}+\flac{q^{14}}{a^2}+\flac{q^{6}}{a^2}-\flac{q^{4}}{a^2}-\flac{q^{2}}{a^2}+\flac{1}{a^2}, \\
 B_3 = \flac{q^{24}}{a^2}-\flac{2q^{22}}{a^2}+\flac{q^{20}}{a^2}+\flac{q^{10}}{a^2}-\flac{2q^{8}}{a^2}+\flac{q^{6}}{a^2}, \\
 B_4 = \flac{q^{20}}{a^2}-\flac{2q^{18}}{a^2}+\flac{q^{16}}{a^2}+\flac{q^{6}}{a^2}-\flac{2q^{4}}{a^2}+\flac{q^{2}}{a^2}, \\
 B_5 = \flac{q^{18}}{a^2}-\flac{2q^{16}}{a^2}+\flac{q^{14}}{a^2}+\flac{q^{4}}{a^2}-\flac{2q^{2}}{a^2}+\flac{1}{a^2}, \\
 B_6 = \flac{q^{22}}{a^2}-\flac{2q^{20}}{a^2}+\flac{q^{18}}{a^2}+\flac{q^{8}}{a^2}-\flac{2q^{6}}{a^2}+\flac{q^{4}}{a^2},\\
 B_7 = \flac{q^{24}}{a^2}-\flac{q^{22}}{a^2}-\flac{q^{20}}{a^2}+\flac{q^{18}}{a^2}+\flac{q^{10}}{a^2}-\flac{q^{8}}{a^2}-\flac{q^{6}}{a^2}+\flac{q^{4}}{a^2},
\end{gather*}
four $4 \times 4$ blocks of the form
\begin{align*}
 a^{-2}\begin{bmatrix}
 * & B_1 & B_1 & B_2 \\
 B_3 & * & 0 & B_4 \\
 B_3 & 0 & * & B_4 \\
 B_5 & B_6 & B_6 & *
 \end{bmatrix},
\end{align*}
where
\begin{gather*}
 B_1 = \flac{q^{20}}{a^2}-\flac{2q^{18}}{a^2}+\flac{q^{16}}{a^2}+\flac{q^{6}}{a^2}-\flac{2q^{4}}{a^2}+\flac{q^{2}}{a^2}, \\
 B_2 = \flac{q^{18}}{a^2}-\flac{2q^{16}}{a^2}+\flac{q^{14}}{a^2}+\flac{q^{4}}{a^2}-\flac{2q^{2}}{a^2}+\flac{1}{a^2}, \\
 B_3 = \flac{q^{24}}{a^2}-\flac{q^{22}}{a^2}-\flac{q^{20}}{a^2}+\flac{q^{18}}{a^2}+\flac{q^{10}}{a^2}-\flac{q^{8}}{a^2}-\flac{q^{6}}{a^2}+\flac{q^{4}}{a^2}, \\
 B_4 = \flac{q^{20}}{a^2}-\flac{q^{18}}{a^2}-\flac{q^{16}}{a^2}+\flac{q^{14}}{a^2}+\flac{q^{6}}{a^2}-\flac{q^{4}}{a^2}-\flac{q^{2}}{a^2}+\flac{1}{a^2}, \\
 B_5 = \flac{q^{24}}{a^2}-\flac{2q^{22}}{a^2}+\flac{q^{20}}{a^2}+\flac{q^{10}}{a^2}-\flac{2q^{8}}{a^2}+\flac{q^{6}}{a^2}, \\
 B_6 = \flac{q^{22}}{a^2}-\flac{2q^{20}}{a^2}+\flac{q^{18}}{a^2}+\flac{q^{8}}{a^2}-\flac{2q^{6}}{a^2}+\flac{q^{4}}{a^2},
\end{gather*}
eight $3 \times 3$ blocks of the form
\begin{align*}
 a^{-2} \begin{bmatrix}
 B_1 & -B_1 & 0 \\
 B_2 & B_3 & B_4 \\
 0 & B_5 & -B_5
 \end{bmatrix},
\end{align*}
where
\begin{gather*}
 B_1 = -\flac{q^{22}}{a^2}+\flac{2q^{18}}{a^2}-\flac{q^{14}}{a^2}-\flac{q^{8}}{a^2}+\flac{2q^{4}}{a^2}-\flac{1}{a^2}, \\
 B_2 = \flac{q^{24}}{a^2}-\flac{2q^{22}}{a^2}+\flac{q^{20}}{a^2}+\flac{q^{10}}{a^2}-\flac{2q^{8}}{a^2}+\flac{q^{6}}{a^2}, \\
 B_3 = -\flac{q^{24}}{a^2}+\flac{2q^{22}}{a^2}-\flac{q^{20}}{a^2}-\flac{q^{18}}{a^2}+\flac{2q^{16}}{a^2}-\flac{q^{14}}{a^2}-\flac{q^{10}}{a^2}+\flac{2q^{8}}{a^2}-\flac{q^{6}}{a^2}-\flac{q^{4}}{a^2}+\flac{2q^{2}}{a^2}-\flac{1}{a^2}, \\
 B_4 = \flac{q^{18}}{a^2}-\flac{2q^{16}}{a^2}+\flac{q^{14}}{a^2}+\flac{q^{4}}{a^2}-\flac{2q^{2}}{a^2}+\flac{1}{a^2}, \\
 B_5 = \flac{q^{24}}{a^2}-\flac{2q^{20}}{a^2}+\flac{q^{16}}{a^2}+\flac{q^{10}}{a^2}-\flac{2q^{6}}{a^2}+\flac{q^{2}}{a^2},
\end{gather*}
sixteen $2 \times 2$ blocks of the form
\smash{$
 a^{-2} \big[\begin{smallmatrix}
 B_1 & -B_1 \\
 B_2 & -B_2
 \end{smallmatrix}\big]$},
where
\begin{align*}
 B_1 &= -\flac{q^{20}}{a^2}+\flac{q^{18}}{a^2}+\flac{q^{16}}{a^2}-\flac{q^{14}}{a^2}-\flac{q^{6}}{a^2}+\flac{q^{4}}{a^2}+\flac{q^{2}}{a^2}-\flac{1}{a^2}, \\
 B_2 &= \flac{q^{24}}{a^2}-\flac{q^{22}}{a^2}-\flac{q^{20}}{a^2}+\flac{q^{18}}{a^2}+\flac{q^{10}}{a^2}-\flac{q^{8}}{a^2}-\flac{q^{6}}{a^2}+\flac{q^{4}}{a^2},
\end{align*}
and five $1 \times 1$ blocks with entry $0$.

Note that multiplying the $3 \times 3$ and $2 \times 2$ blocks of the algebraic generator by the scalar
\begin{align*}
 \frac{a^2}{q^{12}} \cdot \frac{q^2-1+q^{-2}}{\bigl(q-q^{-1}\bigr)^2 \bigl(q+q^{-1}\bigr) \bigl(q^7+q^{-7}\bigr)}
\end{align*}
yields the $3 \times 3$ and $2 \times 2$ blocks of $L_Q$ (as described in Section~\ref{section--ProbProofs}) when $n=2$. Note additionally that no larger block of $L_Q$ (with $n=2$) is a~constant multiple of a~block of the algebraic generator.

\subsection[Type A ground state transformation]{Type A ground state transformation} \label{subsection-groundstatetransformation}
The following notes were provided to us by Dr.~Jeffrey Kuan to aid in our calculation of the ground state transformation.

Suppose $C$ is a~central element and $H= (\pi_V \otimes \pi_V) (\Delta (C))$. Let $v \in V$ be a~highest-weight vector of a~highest-weight representation $V$. Then,
$H(v \otimes v) = av\otimes v$
for some constant $a$. By the definition of a~central element,
$\Delta\bigl(F_j^k\bigr) C = C \Delta\bigl(F_j^k\bigr)$.
Given unit basis vectors $u_i \otimes u_j \in V \otimes V$, let $(k_1, k_2, k_3) \in \mathbb{Z}_{\geq 0}^3$ such that
\[
\bigl\langle u_i \otimes u_j|\Delta\bigl(F_3^{k_3}\bigr) \cdot \Delta\bigl(F_2^{k_2}\bigr) \cdot \Delta\bigl(F_1^{k_1}\bigr) | v \otimes v \bigr\rangle \neq 0.
\]
$F_3$, $F_2$, $F_1$ may need to be in a~different order, but the values of \textbf{k}$= (k_1, k_2, k_3)$ should be the same.

Set $G\colon V \otimes V \to V \otimes V$ by
\[G(u_i \otimes u_j) =\bigl\langle u_i \otimes u_j | \Delta\bigl(F_3^{k_3}\bigr) \cdot \Delta\bigl(F_2^{k_2}\bigr) \cdot \Delta\bigl(F_1^{k_1}\bigr) | v \otimes v \bigr\rangle u_i \otimes u_j,
 \]
where \textbf{k} is the above term.

Now, set $S = G^{-1}HG$.
\begin{Theorem}
 $a^{-1}S-\mathrm{Id}$ has rows that sum to $0$.
\end{Theorem}
\begin{proof}
 For fixed $u_i \otimes u_j$, we have that the sum of the row indexed by $u_i \otimes u_j$ equals
 \[
 \sum_{w_1,w_2} \frac{\langle u_i \otimes u_j | H | w_1 \otimes w_2 \rangle \langle w_1 \otimes w_2 | G | w_1 \otimes w_2 \rangle}{\langle u_i \otimes u_j | G | u_i \otimes u_j \rangle},
 \]
 which then equals
 \[
 \sum_{w_1,w_2} \left[ \frac{\langle u_i \otimes u_j | H | w_1 \otimes w_2 \rangle }{\langle u_i \otimes u_j | G | u_i \otimes u_j \rangle} \sum_{\text{\textbf{k}}} \langle w_1 \otimes w_2 | \Delta\bigl(F_3^{k_3}\bigr) \cdot \Delta\bigl(F_2^{k_2}\bigr) \cdot \Delta\bigl(F_1^{k_1}\bigr) | v \otimes v \rangle \right].\]
 By commutation between the coproduct and $H$, this must then equal
 \[
 \sum_{\text{\textbf{k}}} \frac{\bigl\langle u_i \otimes u_j | \Delta\bigl(F_3^{k_3}\bigr) \cdot \Delta\bigl(F_2^{k_2}\bigr) \cdot \Delta\bigl(F_1^{k_1}\bigr)H | v \otimes v \bigr\rangle}{\langle u_i \otimes u_j | G | u_i \otimes u_j \rangle}.
 \]
 Finally, this equals
 \[
 \sum_{\text{\textbf{k}}} \frac{a\bigl\langle u_i \otimes u_j | \Delta\bigl(F_3^{k_3}\bigr) \cdot \Delta\bigl(F_2^{k_2}\bigr) \cdot \Delta\bigl(F_1^{k_1}\bigr) | v \otimes v \bigr\rangle}{\langle u_i \otimes u_j | G | u_i \otimes u_j \rangle} = a.
 \]
 Multiplying by $a^{-1}$ and subtracting the identity, we obtain a~matrix whose rows sum to $0$.
\end{proof}

\subsection[Explicit values of E\_i, F\_i, q\^{}\{H\_i\} in W otimes W]{Explicit values of $\boldsymbol{E_i}$, $\boldsymbol{F_i}$, $\boldsymbol{q^{H_i}}$ in $\boldsymbol{W \otimes W}$}
\label{subsection--EiFiqHi}
Now we give explicitly the values of $\pi_{W \otimes W}(E_i)$, $\pi_{W \otimes W}(F_i)$, and $\pi_{W \otimes W}\bigl(q^{H_i}\bigr)$ for $1 \le i \le 3$. When performing the calculations, the decomposition of $W\otimes W$ into irreducible representations and weight spaces significantly reduces the computational effort, allowing all computations to be done by hand. The matrices below are given for reference only, and are strictly speaking not necessary. In the expressions below, $E_{i,j}$ denotes the $400 \times 400$ matrix with all entries $0$ except for the $(i,j)$-th entry, which is a~$1$,
\begin{align*}
 \pi_{W \otimes W}(E_1) ={} &\bigl(q^3+q\bigr)(E_{1,2}+E_{2,3}+E_{8,9}+E_{19,20}+E_{141,142}+E_{142,143}+E_{148,149}\\
 &+E_{159,160}+E_{341,342}+E_{342,343}+E_{348,349}+E_{359,360}) \\
 &+\bigl(q^2+1\bigr)(E_{61,62}+E_{62,63}+E_{68,69}+E_{79,80}+E_{81,82}+E_{82,83}+E_{88,89} \\
 &+E_{99,100}+E_{261,262}+E_{262,263}+E_{268,269}+E_{279,280}+E_{281,282}+E_{282,283} \\
 &+E_{288,289}+E_{299,300})\\
 &+\bigl(q^2\bigr)(E_{4,6}+E_{5,7}+E_{8,11}+E_{9,13}+E_{14,16}+E_{15,17}+E_{18,19}+E_{144,146}\\
 &+E_{145,147}+E_{148,151}+E_{149,153}+E_{154,156}+E_{155,157}+E_{158,159}+E_{344,346}\\
 &+E_{345,347}+E_{348,351}+E_{349,353}+E_{354,356}+E_{355,357}+E_{358,359})\\
 &+(q)(E_{64,66}+E_{65,67}+E_{68,71}+E_{69,73}+E_{74,76}+E_{75,77}+E_{78,79}+E_{84,86} \\
 &+E_{85,87}+E_{88,91}+E_{89,93}+E_{94,96}+E_{95,97}+E_{98,99}+E_{264,266}+E_{265,267}\\
 &+E_{268,271}+E_{269,273}+E_{274,276} +E_{275,277}+E_{278,279}+E_{284,286}+E_{285,287}\\
 &+E_{288,291}+E_{289,293}+E_{294,296}+E_{295,297}+E_{298,299}) \\
 &+\bigl(q+q^{-1}\bigr)\left( \sum_{i=1}^{40} E_{i,i+20} + \sum_{i=141}^{160} E_{i,i+20} +\sum_{i=361}^{400}\bigl(E_{i,i+20} + E_{22,23}+E_{28,29} \right. \\
 &+E_{39,40}+E_{362,363}+E_{379,380}\bigr)+E_{162,163}+E_{168,169}+E_{179,180}+E_{181,182}\\
 &+E_{182,183}+E_{188,189}+E_{199,200}+E_{201,202}+E_{202,203}+E_{208,209}+E_{219,220}\\
 &\left.+E_{221,222}+E_{222,223}+E_{228,229}+E_{239,240}+E_{361,362} \right) \\
 &+(1) \left( \sum_{i=61}^{100} E_{i,i+40} + \sum_{i=141}^{160} E_{i,i+60}+\sum_{i=161}^{180} E_{i,i+80} +\sum_{i=261}^{300} E_{i,i+40} \right. \\
 &+\sum_{i=341}^{360}E_{i,i+20}+E_{24,26}+E_{25,27}+E_{28,31}+E_{29,33}+E_{34,36}+E_{35,37}\\
 &+E_{38,39}+E_{164,166}+E_{165,167}+E_{169,173}+E_{174,176}+E_{175,177}+E_{178,179}\\
 &+E_{184,186}+E_{185,187}+E_{188,191}+E_{189,193}+E_{194,196}+E_{195,197}+E_{198,199}\\
 &+E_{204,206}+E_{205,207}+E_{208,211}+E_{209,213}+E_{214,216}+E_{215,217}+E_{218,219} \\
 &+E_{224,226}+E_{225,227}+E_{228,231}+E_{229,233}+E_{234,236}+E_{235,237}+E_{238,239} \\
 &\left.+E_{364,366}+E_{365,367}+E_{368,371}+E_{369,373}+E_{374,376}+E_{375,377}+E_{378,379})\right) \\
 &+\bigl(q^{-1}\bigr)(E_{104,106}+E_{105,107}+E_{108,111}+E_{109,113}+E_{114,116}+E_{115,117} \\
 &+E_{118,119}+E_{124,126}+E_{125,127}+E_{128,131}+E_{129,133}+E_{134,136}+E_{135,137}\\
 &+E_{138,139} +E_{304,306}+E_{305,307}+E_{308,311}+E_{309,313}+E_{314,316}+E_{315,317} \\
 &+E_{318,319}+E_{324,326}+E_{325,327}+E_{328,331}+E_{329,333}+E_{334,336}+E_{335,337} \\
 &+E_{338,339})+\bigl(q^{-2}\bigr)(E_{44,46}+E_{45,47}+E_{48,51}+E_{49,53}+E_{54,56}+E_{55,57} \\
 &+E_{58,59}+E_{244,246}+E_{245,247}+E_{248,251}+E_{249,253}+E_{254,256}+E_{255,257} \\
 &+E_{258,259}+E_{384,386}+E_{385,387}+E_{388,391}+E_{389,393}+E_{394,396}+E_{395,397} \\
 &+E_{398,399})+\bigl(1+q^{-2}\bigr)(E_{101,102}+E_{102,103}+E_{108,109}+E_{119,120}+E_{121,122} \\
 &+E_{122,123}+E_{128,129}+E_{139,140}+E_{301,302}+E_{302,303}+E_{308,309}+E_{319,320} \\
 &+E_{321,322}+E_{322,323}+E_{328,329}+E_{339,340}) \\
 &+\bigl(q^{-1}+q^{-3}\bigr)(E_{41,42}+E_{42,43}+E_{48,49}+E_{59,60}+E_{241,242}+E_{242,243} \\
 &+E_{248,249}+E_{259,260}+E_{381,382}+E_{382,383}+E_{388,389}+E_{399,400}),
\\
 \pi_{W \otimes W}(E_2) ={}&
\bigl(q^3 + q\bigr)(E_{46,50}+E_{47,51}+E_{51,54}+E_{52,55}+E_{53,56}+E_{55,58}+E_{126,130}\\
&+E_{127,131}+E_{131,134}+E_{132,135}+E_{133,136}+E_{135,138}+E_{226,230}+E_{227,231}\\
&+E_{231,234}+E_{232,235}+E_{233,236}+E_{235,238})\\
&+\bigl(q^2 + 1\bigr)(E_{26,30}+E_{27,31}+E_{31,34}+E_{32,35}+E_{33,36}+E_{35,38}+E_{86,90}\\
&+E_{87,91}+E_{91,94}+E_{92,95}+E_{93,96}+E_{95,98}+E_{246,250}+E_{247,251}+E_{251,254}\\
&+E_{252,255}+E_{253,256}+E_{255,258}+E_{326,330}+E_{327,331}+E_{331,334}+E_{332,335}\\
&+E_{333,336}+E_{335,338})\\
&+\bigl(q^2\bigr)(E_{42,44}+E_{43,46}+E_{45,48}+E_{47,49}+E_{57,59}+E_{122,124}+E_{123,126}\\
&+E_{125,128}+E_{127,129}+E_{137,139}+E_{222,224}+E_{223,226}+E_{225,228}\\
&+E_{227,229}+E_{237,239})\\
&+(q)(E_{22,24}+E_{23,26}+E_{25,28}+E_{27,29}+E_{37,39}+E_{82,84}+E_{83,86}+E_{85,88}\\
&+E_{87,89}+E_{97,99}+E_{242,244}+E_{243,246}+E_{245,248}+E_{247,249}+E_{257,259}\\
&+E_{322,324}+E_{323,326}+E_{325,328}+E_{327,329}+E_{337,339})
\\
&+\bigl(q + q^{-1}\bigr)\left(\sum_{i=101}^{140} E_{i, i+80} +\sum_{i=201}^{260} E_{i,i+60} + \sum_{i=281} E_{i,i+60} + E_{6,10}+E_{7,11}\right. \\
&+E_{11,14}
+E_{12,15}+E_{13,16}+E_{15,18}+E_{106,110}+E_{107,111}+E_{111,114}\\
&+E_{112,115}+E_{113,116}+E_{115,118}+E_{166,170}+E_{167,171}+E_{171,174}+E_{172,175}\\
&+E_{173,176}+E_{175,178}+E_{206,210}+E_{207,211}+E_{211,214}+E_{212,215}+E_{213,216}\\
&+E_{215,218}+E_{286,290}+E_{287,291}+E_{291,294}+E_{292,295}+E_{293,296}+E_{295,298}\\
&\left.+E_{386,390}+E_{387,391}+E_{391,394}+E_{392,395}+E_{393,396}+E_{395,398}\right)
\\
&+(1)\left(\sum_{i=21}^{40}E_{i,i+40} +\sum_{i=41}^{60}E_{i,i+60} \sum_{i=81}^{100}E_{i,i+60} + \sum_{i=121}^{140}E_{i,i+40} +\sum_{i=321}^{340} E_{i,i+40} \right. \\
&+E_{2,4}+E_{3,6}+E_{5,8}+E_{7,9}+E_{17,19}+E_{102,104}+E_{103,106}+E_{105,108}\\
&+E_{107,109}+E_{117,119}+E_{162,164}+E_{163,166}+E_{165,168}+E_{167,169}+E_{177,179}\\
&+E_{202,204}+E_{203,206}+E_{205,208}+E_{207,209}+E_{217,219}+E_{282,284}+E_{283,286}\\
&+E_{285,288}+E_{287,289}+E_{297,299}+E_{382,384}+E_{383,386} +E_{385,388}\\
&+E_{387,389}+E_{397,399})
\\
&+\bigl(q^{-1}\bigr)(E_{62,64}+E_{63,66}+E_{65,68}+E_{67,69}+E_{77,79}+E_{142,144}+E_{143,146}\\
&+E_{145,148}+E_{147,149}+E_{157,159}+E_{302,304}+E_{303,306}+E_{305,308}+E_{307,309}\\
&+E_{317,319}+E_{362,364}+E_{363,366}+E_{365,368}+E_{367,369}+E_{377,379})
\\
&+\bigl(q^{-2}\bigr)(E_{182,184}+E_{183,186}+E_{185,188}+E_{187,189}+E_{197,199}+E_{262,264}\\
&+E_{263,266}+E_{265,268}+E_{267,269}+E_{277,279}+E_{342,344}+E_{343,346}+E_{345,348}\\
&+E_{347,349}+E_{357,359})
\\
&+\bigl(1 + q^{-2}\bigr)(E_{66,70}+E_{67,71}+E_{71,74}+E_{72,75}+E_{73,76}+E_{75,78}+E_{146,150}\\
&+E_{147,151}+E_{151,154}+E_{152,155}+E_{153,156}+E_{155,158}+E_{306,310}+E_{307,311}\\
&+E_{311,314}+E_{312,315}+E_{313,316}+E_{315,318}+E_{366,370}+E_{367,371}+E_{371,374}\\
&+E_{372,375}+E_{373,376}+E_{375,378})
\\
&+\bigl(q^{-1} + q^{-3}\bigr)(E_{186,190}+E_{187,191}+E_{191,194}+E_{192,195}+E_{193,196}+E_{195,198}\\
&+E_{266,270}+E_{267,271}+E_{271,274}+E_{272,275}+E_{273,276}+E_{275,278}+E_{346,350}\\
&+E_{347,351}+E_{351,354}+E_{352,355}
+E_{353,356}+E_{355,358}),\\
 \pi_{W \otimes W}(E_3) ={}&
\bigl(-q^3 - q\bigr)(E_{46,51}+E_{47,52}+E_{50,54}+E_{51,55}+E_{53,57}+E_{54,58}+E_{106,111}\\
&
+E_{107,112}+E_{110,114}+E_{111,115}+E_{113,117}+E_{114,118}+E_{186,191}+E_{187,192}\\
&+E_{190,194}+E_{191,195}+E_{193,197}+E_{194,198})\\
&+\bigl(-q^2 - 1\bigr)(E_{26,31}+E_{27,32}+E_{30,34}+E_{31,35}+E_{33,37}+E_{34,38}+E_{66,71}\\
&+E_{67,72}+E_{70,74}+E_{71,75}+E_{73,77}+E_{74,78}+E_{246,251}+E_{247,252}+E_{250,254}\\
&+E_{251,255}+E_{253,257}+E_{254,258}+E_{306,311}+E_{307,312}+E_{310,314}+E_{311,315}\\
&+E_{313,317}+E_{314,318})\\
&+\bigl(-q^2\bigr)(E_{42,45}+E_{43,47}+E_{44,48}+E_{46,49}+E_{56,59}+E_{102,105}+E_{103,107}\\
&+E_{104,108}+E_{106,109}+E_{116,119}+E_{182,185}+E_{183,187}+E_{184,188}\\
&+E_{186,189}+E_{196,199})
\\
&+(-q)(E_{22,25}+E_{23,27}+E_{24,28}+E_{26,29}+E_{36,39}+E_{62,65}+E_{63,67}\\
&+E_{64,68}+E_{66,69}+E_{76,79}+E_{242,245}+E_{243,247}+E_{244,248}+E_{246,249}\\
&+E_{256,259}+E_{302,305}+E_{303,307}+E_{304,308}+E_{306,309}+E_{316,319})
\\
&+\bigl(-q - q^{-1}\bigr)\left(\sum_{i=101}^{140}E_{i,i+100}+\sum_{i=181}^{220}E_{i,i+80}+ \sum_{i=241}^{280}E_{i,i+80}+E_{6,11}+E_{7,12}\right.\\
&+E_{10,14}+E_{11,15}+E_{13,17}+E_{14,18}+E_{126,131}+E_{127,132}+E_{130,134}\\
&+E_{131,135}+E_{133,137}+E_{134,138}+E_{166,171}+E_{167,172}+E_{170,174}+E_{171,175}\\
&+E_{173,177}+E_{174,178}+E_{206,211}+E_{207,212}+E_{210,214}+E_{211,215}+E_{213,217}\\
&+E_{214,218}+E_{267,272}+E_{270,274}+E_{271,275}+E_{273,277}+E_{274,278}+E_{386,391}\\
&+E_{387,392}+E_{390,394}+E_{391,395}+E_{393,397}+E_{394,398})\\
&
+(-1)\left(\sum_{i=21}^{40} + \sum_{i=41} E_{i,i+80}+E_{i,i+60}^{80}+\sum_{i=101}^{120}E_{i,i+60}+\sum_{i=301}^{320}E_{i,i+60}+E_{2,5}\right. \\
&+E_{3,7}+E_{4,8}+E_{6,9}+E_{16,19}+E_{122,125}+E_{123,127}+E_{124,128}+E_{126,129}\\
&+E_{136,139}+E_{162,165}+E_{163,167}+E_{164,168}+E_{166,169}+E_{176,179}+E_{202,205}\\
&+E_{203,207}+E_{204,208}+E_{206,209}+E_{216,219}+E_{262,265}+E_{263,267}+E_{264,268}\\
&+E_{266,269}+E_{276,279}+E_{382,385}+E_{383,387}+E_{384,388}+E_{386,389}+E_{396,399})
\\
&+\bigl(-q^{-1}\bigr)(E_{82,85}+E_{83,87}+E_{84,88}+E_{86,89}+E_{96,99}+E_{142,145}\\
&+E_{143,147}+E_{144,148}+E_{146,149}+E_{156,159}+E_{322,325}+E_{323,327}+E_{324,328}\\
&+E_{326,329}+E_{336,339}+E_{362,365}+E_{363,367}+E_{364,368}+E_{366,369}+E_{376,379})
\\
&+\bigl(-q^{-2}\bigr)(E_{222,225}+E_{223,227}+E_{224,228}+E_{226,229}+E_{236,239}+E_{282,285}\\
&+E_{283,287}+E_{284,288}+E_{286,289}+E_{296,299}+E_{342,345}+E_{343,347}+E_{344,348}\\
&+E_{346,349}+E_{356,359})\\
&+\bigl(-1 - q^{-2}\bigr)(E_{86,91}+E_{87,92}+E_{90,94}+E_{91,95}+E_{93,97}+E_{94,98}+E_{146,151}\\
&+E_{147,152}+E_{150,154}+E_{151,155}+E_{153,157}+E_{154,158}+E_{326,331}+E_{327,332}\\
&+E_{330,334}+E_{331,335}+E_{333,337}+E_{334,338}+E_{366,371}+E_{367,372}+E_{370,374}\\
&+E_{371,375}+E_{373,377}+E_{374,378})\\
&+\bigl(-q^{-1} - q^{-3}\bigr)(E_{226,231}+E_{227,232}+E_{230,234}+E_{231,235}+E_{233,237}\\
&+E_{234,238}+E_{286,291}+E_{287,292}+E_{290,294}+E_{291,295}+E_{293,297}+E_{294,298}\\
&+E_{346,351}+E_{347,352}+E_{350,354}+E_{351,355}+E_{353,357}+E_{354,358}),
\\
 \pi_{W \otimes W}(F_1) ={}
&\bigl(q^3 + q\bigr)(E_{243,163}+E_{253,173}+E_{260,180}+E_{363,343}+E_{373,353}+E_{380,360})\\
&+\bigl(q^2 + 1\bigr)(E_{246,166}+E_{247,167}+E_{256,176}+E_{257,177}+E_{366,346}+E_{367,347}\\
&+E_{376,356}+E_{377,357})\\
&+\bigl(q^2\bigr)(E_{23,3}+E_{33,13}+E_{40,20}+E_{43,23}+E_{53,33}+E_{60,40}+E_{103,63}+E_{113,73}\\
&+E_{120,80}+E_{123,83}+E_{133,93}+E_{140,100}+E_{163,143}+E_{173,153}+E_{180,160}\\
&+E_{243,203}+E_{253,213}+E_{260,220}+E_{303,263}+E_{313,273}+E_{320,280}+E_{323,283}\\
&+E_{333,293}+E_{340,300}+E_{383,363}+E_{393,373}+E_{400,380})\\
&+(q)\left(\sum_{i=1}^{4} E_{10i+16,10i-4} + E_{10i+17,10i-3} + E_{106,66} + E_{107,67} + E_{116,76} \vphantom{\sum_{i=1}^4}\right) \\
&+\bigl(q + q^{-1}\bigr)\left(\sum_{i=1}^{20} E_{20i-7,20i-11} + E_{20i-1,20i-2} \right. \\
&+ \sum_{i=1}^2 E_{120i+122,180i-18} + E_{120i+129,180i-11} + E_{120i+130,180i-10} \\
&\left. + E_{120i+131,180i-9} + E_{120i+132,180i-8} + E_{120i+139,180i-1} \vphantom{\sum_{i=1}^{20}}\right) \\
&+(1)\left(\sum_{i=1}^{20} E_{20i-18,20i-19} + E_{20i-17,20i-18} + E_{20i-11,20i-12} +E_{20i,20i-1} \right. \\
&+ E_{20i-14,20i-16}+ E_{20i-13,20i-15} + E_{20i-7,20i-9}\\& + E_{20i-4,20i-6} + E_{20i-3,20i-5} \\
&+ \sum_{i \in \{1,2,8,19\}} E_{20i+2,20i-18} + E_{20i+9,20i-11} + E_{20i+10,20i-10} + E_{20i+11,20i-9}\\
& + E_{20i+12,20i-8} + E_{20i+19,20i-1} \\
&+ \sum_{i \in \{1,2,8,11,12\}} E_{20i+82,20i+42} + E_{20i+89,20i+49} + E_{20i+90,20i+50} \\
&\left. + E_{20i+91,20i+51}+ E_{20i+92,20i+52} + E_{20i+99,20i+59} \vphantom{\sum_{i=1}^{20}}\right) \\
&+\bigl(q^{-1}\bigr)\left(\sum_{i \in \{1,2,3,4,15,16,37,38\}} E_{10i+14,10i-6} + E_{10i+15,10i-5} \right. \\
&\left. + \sum_{i \in \{1,2,3,4,15,16,21,22,23,24\}} E_{10i+94,10i+54} + E_{10i+95,10i+55} \right) \\
&
+\bigl(q^{-2}\bigr)(E_{21,1}+E_{28,8}+E_{38,18}+E_{41,21}+E_{48,28}+E_{58,38}+E_{101,61}\\
&+E_{108,68}+E_{118,78}+E_{121,81}+E_{128,88}+E_{138,98}+E_{161,141}+E_{168,148}\\
&+E_{178,158}+E_{241,201}+E_{248,208}+E_{258,218}+E_{301,261}+E_{308,268}+E_{318,278}\\
&+E_{321,281}+E_{328,288}+E_{338,298}+E_{381,361}+E_{388,368}+E_{398,378})\\
&+\bigl(1 + q^{-2}\bigr)(E_{244,164}+E_{245,165}+E_{254,174}+E_{255,175}+E_{364,344}+E_{365,345}\\
&+E_{374,354}+E_{375,355})\\
&+\bigl(q^{-1} + q^{-3}\bigr)(E_{241,161}+E_{248,168}+E_{258,178}+E_{361,341}\\
&+E_{368,348}+E_{378,358}),
\\
 \pi_{W \otimes W}(F_2) ={}
&\bigl(q^3 + q\bigr)(E_{110,50}+E_{114,54}+E_{118,58})\\
&+\bigl(q^2 + 1\bigr)(E_{104,44}+E_{108,48}+E_{116,56}+E_{119,59})\\
&+\bigl(q^2\bigr)\left(\sum_{i \in \{1,16\}} E_{20i+50,20i+10} + E_{20i+54,20i+14} + E_{20i+58,20i+18} \right.\\
&+\sum_{i \in \{4,10,11,14\}} E_{20i+70,20i+10} + E_{20i+74,20i+14} + E_{20i+78,20i+18} \\
&\left.+\sum_{i \in \{5,6\}} E_{20i+90,20i+10} + E_{20i+94,20i+14} + E_{20i+98,20i+18} \right) \\
&+(q)\left(\sum_{i \in \{1,16\}} E_{20i+44,20i+4} + E_{20i+48,20i+8} + E_{20i+56,20i+16} \right.\\
&+ E_{20i+59,20i+19}+\sum_{i \in \{4,10,11,14\}} E_{20i+64,20i+4} + E_{20i+68,20i+8} \\
&+ E_{20i+76,20i+16} + E_{20i+79,20i+19}+ \sum_{i \in \{5,6\}} E_{20i+84,20i+4} + E_{20i+88,20i+8} \\
&\left. + E_{20i+96,20i+16}+ E_{20i+99,20i+19} \right)
\\
&
+\bigl(q+q^{-1}\bigr) \left( \sum_{i=1}^{20} E_{20i-14,20i-17} \vphantom{\sum_{j \in \{1,6,9,11,15,20\}}} + \sum_{j \in \{1,6,9,11,15,20\}} E_{j+100,j+40} \right)\\
&+(1)\left( \sum_{i=1}^{20} E_{20i-16,20i-18} + E_{20i-12,20i-15} + E_{20i-10,20i-14} + E_{20i-9,20i-13} \right. \\
&+ E_{20i-6,20i-9}+ E_{20i-5,20i-8} + E_{20i-2,20i-5} + E_{20i-1,20i-3}\\
&+\sum_{i \in \{1,16\}} E_{20i+41,20i+1} + E_{20i+46,20i+6} + E_{20i+49,20i+9} + E_{20i+51,20i+11} \\
&+ E_{20i+55,20i+15}+E_{20i+60,20i+20} \\
&+\sum_{i \in \{4,10,11,14\}} E_{20i+61,20i+1} + E_{20i+66,20i+6} + E_{20i+69,20i+9} + E_{20i+71,20i+11} \\
& + E_{20i+75,20i+15}+E_{20i+80,20i+20} \\
&+\sum_{i \in \{5,6\}} E_{20i+81,20i+1} + E_{20i+86,20i+6} + E_{20i+89,20i+9} + E_{20i+91,20i+11} \\
&\left.+ E_{20i+95,20i+15}+E_{20i+100,20i+20} \vphantom{\sum_{i=1}^{20}}\right) \\
&+\bigl(q^{-1}\bigr)\left( \sum_{i \in \{1,16\}} E_{20i+42,20i+2} + E_{20i+45,20i+5} + E_{20i+53,20i+13} \right.\\
&\left. + E_{20i+57,20i+17}+\sum_{i \in \{4,10,11,14\}} E_{20i+62,20i+2} + E_{20i+65,20i+5} \right.\\
&+ E_{20i+73,20i+13} + E_{20i+77,20i+17} +\sum_{i \in \{5,6\}} E_{20i+82,20i+2} + E_{20i+85,20i+5} \\
&\left.+ E_{20i+93,20i+13} + E_{20i+97,20i+17} \right) \\
&+\bigl(q^{-2}\bigr)\left( \sum_{i \in \{1,16\}} E_{20i+43,20i+3} + E_{20i+47,20i+7} + E_{20i+52,20i+12}\right.\\
&+\sum_{i \in \{4,10,11,14\}} E_{20i+63,20i+3} + E_{20i+67,20i+7} + E_{20i+72,20i+12} \\
&\left. +\sum_{i \in \{5,6\}} E_{20i+83,20i+3} + E_{20i+87,20i+7} + E_{20i+92,20i+12} \right) \\
&+\bigl(1 + q^{-2}\bigr)(E_{102,42}+E_{105,45}+E_{113,53}+E_{117,57})\\
&+\bigl(q^{-1} + q^{-3}\bigr)(E_{103,43}+E_{107,47}+E_{112,52})\\
&+\left(\frac{q^2}{q + q^{-1}}\right)(E_{270,170}+E_{274,174}+E_{278,178}+E_{310,250}+E_{314,254}\\
&+E_{318,258})+\left(\frac{q}{q + q^{-1}}\right)(E_{264,164}+E_{268,168}+E_{276,176}+E_{279,179}\\
&+E_{304,244}+E_{308,248}+E_{316,256}+E_{319,259})\\
&+\left(\frac{1}{q + q^{-1}}\right)\left( \sum_{i=1}^{20} E_{20i-6,20i-11} + E_{20i-4,20i-7} \vphantom{\sum_{j \in \{1,6,9,11,15,20\}}}\right.\\
&\left. + \sum_{j \in \{1,6,9,11,15,20\}} E_{j+300,j+240} + E_{j+260,j+160} \right) \\
&+\left(\frac{q^{-1}}{q + q^{-1}}\right)(E_{262,162}+E_{265,165}+E_{273,173}+E_{277,177}+E_{302,242}\\
&+E_{305,245}+E_{313,253}+E_{317,257})+\left(\frac{q^{-2}}{q + q^{-1}}\right)(E_{263,163}\\
&+E_{267,167}+E_{272,172}+E_{303,243}+E_{307,247}+E_{312,252}),
\\
 \pi_{W \otimes W}(F_3) ={}&
-\bigl(q^3 + q\bigr)(E_{132,52}+E_{135,55}+E_{138,58})\\
&-\bigl(q^2 + 1\bigr)(E_{125,45}+E_{128,48}+E_{137,57}+E_{139,59})\\
&-\bigl(q^2\bigr)\left( \sum_{i \in \{1,15\}} E_{20i+72,20i+12} + E_{20i+75,20i+15} + E_{20i+78,20i+18} \right.\\
&+ \sum_{i \in \{3,9,10,13\}} E_{20i+92,20i+12} + E_{20i+95,20i+15} + E_{20i+98,20i+18} \\
&\left.+ \sum_{i \in \{5,6\}} E_{20i+112,20i+12} + E_{20i+115,20i+15} + E_{20i+118,20i+18} \right)\\
&-(q)\left( \sum_{i \in \{1,15\}} E_{20i+65,20i+5} + E_{20i+68,20i+8} + E_{20i+77,20i+17} + E_{20i+79,20i+19} \right.\\
&+ \sum_{i \in \{3,9,10,13\}} E_{20i+85,20i+5} + E_{20i+88,20i+8} + E_{20i+97,20i+17} + E_{20i+99,20i+19} \\
&\left. + \sum_{i \in \{5,6\}} E_{20i+105,20i+5} + E_{20i+108,20i+8} + E_{20i+117,20i+17} + E_{20i+119,20i+19} \right)\\
&-\bigl(q+q^{-1}\bigr)\left( \sum_{i=1}^{20} E_{20i-13,20i-17} \vphantom{\sum_{j \in \{1,7,9,11,14,20\}}} + \sum_{j \in \{1,7,9,11,14,20\}} E_{j+120,j+40} \right)\\
&-(1)\left( \sum_{i=1}^{20} E_{20i-15,20i-18} + E_{20i-12,20i-16} + E_{20i-9,20i-14} + E_{20i-8,20i-13} \right.\\
&+ E_{20i-6,20i-10} + E_{20i-5,20i-9} + E_{20i-2,20i-6} + E_{20i-1,20i-4} \\
&+ \sum_{i \in \{1,15\}} E_{20i+61,20i+1} + E_{20i+67,20i+7} + E_{20i+69,20i+9} + E_{20i+71,20i+11} \\
&+ E_{20i+74,20i+14}+ E_{20i+80,20i+20} \\
&+ \sum_{i \in \{3,9,10,13\}} E_{20i+81,20i+1} + E_{20i+87,20i+7} + E_{20i+89,20i+9} + E_{20i+91,20i+11} \\
& + E_{20i+94,20i+14}+ E_{20i+80,20i+20} \\
&+ \sum_{i \in \{5,6\}} E_{20i+101,20i+1} + E_{20i+107,20i+7} + E_{20i+109,20i+9} + E_{20i+111,20i+11} \\
&\left.+ E_{20i+114,20i+14}+ E_{20i+120,20i+20} \vphantom{\sum_{i \in \{5,6\}}}\right)-\bigl(q^{-1}\bigr)\left( \sum_{i \in \{1,15\}} E_{20i+62,20i+2}\right. \\
&+ E_{20i+64,20i+4} + E_{20i+73,20i+13} + E_{20i+76,20i+16} \\
&+ \sum_{i \in \{3,9,10,13\}} E_{20i+82,20i+2} + E_{20i+84,20i+4} + E_{20i+93,20i+13} + E_{20i+96,20i+16} \\
&\left. + \sum_{i \in \{5,6\}} E_{20i+102,20i+2} + E_{20i+104,20i+4} + E_{20i+113,20i+13} + E_{20i+116,20i+16} \right)\\
&-\bigl(q^{-2}\bigr)\left( \sum_{i \in \{1,15\}} E_{20i+63,20i+3} + E_{20i+66,20i+6} + E_{20i+70,20i+10}\right.\\
&+ \sum_{i \in \{3,9,10,13\}} E_{20i+83,20i+3} + E_{20i+86,20i+6} + E_{20i+90,20i+10} \\
&\left. + \sum_{i \in \{1,15\}} E_{20i+103,20i+3} + E_{20i+106,20i+6} + E_{20i+110,20i+10} \right)\\
&-\bigl(1 + q^{-2}\bigr)(E_{122,42}+E_{124,44}+E_{133,53}+E_{136,56})\\
&-\bigl(q^{-1} + q^{-3}\bigr)(E_{123,43}+E_{126,46}+E_{130,50})\\
&-\left(\frac{q^2}{q + q^{-1}}\right)(E_{292,172}+E_{295,175}+E_{298,178}+E_{332,252}\\
&+E_{335,255}+E_{338,258})\\
&-\left(\frac{q}{q + q^{-1}}\right)(E_{285,165}+E_{288,168}+E_{297,177}+E_{299,179}+E_{325,245}+E_{328,248}\\
&+E_{337,257}+E_{339,259})\\
&-\left(\frac{1}{q + q^{-1}}\right)\left( \sum_{i=1}^{20} E_{20i-5,20i-11} + E_{20i-3,20i-7} \vphantom{\sum_{j \in \{1,7,9,11,14,20\}}}\right.\\
&\left. + \sum_{j \in \{1,7,9,11,14,20\}} E_{j+280,j+160} + E_{j+320,j+240} \right) \\
&-\left(\frac{q^{-1}}{q + q^{-1}}\right)(E_{282,162}+E_{284,164}+E_{293,173}+E_{296,176}+E_{322,242}\\
&+E_{324,244}+E_{333,253}+E_{336,256})\\
&-\left(\frac{q^{-2}}{q + q^{-1}}\right)(E_{283,163}+E_{286,166}+E_{290,170}+E_{323,243}\\
&+E_{326,246}+E_{330,250}),
\\
 \pi_{W \otimes W}\bigl(q^{H_1}\bigr) ={}&
\bigl(q^4\bigr)(E_{1,1}+E_{8,8}+E_{18,18}+E_{141,141}+E_{148,148}+E_{158,158}+E_{341,341}\\
&+E_{348,348}+E_{358,358})\\
&+\bigl(q^3\bigr)(E_{4,4}+E_{5,5}+E_{14,14}+E_{15,15}+E_{61,61}+E_{68,68}+E_{78,78}+E_{81,81}\\
&+E_{88,88}+E_{98,98}+E_{144,144}+E_{145,145}+E_{154,154}+E_{155,155}+E_{261,261}\\
&+E_{268,268}+E_{278,278}+E_{281,281}+E_{288,288}+E_{298,298}+E_{344,344}+E_{345,345}\\
&+E_{354,354}+E_{355,355})\\
&+\bigl(q^2\bigr)(E_{2,2}+E_{9,9}+E_{10,10}+E_{11,11}+E_{12,12}+E_{19,19}+E_{21,21}+E_{28,28}\\
&+E_{38,38}+E_{64,64}+E_{65,65}+E_{74,74}+E_{75,75}+E_{84,84}+E_{85,85}+E_{94,94}\\
&+E_{95,95}+E_{142,142}+E_{149,149}+E_{150,150}+E_{151,151}+E_{152,152}+E_{159,159}\\
&+E_{161,161}+E_{168,168}+E_{178,178}+E_{181,181}+E_{188,188}+E_{198,198}+E_{201,201}\\
&+E_{208,208}+E_{218,218}+E_{221,221}+E_{228,228}+E_{238,238}+E_{264,264}+E_{265,265}\\
&+E_{274,274}+E_{275,275}+E_{284,284}+E_{285,285}+E_{294,294}+E_{295,295}+E_{342,342}\\
&+E_{349,349}+E_{350,350}+E_{351,351}+E_{352,352}+E_{359,359}+E_{361,361}+E_{368,368}\\
&+E_{378,378})\\
&+(q)(E_{6,6}+E_{7,7}+E_{16,16}+E_{17,17}+E_{24,24}+E_{25,25}+E_{34,34}+E_{35,35}\\
&+E_{62,62}+E_{69,69}+E_{70,70}+E_{71,71}+E_{72,72}+E_{79,79}+E_{82,82}+E_{89,89}\\
&+E_{90,90}+E_{91,91}+E_{92,92}+E_{99,99}+E_{101,101}+E_{108,108}+E_{118,118}\\
&+E_{121,121}+E_{128,128}+E_{138,138}+E_{146,146}+E_{147,147}+E_{156,156}+E_{157,157}\\
&+E_{164,164}+E_{165,165}+E_{174,174}+E_{175,175}+E_{184,184}+E_{185,185}+E_{194,194}\\
&+E_{195,195}+E_{204,204}+E_{205,205}+E_{214,214}+E_{215,215}+E_{224,224}+E_{225,225}\\
&+E_{234,234}+E_{235,235}+E_{262,262}+E_{269,269}+E_{270,270}+E_{271,271}+E_{272,272}\\
&+E_{279,279}+E_{282,282}+E_{289,289}+E_{290,290}+E_{291,291}+E_{292,292}+E_{299,299}\\
&+E_{301,301}+E_{308,308}+E_{318,318}+E_{321,321}+E_{328,328}+E_{338,338}+E_{346,346}\\
&+E_{347,347}+E_{356,356}+E_{357,357}+E_{364,364}+E_{365,365}+E_{374,374}+E_{375,375})\\
&+(1)(E_{3,3}+E_{13,13}+E_{20,20}+E_{22,22}+E_{29,29}+E_{30,30}+E_{31,31}+E_{32,32}\\
&+E_{39,39}+E_{41,41}+E_{48,48}+E_{58,58}+E_{66,66}+E_{67,67}+E_{76,76}+E_{77,77}\\
&+E_{86,86}+E_{87,87}+E_{96,96}+E_{97,97}+E_{104,104}+E_{105,105}+E_{114,114}\\
&+E_{115,115}+E_{124,124}+E_{125,125}+E_{134,134}+E_{135,135}+E_{143,143}+E_{153,153}\\
&+E_{160,160}+E_{162,162}+E_{169,169}+E_{170,170}+E_{171,171}+E_{172,172}+E_{179,179}\\
&+E_{182,182}+E_{189,189}+E_{190,190}+E_{191,191}+E_{192,192}+E_{199,199}+E_{202,202}\\
&+E_{209,209}+E_{210,210}+E_{211,211}+E_{212,212}+E_{219,219}+E_{222,222}+E_{229,229}\\
&+E_{230,230}+E_{231,231}+E_{232,232}+E_{239,239}+E_{241,241}+E_{248,248}+E_{258,258}\\
&+E_{266,266}+E_{267,267}+E_{276,276}+E_{277,277}+E_{286,286}+E_{287,287}+E_{296,296}\\
&+E_{297,297}+E_{304,304}+E_{305,305}+E_{314,314}+E_{315,315}+E_{324,324}+E_{325,325}\\
&+E_{334,334}+E_{335,335}+E_{343,343}+E_{353,353}+E_{360,360}+E_{362,362}+E_{369,369}\\
&+E_{370,370}+E_{371,371}+E_{372,372}+E_{379,379}+E_{381,381}+E_{388,388}+E_{398,398})\\
&+\bigl(q^{-1}\bigr)(E_{26,26}+E_{27,27}+E_{36,36}+E_{37,37}+E_{44,44}+E_{45,45}+E_{54,54}\\
&+E_{55,55}+E_{63,63}+E_{73,73}+E_{80,80}+E_{83,83}+E_{93,93}+E_{100,100}\\
&+E_{102,102}+E_{109,109}+E_{110,110}+E_{111,111}+E_{112,112}+E_{119,119}+E_{122,122}\\
&+E_{129,129}+E_{130,130}+E_{131,131}+E_{132,132}+E_{139,139}+E_{166,166}+E_{167,167}\\
&+E_{176,176}+E_{177,177}+E_{186,186}+E_{187,187}+E_{196,196}+E_{197,197}+E_{206,206}\\
&+E_{207,207}+E_{216,216}+E_{217,217}+E_{226,226}+E_{227,227}+E_{236,236}+E_{237,237}\\
&+E_{244,244}+E_{245,245}+E_{254,254}+E_{255,255}+E_{263,263}+E_{273,273}+E_{280,280}\\
&+E_{283,283}+E_{293,293}+E_{300,300}+E_{302,302}+E_{309,309}+E_{310,310}+E_{311,311}\\
&+E_{312,312}+E_{319,319}+E_{322,322}+E_{329,329}+E_{330,330}+E_{331,331}+E_{332,332}\\
&+E_{339,339}+E_{366,366}+E_{367,367}+E_{376,376}+E_{377,377}+E_{384,384}+E_{385,385}\\
&+E_{394,394}+E_{395,395})\\
&+\bigl(q^{-2}\bigr)(E_{23,23}+E_{33,33}+E_{40,40}+E_{42,42}+E_{49,49}+E_{50,50}+E_{51,51}\\
&+E_{52,52}+E_{59,59}+E_{106,106}+E_{107,107}+E_{116,116}+E_{117,117}+E_{126,126}\\
&+E_{127,127}+E_{136,136}+E_{137,137}+E_{163,163}+E_{173,173}+E_{180,180}+E_{183,183}\\
&+E_{193,193}+E_{200,200}+E_{203,203}+E_{213,213}+E_{220,220}+E_{223,223}+E_{233,233}\\
&+E_{240,240}+E_{242,242}+E_{249,249}+E_{250,250}+E_{251,251}+E_{252,252}+E_{259,259}\\
&+E_{306,306}+E_{307,307}+E_{316,316}+E_{317,317}+E_{326,326}+E_{327,327}+E_{336,336}\\
&+E_{337,337}+E_{363,363}+E_{373,373}+E_{380,380}+E_{382,382}+E_{389,389}+E_{390,390}\\
&+E_{391,391}+E_{392,392}+E_{399,399})\\
&+\bigl(q^{-3}\bigr)(E_{46,46}+E_{47,47}+E_{56,56}+E_{57,57}+E_{103,103}+E_{113,113}+E_{120,120}\\
&+E_{123,123}+E_{133,133}+E_{140,140}+E_{246,246}+E_{247,247}+E_{256,256}+E_{257,257}\\
&+E_{303,303}+E_{313,313}+E_{320,320}+E_{323,323}+E_{333,333}+E_{340,340}+E_{386,386}\\
&+E_{387,387}+E_{396,396}+E_{397,397})\\
&+\bigl(q^{-4}\bigr)(E_{43,43}+E_{53,53}+E_{60,60}+E_{243,243}+E_{253,253}+E_{260,260}+E_{383,383}\\
&+E_{393,393}+E_{400,400}),
\\
 \pi_{W \otimes W}\bigl(q^{H_2}\bigr) ={}&
\bigl(q^4\bigr)(E_{43,43}+E_{47,47}+E_{52,52}+E_{123,123}+E_{127,127}+E_{132,132}+E_{223,223}\\
&+E_{227,227}+E_{232,232})\\
&+\bigl(q^3\bigr)(E_{23,23}+E_{27,27}+E_{32,32}+E_{42,42}+E_{45,45}+E_{53,53}+E_{57,57}\\
&+E_{83,83}+E_{87,87}+E_{92,92}+E_{122,122}+E_{125,125}+E_{133,133}+E_{137,137}\\
&+E_{222,222}+E_{225,225}+E_{233,233}+E_{237,237}+E_{243,243}+E_{247,247}+E_{252,252}\\
&+E_{323,323}+E_{327,327}+E_{332,332})\\
&+\bigl(q^2\bigr)(E_{3,3}+E_{7,7}+E_{12,12}+E_{22,22}+E_{25,25}+E_{33,33}+E_{37,37}+E_{41,41}\\
&+E_{46,46}+E_{49,49}+E_{51,51}+E_{55,55}+E_{60,60}+E_{82,82}+E_{85,85}+E_{93,93}\\
&+E_{97,97}+E_{103,103}+E_{107,107}+E_{112,112}+E_{121,121}+E_{126,126}+E_{129,129}\\
&+E_{131,131}+E_{135,135}+E_{140,140}+E_{163,163}+E_{167,167}+E_{172,172}+E_{203,203}\\
&+E_{207,207}+E_{212,212}+E_{221,221}+E_{226,226}+E_{229,229}+E_{231,231}+E_{235,235}\\
&+E_{240,240}+E_{242,242}+E_{245,245}+E_{253,253}+E_{257,257}+E_{283,283}+E_{287,287}\\
&+E_{292,292}+E_{322,322}+E_{325,325}+E_{333,333}+E_{337,337}+E_{383,383}+E_{387,387}\\
&+E_{392,392})\\
&+(q)(E_{2,2}+E_{5,5}+E_{13,13}+E_{17,17}+E_{21,21}+E_{26,26}+E_{29,29}+E_{31,31}\\
&+E_{35,35}+E_{40,40}+E_{44,44}+E_{48,48}+E_{56,56}+E_{59,59}+E_{63,63}+E_{67,67}\\
&+E_{72,72}+E_{81,81}+E_{86,86}+E_{89,89}+E_{91,91}+E_{95,95}+E_{100,100}+E_{102,102}\\
&+E_{105,105}+E_{113,113}+E_{117,117}+E_{124,124}+E_{128,128}+E_{136,136}+E_{139,139}\\
&+E_{143,143}+E_{147,147}+E_{152,152}+E_{162,162}+E_{165,165}+E_{173,173}+E_{177,177}\\
&+E_{202,202}+E_{205,205}+E_{213,213}+E_{217,217}+E_{224,224}+E_{228,228}+E_{236,236}\\
&+E_{239,239}+E_{241,241}+E_{246,246}+E_{249,249}+E_{251,251}+E_{255,255}+E_{260,260}\\
&+E_{282,282}+E_{285,285}+E_{293,293}+E_{297,297}+E_{303,303}+E_{307,307}+E_{312,312}\\
&+E_{321,321}+E_{326,326}+E_{329,329}+E_{331,331}+E_{335,335}+E_{340,340}+E_{363,363}\\
&+E_{367,367}+E_{372,372}+E_{382,382}+E_{385,385}+E_{393,393}+E_{397,397})\\
&+(1)(E_{1,1}+E_{6,6}+E_{9,9}+E_{11,11}+E_{15,15}+E_{20,20}+E_{24,24}+E_{28,28}\\
&+E_{36,36}+E_{39,39}+E_{50,50}+E_{54,54}+E_{58,58}+E_{62,62}+E_{65,65}+E_{73,73}\\
&+E_{77,77}+E_{84,84}+E_{88,88}+E_{96,96}+E_{99,99}+E_{101,101}+E_{106,106}\\
&+E_{109,109}+E_{111,111}+E_{115,115}+E_{120,120}+E_{130,130}+E_{134,134}+E_{138,138}\\
&+E_{142,142}+E_{145,145}+E_{153,153}+E_{157,157}+E_{161,161}+E_{166,166}+E_{169,169}\\
&+E_{171,171}+E_{175,175}+E_{180,180}+E_{183,183}+E_{187,187}+E_{192,192}+E_{201,201}\\
&+E_{206,206}+E_{209,209}+E_{211,211}+E_{215,215}+E_{220,220}+E_{230,230}+E_{234,234}\\
&+E_{238,238}+E_{244,244}+E_{248,248}+E_{256,256}+E_{259,259}+E_{263,263}+E_{267,267}\\
&+E_{272,272}+E_{281,281}+E_{286,286}+E_{289,289}+E_{291,291}+E_{295,295}+E_{300,300}\\
&+E_{302,302}+E_{305,305}+E_{313,313}+E_{317,317}+E_{324,324}+E_{328,328}+E_{336,336}\\
&+E_{339,339}+E_{343,343}+E_{347,347}+E_{352,352}+E_{362,362}+E_{365,365}+E_{373,373}\\
&+E_{377,377}+E_{381,381}+E_{386,386}+E_{389,389}+E_{391,391}+E_{395,395}+E_{400,400})\\
&+\bigl(q^{-1}\bigr)(E_{4,4}+E_{8,8}+E_{16,16}+E_{19,19}+E_{30,30}+E_{34,34}+E_{38,38}+E_{61,61}\\
&+E_{66,66}+E_{69,69}+E_{71,71}+E_{75,75}+E_{80,80}+E_{90,90}+E_{94,94}+E_{98,98}\\
&+E_{104,104}+E_{108,108}+E_{116,116}+E_{119,119}+E_{141,141}+E_{146,146}+E_{149,149}\\
&+E_{151,151}+E_{155,155}+E_{160,160}+E_{164,164}+E_{168,168}+E_{176,176}+E_{179,179}\\
&+E_{182,182}+E_{185,185}+E_{193,193}+E_{197,197}+E_{204,204}+E_{208,208}+E_{216,216}\\
&+E_{219,219}+E_{250,250}+E_{254,254}+E_{258,258}+E_{262,262}+E_{265,265}+E_{273,273}\\
&+E_{277,277}+E_{284,284}+E_{288,288}+E_{296,296}+E_{299,299}+E_{301,301}+E_{306,306}\\
&+E_{309,309}+E_{311,311}+E_{315,315}+E_{320,320}+E_{330,330}+E_{334,334}+E_{338,338}\\
&+E_{342,342}+E_{345,345}+E_{353,353}+E_{357,357}+E_{361,361}+E_{366,366}+E_{369,369}\\
&+E_{371,371}+E_{375,375}+E_{380,380}+E_{384,384}+E_{388,388}
+E_{396,396}+E_{399,399})\\
&+\bigl(q^{-2}\bigr)(E_{10,10}+E_{14,14}+E_{18,18}+E_{64,64}+E_{68,68}+E_{76,76}+E_{79,79}\\
&+E_{110,110}+E_{114,114}+E_{118,118}+E_{144,144}+E_{148,148}+E_{156,156}+E_{159,159}\\
&+E_{170,170}+E_{174,174}+E_{178,178}+E_{181,181}+E_{186,186}+E_{189,189}+E_{191,191}\\
&+E_{195,195}+E_{200,200}+E_{210,210}+E_{214,214}+E_{218,218}+E_{261,261}+E_{266,266}\\
&+E_{269,269}+E_{271,271}+E_{275,275}+E_{280,280}+E_{290,290}+E_{294,294}+E_{298,298}\\
&+E_{304,304}+E_{308,308}+E_{316,316}+E_{319,319}+E_{341,341}+E_{346,346}+E_{349,349}\\
&+E_{351,351}+E_{355,355}+E_{360,360}+E_{364,364}+E_{368,368}+E_{376,376}+E_{379,379}\\
&+E_{390,390}+E_{394,394}+E_{398,398})\\
&+\bigl(q^{-3}\bigr)(E_{70,70}+E_{74,74}+E_{78,78}+E_{150,150}+E_{154,154}+E_{158,158}\\
&+E_{184,184}+E_{188,188}+E_{196,196}+E_{199,199}+E_{264,264}+E_{268,268}+E_{276,276}\\
&+E_{279,279}+E_{310,310}+E_{314,314}+E_{318,318}+E_{344,344}+E_{348,348}+E_{356,356}\\
&+E_{359,359}+E_{370,370}+E_{374,374}+E_{378,378})\\
&+\bigl(q^{-4}\bigr)(E_{190,190}+E_{194,194}+E_{198,198}+E_{270,270}+E_{274,274}+E_{278,278}\\
&+E_{350,350}+E_{354,354}+E_{358,358}),
\\
 \pi_{W \otimes W}\bigl(q^{H_3}\bigr) ={}&
\bigl(q^4\bigr)(E_{43,43}+E_{46,46}+E_{50,50}+E_{103,103}+E_{106,106}+E_{110,110}+E_{183,183}\\
&+E_{186,186}+E_{190,190})\\
&+\bigl(q^3\bigr)(E_{23,23}+E_{26,26}+E_{30,30}+E_{42,42}+E_{44,44}+E_{53,53}+E_{56,56}\\
&+E_{63,63}+E_{66,66}+E_{70,70}+E_{102,102}+E_{104,104}+E_{113,113}+E_{116,116}\\
&+E_{182,182}+E_{184,184}+E_{193,193}+E_{196,196}+E_{243,243}+E_{246,246}+E_{250,250}\\
&+E_{303,303}+E_{306,306}+E_{310,310})\\
&+\bigl(q^2\bigr)(E_{3,3}+E_{6,6}+E_{10,10}+E_{22,22}+E_{24,24}+E_{33,33}+E_{36,36}+E_{41,41}\\
&+E_{47,47}+E_{49,49}+E_{51,51}+E_{54,54}+E_{60,60}+E_{62,62}+E_{64,64}+E_{73,73}\\
&+E_{76,76}+E_{101,101}+E_{107,107}+E_{109,109}+E_{111,111}+E_{114,114}+E_{120,120}\\
&+E_{123,123}+E_{126,126}+E_{130,130}+E_{163,163}+E_{166,166}+E_{170,170}+E_{181,181}\\
&+E_{187,187}+E_{189,189}+E_{191,191}+E_{194,194}+E_{200,200}+E_{203,203}+E_{206,206}\\
&+E_{210,210}+E_{242,242}+E_{244,244}+E_{253,253}+E_{256,256}+E_{263,263}+E_{266,266}\\
&+E_{270,270}+E_{302,302}+E_{304,304}+E_{313,313}+E_{316,316}+E_{383,383}+E_{386,386}\\
&+E_{390,390})\\
&+(q)(E_{2,2}+E_{4,4}+E_{13,13}+E_{16,16}+E_{21,21}+E_{27,27}+E_{29,29}+E_{31,31}\\
&+E_{34,34}+E_{40,40}+E_{45,45}+E_{48,48}+E_{57,57}+E_{59,59}+E_{61,61}\\
&+E_{67,67}+E_{69,69}+E_{71,71}+E_{74,74}+E_{80,80}+E_{83,83}\\
&+E_{86,86}+E_{90,90}+E_{105,105}+E_{108,108}+E_{117,117}+E_{119,119}+E_{122,122}\\
&+E_{124,124}+E_{133,133}+E_{136,136}+E_{143,143}+E_{146,146}+E_{150,150}+E_{162,162}\\
&+E_{164,164}+E_{173,173}+E_{176,176}+E_{185,185}+E_{188,188}+E_{197,197}+E_{199,199}\\
&+E_{202,202}+E_{204,204}+E_{213,213}+E_{216,216}+E_{241,241}+E_{247,247}+E_{249,249}\\
&+E_{251,251}+E_{254,254}+E_{260,260}+E_{262,262}+E_{264,264}+E_{273,273}+E_{276,276}\\
&+E_{301,301}+E_{307,307}+E_{309,309}+E_{311,311}+E_{314,314}+E_{320,320}+E_{323,323}\\
&+E_{326,326}+E_{330,330}+E_{363,363}+E_{366,366}+E_{370,370}+E_{382,382}+E_{384,384}\\
&+E_{393,393}+E_{396,396})\\
&+(1)(E_{1,1}+E_{7,7}+E_{9,9}+E_{11,11}+E_{14,14}+E_{20,20}+E_{25,25}+E_{28,28}\\
&+E_{37,37}+E_{39,39}+E_{52,52}+E_{55,55}+E_{58,58}+E_{65,65}+E_{68,68}+E_{77,77}\\
&+E_{79,79}+E_{82,82}+E_{84,84}+E_{93,93}+E_{96,96}+E_{112,112}+E_{115,115}\\
&+E_{118,118}+E_{121,121}+E_{127,127}+E_{129,129}+E_{131,131}+E_{134,134}+E_{140,140}\\
&+E_{142,142}+E_{144,144}+E_{153,153}+E_{156,156}+E_{161,161}+E_{167,167}+E_{169,169}\\
&+E_{171,171}+E_{174,174}+E_{180,180}+E_{192,192}+E_{195,195}+E_{198,198}+E_{201,201}\\
&+E_{207,207}+E_{209,209}+E_{211,211}+E_{214,214}+E_{220,220}+E_{223,223}+E_{226,226}\\
&+E_{230,230}+E_{245,245}+E_{248,248}+E_{257,257}+E_{259,259}+E_{261,261}+E_{267,267}\\
&+E_{269,269}+E_{271,271}+E_{274,274}+E_{280,280}+E_{283,283}+E_{286,286}+E_{290,290}\\
&+E_{305,305}+E_{308,308}+E_{317,317}+E_{319,319}+E_{322,322}+E_{324,324}+E_{333,333}\\
&+E_{336,336}+E_{343,343}+E_{346,346}+E_{350,350}+E_{362,362}+E_{364,364}+E_{373,373}\\
&+E_{376,376}+E_{381,381}+E_{387,387}+E_{389,389}+E_{391,391}+E_{394,394}+E_{400,400})\\
&+\bigl(q^{-1}\bigr)(E_{5,5}+E_{8,8}+E_{17,17}+E_{19,19}+E_{32,32}+E_{35,35}+E_{38,38}+E_{72,72}\\
&+E_{75,75}+E_{78,78}+E_{81,81}+E_{87,87}+E_{89,89}+E_{91,91}+E_{94,94}+E_{100,100}\\
&+E_{125,125}+E_{128,128}+E_{137,137}+E_{139,139}+E_{141,141}+E_{147,147}+E_{149,149}\\
&+E_{151,151}+E_{154,154}+E_{160,160}+E_{165,165}+E_{168,168}+E_{177,177}+E_{179,179}\\
&+E_{205,205}+E_{208,208}+E_{217,217}+E_{219,219}+E_{222,222}+E_{224,224}+E_{233,233}\\
&+E_{236,236}+E_{252,252}+E_{255,255}+E_{258,258}+E_{265,265}+E_{268,268}+E_{277,277}\\
&+E_{279,279}+E_{282,282}+E_{284,284}+E_{293,293}+E_{296,296}+E_{312,312}+E_{315,315}\\
&+E_{318,318}+E_{321,321}+E_{327,327}+E_{329,329}+E_{331,331}+E_{334,334}+E_{340,340}\\
&+E_{342,342}+E_{344,344}+E_{353,353}+E_{356,356}+E_{361,361}+E_{367,367}+E_{369,369}\\
&+E_{371,371}+E_{374,374}+E_{380,380}+E_{385,385}+E_{388,388}+E_{397,397}+E_{399,399})\\
&+\bigl(q^{-2}\bigr)(E_{12,12}+E_{15,15}+E_{18,18}+E_{85,85}+E_{88,88}+E_{97,97}+E_{99,99}\\
&+E_{132,132}+E_{135,135}+E_{138,138}+E_{145,145}+E_{148,148}+E_{157,157}+E_{159,159}\\
&+E_{172,172}+E_{175,175}+E_{178,178}+E_{212,212}+E_{215,215}+E_{218,218}+E_{221,221}\\
&+E_{227,227}+E_{229,229}+E_{231,231}+E_{234,234}+E_{240,240}+E_{272,272}+E_{275,275}\\
&+E_{278,278}+E_{281,281}+E_{287,287}+E_{289,289}+E_{291,291}+E_{294,294}+E_{300,300}\\
&+E_{325,325}+E_{328,328}+E_{337,337}+E_{339,339}+E_{341,341}+E_{347,347}+E_{349,349}\\
&+E_{351,351}+E_{354,354}+E_{360,360}+E_{365,365}+E_{368,368}+E_{377,377}+E_{379,379}\\
&+E_{392,392}+E_{395,395}+E_{398,398})\\
&+\bigl(q^{-3}\bigr)(E_{92,92}+E_{95,95}+E_{98,98}+E_{152,152}+E_{155,155}+E_{158,158}+E_{225,225}\\
&+E_{228,228}+E_{237,237}+E_{239,239}+E_{285,285}+E_{288,288}+E_{297,297}+E_{299,299}\\
&+E_{332,332}+E_{335,335}+E_{338,338}+E_{345,345}+E_{348,348}+E_{357,357}+E_{359,359}\\
&+E_{372,372}+E_{375,375}+E_{378,378})\\
&+\bigl(q^{-4}\bigr)(E_{232,232}+E_{235,235}+E_{238,238}+E_{292,292}+E_{295,295}+E_{298,298}\\
&+E_{352,352}+E_{355,355}+E_{358,358}).
\end{align*}

\subsection[Listing 4 times 4 blocks of pi\_\{W otimes W\}(C)]{Listing $\boldsymbol{4 \times 4}$ blocks of $\boldsymbol{\pi_{W \otimes W}(C)}$}
\label{4x4blocks}

The twenty-four $4 \times 4$ blocks of $\pi_{W \otimes W}(C)$ take the following forms: ten blocks take the form
\begin{align*}
 \begin{bmatrix}
B_1 & B_2 & B_3 & 0\\
B_4 & B_5 & B_3 & B_6\\
B_7 & B_3 & B_8 & B_9\\
0 & B_{10} & B_{11} & B_{12}
 \end{bmatrix},
\end{align*}
four blocks take the form
\begin{align*}
 \begin{bmatrix}
B_1 & B_7 & B_6 & 0\\
B_{13} & B_5 & B_3 & B_3\\
B_2 & B_3 & B_8 & B_{10}\\
0 & B_9 & B_{14} & B_{12}
 \end{bmatrix},
\end{align*}
four blocks take the form
\begin{align*}
 \begin{bmatrix}
B_{15} & B_7 & B_7 & B_3\\
B_3 & B_{16} & 0 & B_{10}\\
B_3 & 0 & B_{16} & B_{10}\\
B_3 & B_6 & B_6 & B_{17}
 \end{bmatrix},
\end{align*}
four blocks take the form
\begin{align*}
 \begin{bmatrix}
B_{15} & B_2 & B_2 & B_3\\
B_6 & B_{16} & 0 & B_9\\
B_6 & 0 & B_{16} & B_9\\
B_3 & B_3 & B_3 & B_{17}
 \end{bmatrix},
\end{align*}
and two blocks take the form
\begin{align*}
 \begin{bmatrix}
B_1 & B_{18} & B_{19} & 0\\
B_{20} & B_5 & B_3 & B_{21}\\
B_{22} & B_3 & B_8 & B_{23}\\
0 & B_{24} & B_{25} & B_{12}
 \end{bmatrix},
\end{align*}
where\vspace{-1mm}
\begin{gather*}
 B_1 = q^{12}-q^{10}+2q^{6}+2+q^{-2}-q^{-4}+2q^{-8}, \\[-1pt]
 B_2 = q^{10}-2q^{8}+q^{6}+q^{-4}-2q^{-6}+q^{-8}, \\[-1pt]
 B_3 = q^{9}-2q^{7}+q^{5}+q^{-5}-2q^{-7}+q^{-9}, \\[-1pt]
 B_4 = q^{12}-q^{10}-q^{8}+q^{6}+q^{-2}-q^{-4}-q^{-6}+q^{-8}, \\[-1pt]
 B_5 = 2q^{10}-2q^{8}+q^{6}+q^{4}+2+2q^{-4}-2q^{-6}+q^{-8}+q^{-10}, \\[-1pt]
 B_6 = q^{10}-q^{8}-q^{6}+q^{4}+q^{-4}-q^{-6}-q^{-8}+q^{-10}, \\[-1pt]
 B_7 = q^{11}-q^{9}-q^{7}+q^{5}+q^{-3}-q^{-5}-q^{-7}+q^{-9},\\[-1pt]
 B_8 = q^{10}+q^{8}-2q^{6}+2q^{4}+2+q^{-4}+q^{-6}-2q^{-8}+2q^{-10}, \\[-1pt]
 B_9 = q^{9}-q^{7}-q^{5}+q^{3}+q^{-5}-q^{-7}-q^{-9}+q^{-11}, \\[-1pt]
 B_{10} = q^{8}-2q^{6}+q^{4}+q^{-6}-2q^{-8}+q^{-10}, \\[-1pt]
 B_{11} = q^{7}-2q^{5}+q^{3}+q^{-7}-2q^{-9}+q^{-11}, \\[-1pt]
 B_{12} = 2q^{8}-q^{4}+q^{2}+2+2q^{-6}-q^{-10}+q^{-12}, \\[-1pt]
 B_{13} = q^{11}-2q^{9}+q^{7}+q^{-3}-2q^{-5}+q^{-7}, \\[-1pt]
 B_{14} = q^{8}-q^{6}-q^{4}+q^{2}+q^{-6}-q^{-8}-q^{-10}+q^{-12}, \\[-1pt]
 B_{15} = q^{12}-2q^{8}+3q^{6}+2+q^{-2}-2q^{-6}+3q^{-8},\\[-1pt]
 B_{16} = q^{10}+q^{4}+2+q^{-4}+q^{-10}, \\[-1pt]
 B_{17} = 3q^{8}-2q^{6}+q^{2}+2+3q^{-6}-2q^{-8}+q^{-12}, \\[-1pt]
 B_{18} = q^{9}-3q^{7}+4q^{5}-4q^{3}+4q-4q^{-4/q^{-}}-3q^{-5}+q^{-7}, \\[-1pt]
 B_{19} = q^{8}-3q^{6}+4q^{4}-4q^{2}+4-4q^{-2}+4q^{-4}-3q^{-6}+q^{-8}, \\[-1pt]
 B_{20} = q^{13}-2q^{9}+q^{5}+q^{-2/q^{-}}+q^{-9}, \\[-1pt]
 B_{21} = q^{11}-2q^{7}+q^{3}+q^{-3}-2q^{-7}+q^{-11}, \\[-1pt]
 B_{22} = q^{12}-2q^{8}+q^{4}+q^{-2}-2q^{-6}+q^{-10}, \\[-1pt]
 B_{23} = q^{10}-2q^{6}+q^{2}+q^{-4}-2q^{-8}+q^{-12}, \\[-1pt]
 B_{24} = q^{7}-3q^{5}+4q^{3}-4q+4q^{-4/q^{-}}+4q^{-5}-3q^{-7}+q^{-9}, \\[-1pt]
 B_{25} = q^{6}-3q^{4}+4q^{2}-4+4q^{-2}-4q^{-4}+4q^{-6}-3q^{-8}+q^{-10}.
\end{gather*}

\vspace{-1mm}

\subsection*{Accessibility statement}
  A~WCAG2.1AA compliant version of this PDF will be posted at \href{https://u.osu.edu/jkuanpapers/2026/01/23/reu_2024_paper/}{our advisor's webpage.} We thank the National Science
Foundation's Office for Strategic Initiatives on `STEM Access for Persons with Disabilities' for funding the creation of the accessible version.

\subsection*{Acknowledgments}
This research was conducted during a~mathematics REU at Texas A\&M University, funded by the National Science Foundation under REU Site Grant DMS-2150094 and the National Science Foundation's Office for Strategic Initiatives on `STEM Access for Persons with Disabilities'. This research was also supported by Princeton University's Office of Undergraduate Research Undergraduate Fund for Academic Conferences through the Hewlett Fund. The authors express their gratitude to their mentor, Dr. Jeffrey Kuan, and teaching assistant, Kennedi Hatcher, for their involvement in the project. In particular, the authors thank Dr. Kuan for his assistance with Proposition~\ref{generalizeStationary} and for his guidance on how to extend the Markov dual in Section~\ref{subsection--MarkovDuality}. The authors also thank the anonymous referees for their helpful feedback, which significantly improved the contents of this paper.

\pdfbookmark[1]{References}{ref}
\LastPageEnding

\end{document}